\documentclass[12pt,reqno] {amsart}%
\pdfoutput=1
\usepackage{amsfonts}
\usepackage{amsmath}
\usepackage{amssymb}
\usepackage{xy}
\usepackage{color}
\usepackage{graphicx}%
\usepackage{hyperref}
\usepackage{cleveref}
\usepackage{comment}
 \usepackage{mathrsfs}
 \usepackage{float}
 \usepackage{tikz}
 \usepackage{bm}
 \usepackage{ulem}
 \usepackage{bbm}
\usepackage{slashed}
 \usepackage{braket}
 \usepackage{makecell}

\theoremstyle{plain}
\newtheorem{acknowledgement}{Acknowledgement}

\newtheorem{question}{Question}

\newtheorem{theorem}{Theorem}[section]
\newtheorem{corollary}[theorem]{Corollary}
\newtheorem{lemma}[theorem]{Lemma}
\newtheorem{proposition}[theorem]{Proposition}
\newtheorem{remark}[theorem]{Remark}
\newtheorem{definition}[theorem]{Definition}

\newtheorem{notation}[theorem]{Notation}
\newtheorem{conjecture}{Conjecture}

\newtheorem{main theorem}{Main Theorem}

\newcommand{\BB}{{\partial}} 
\newcommand{\TB}{\partial_T} 
\newcommand{\STB}{\partial_{ST}} 

\newcommand{\PL}{~}

\newcommand{\PP}{P} 

\newcommand{\LS}{LS} 

\newcommand{\stopn}{\mathscr{O}^n} 

\title{Functional Integral Construction of Topological Quantum Field Theory}

\author{Zhengwei Liu}
\address{Z. Liu, Yau Mathematical Sciences Center and Department of Mathematics, Tsinghua University, Beijing, 100084, China}
\address{Beijing Institute of Mathematical Sciences and Applications, Huairou District, Beijing, 101408, China}
\email{liuzhengwei@mail.tsinghua.edu.cn}

\begin{document}

\begin{abstract}
We introduce regular stratified piecewise linear manifolds to describe lattices and investigate the lattice model approach to topological quantum field theory in all dimensions. We introduce the unitary $n+1$ alterfold TQFT and construct it from a linear functional on an $n$-dimensional lattice model on an $n$-sphere satisfying three conditions: reflection positivity, homeomorphic invariance and complete finiteness. A unitary spherical $n$-category is mathematically defined and emerges as the local quantum symmetry of the lattice model. The alterfold construction unifies various constructions of $n+1$ TQFT from $n$-dimensional lattice models and $n$-categories. 
In particular, we construct a non-invertible unitary 3+1 alterfold TQFT from a linear functional and derive its local quantum symmetry as a unitary spherical 3-category of Ising type with explicit 20j-symbols, so that the scalar invariant of 2-knots in piecewise linear 4-manifolds could be computed explicitly.     
\end{abstract}

\maketitle

\section{Introduction}
In this paper we propose a new program to address a long-standing open problem: how to construct a meaningful unitary topological quantum field theory (TQFT) in arbitrary dimension. We believe that the ideas here are relevant both to the development of category theory as well as to the mathematical foundations of physics. We will return to this theme in future works concentrating both on the theory and on examples. 

Here we construct a unitary $n+1$ alterfold TQFT from an functional integral on $n$-dimensional lattice models, for arbitrary $n$, in Theorem \ref{Thm: Alterfold TQFT}. The $n+1$ alterfold TQFT has alternating $A/B$-colored $n+1$ manifolds with the lattice model on the $n$-manifold between them. Our example reduces to Atiyah's TQFT \cite{Ati88} when the boundary is a $B$-color $n$-manifold. This present work builds on our recent work on alterfolds \cite{LMWW23a,LMWW23b}, and on the other works we now explain. 

Witten constructed a 2+1 TQFT using Chern-Simons theory and obtained an invariant of links in 3-manifolds as a path integral  \cite{Wit89}, generalizing the Jones polynomial originated from subfactor theory \cite{Jon83,Jon85,Jon87}, and other link invariants from the representation theory of Drinfeld-Jimbo quantum groups \cite{Jim85,Dri86,HOMFLY85,PT88,Kau90}. 
Feynman's path integral is a powerful method in physics, but the measure of the path space is only mathematically defined for a few cases \cite{GliJaf87}.
Atiyah provided a mathematical axiomatization of TQFT to study its topological invariants \cite{Ati88} which avoids the path integral.
The topological invariant of Witten's 2+1 TQFT can be rigorously defined using the link invariants from quantum groups and the surgery theory, known as the Witten-Reshetikhin-Turaev TQFT \cite{ResTur91}. 
The Turaev-Viro-Barrett-Westbury 2+1 TQFT from a spherical fusion category \cite{TurVir92,BarWes96} is a state sum construction over a triangulation, which is a combinatorial analogue of path integral.

Witten also constructed a 3+1 TQFT \cite{Wit88} which captures Donaldson's invariant of smooth 4-manifolds \cite{Don83,Don83b,Don90}. By this invariant, Donaldson constructed exotic four-dimensional spaces together with the seminal work of Freedman \cite{Fre82}. 
The Turaev-Viro state sum construction has been generalized to construct 3+1 TQFT using a braided fusion category in \cite{CraYet93,CKY97,Cui19} or a fusion 2-category in \cite{DogReu18}.

In higher dimensions, Dijkgraaf-Witten constructed an $n$-dimensional TQFT from the group cohomology of a finite group \cite{DijWit90}. 
Lurie introduced a fruitful theory of $(\infty,n)$ category in \cite{Lur09} to study non-invertible higher symmetries and to answer the cobordism hypothesis of Baez and Dolan \cite{BaeDol95}.  It is widely believed that the state sum construction of TQFT from spherical categories will work in any dimension e.g. \cite{KWZ15,Wal21}, but there is no agreement on the mathematical definition of a (unitary) spherical $n$-category; see a recent discussion in \cite{FHJ24}.

The 2+1 TQFT is exceptionally successful due to the fruitful examples of quantum symmetries coming from the representation theory of quantum groups, subfactors, vertex operator algebras, conformal field theory, etc.
It is highly expected, but remains challenging, to generalize those frameworks to higher dimensions, which should provide examples of non-invertible higher quantum symmetries, such spherical $n$-categories, from their higher representation theory. 

In this paper, we provide a functional integral approach to construct an $n+1$ TQFT using a linear functional $Z$ on an $n$-dimensional lattice model for any dimension $n$. 
This functional integral point of view has been well established in constructive quantum field theory (QFT) \cite{GliJaf87}; we bring that method to the study of TQFT. 

Mathematically, we introduce labelled regular stratified piece-wise linear manifolds to formulate the lattices of a general shape and their configuration space.
We study their basic properties in \S\ref{Sec: Labelled Regular Stratified Manifolds} and prove the transversal theorem between stratified manifolds and triangulations in Theorem \ref{Thm: transversality}. Based on it, we can compute the partition function in $n+1$ TQFT as a state sum according to a transversal triangulation.

In an $n$-dimensional lattice model, the configuration space of a lattice is a Hilbert space, which is the tensor product of Hilbert spaces of local spins, regarded as labels of stratified manifolds. 
We construct an $n+1$ TQFT from a linear functional $Z$ on the configuration spaces of lattices on the $n$-sphere $S^n$ in Theorem \ref{Thm: Alterfold TQFT}, such that the $(n+1)$-manifolds of the topological quantum field theory have $A/B$ colors and the $n$-dimensional hyper surfaces (or domain walls) between them are decorated by $n$-dimensional lattice models. We call such TQFT an $n+1$ alterfold TQFT (Def.~\ref{Def: alterfold TQFT}), generalizing the 2+1 alterfold TQFT in \cite{LMWW23a,LMWW23b}. 

To achieve the construction of the TQFT, we assume three conditions of the linear functional $Z$ in \S \ref{Sec: Hyper-Sphere Functions},
\begin{enumerate}
    \item (RP) reflection positivity;
    \item (HI) homeomorphic invariance;
    \item (CF) complete finiteness.
\end{enumerate}

Condition (RP) means that the inner product induced from $Z$ is positive semi-definite between configuration spaces on half $n$-discs of $S^n$ for any fixed common boundary on the equator $S^{n-1}$.  
The functional integral $Z$ satisfying Reflection Positivity (RP) is a mathematical formulation of the measure on the path space, by the Wick rotation from statistical physics to quantum field theory (QFT) \cite{OstSch73}. 
To formulate the theory over a general field $\mathbb{K}$, we need to replace the condition (RP) by strong semisimplicity.
 
Condition (HI) means that the linear functional $Z$ is invariant under homeomorphisms on $S^n$.
Condition (HI) ensures the QFT to be topological, namely the partition function of the TQFT, still denoted by $Z$, is homeomorphic invariant.

Condition (CF) means that the inner product induced from $Z$ has finite rank between configuration spaces of lattices on $D^k\times S^{n-k}$ for any fixed boundary on $D^{k-1}\times S^{n-k}$. Condition (CF) ensures the partition function of the TQFT to be a finite state sum. 

From the physics point of view, we know everything, if we know the linear functional $Z$. 
We implement this idea mathematically as the Null Principle that if we cannot distinguish two vectors from $Z$, then we consider them to be the same mathematically. 
The null principle encodes numerous algebraic relations from the kernel of $Z$. 
We highlight our identification of vectors on different lattices with the same boundary using the null principle of the configuration spaces, which is different from the renormalization group methods in condensed matter physics. 

In \S \ref{Sec: Hyper Disc Algebras and Sphere Algebras}, we introduce hyper discs (or $D^n$) algebras to study local algebraic relations, generalizing the notion of planar algebras of Jones \cite{Jon21} for $n=2$. It captures the action of the operad of regular stratified manifolds on the configurations on local $n$-discs.

In \S \ref{Sec: Idempotent Completion and Spherical n-category}, we study the higher representation theory of the $D^n$ algebra which form an $n$-category. It suggests a mathematical definition of a unitary/spherical $n$-category, together with the operad action and the linear functional $Z$ with the condition (RP)/(HI).

In \S \ref{Sec: Alterfold TQFT}, we derive the skein relations for the $n+1$ alterfold TQFT in Def.~\ref{Def: S-relation}, using the null principle. The relations can be considered algebraically as a resolution of identity and topologically as removing a $D^k\times \varepsilon D^{n-k+1}$ handle with a bistellar move on its boundary. 
We prove the consistency of the relations and construct an $n+1$ alterfold TQFT in Theorem \ref{Thm: space-time TQFT} and Theorem \ref{Thm: Alterfold TQFT}.
In \S \ref{Sec: Examples}, we give some examples to understand better the concepts in the alterfold TQFT.
If we shrink the $B$-color manifolds and evaluate the partition function according to the triangulation, then we obtain a higher analogue of the Turaev-Viro TQFT of the spherical $n$-category. 
If we shrink the $A$-color manifolds and evaluate the partition function according to surgery move, then we obtain a higher analogue of the Reshetikhin-Turaev TQFT of the higher Drinfeld center of the spherical $n$-category. The later one suggests a fruitful higher braid statistics of membranes, as discussed in \S \ref{Sec: Higher Braid Statistics}.

In \S \ref{Sec: 3+1 Ising}, we give a concrete example to illustrate our functional integral construction of TQFT. We construct a linear function $Z$ for a lattice model on $S^3$. We prove the three conditions (RP), (HI), (CF) and therefore we obtain a unitary 3-category of Ising type and a non-invertible unitary 3+1 alterfold TQFT. We compute all its simplicial morphisms and the 20j-symbols in Table \ref{tab:20j-symbols}. Using the skein relations and the 20j-symbols, the scalar invariant 2-knots in 4-manifolds could be computed explicitly.     
The 20j-symbols satisfies the (3,3), (2,4) and (1,5) Pachner moves \cite{Pan91}, as a one-dimensional higher analogue of the pentagon equations. In this example, there are already more than 50,000 equations, which seems too many to solve directly in general. We conjecture that the construction of the 3+1 TQFT of Ising type works in higher dimensions as well. 
This example is indeed related to the 3D Ising model and the 3D toric code \cite{HZW05,KTZ20}, which we will discuss in near future.

The connection between TQFT and topological orders has been extensively studied from the view of since \cite{Wen90,LevWen05} and from the view of particle excitations since \cite{Kit03}, and a classification of topological orders by unitary multi-fusion $n$-categories is proposed in \cite{KWZ15}. This is a classification of topological orders by quantum symmetries.
From the view of condensed matter physics, another natural approach to study topological orders is characterizing the ground states.

An important question proposed by Wen is what kind family of wave functions on lattice models are topological orders. 
The Riesz representation theorem induces a bijection between the linear functional $Z$ and a vector state on the configuration space for every lattice.
One can consider our linear functional $Z$ on a configuration space as a vector state up to a normalization.
Therefore by Theorem \ref{Thm: Alterfold TQFT}, we give an answer to Wen's question that a linear functional $Z$ with the three conditions (RP) (HI) and (CF) is a topological order.

In topological orders, we usually take this vector state as the ground state of a local Hamiltonian of neighbourhood interactions. The conditions (HI) and (RP) can be derived from the corresponding conditions of the Hamiltonian. In this paper, we focus on the construction of the alterfold TQFT without referring to the Hamiltonian. In a coming paper \cite{Liu24-preparation}, we study lattice models with a local Hamiltonian systematically and prove conditions (HI) and (RP) for the Hamiltonian at any temperature.
Condition (CF) corresponds to the finiteness of the entanglement rank of the ground state for lattices on $S^{k}\times S^{n-k}$ separated by $S^{k-1}\times S^{n-k}$.
The area law and projected entangled pair states (PEPS) have been conjectured to be the entanglement properties characterizing ground states of Hamiltonian with local interactions, see e.g., \cite{SCP10}.
Our functional integral approach will provide positive evidence of this conjecture.

\section{Labelled Regular Stratified Manifolds}\label{Sec: Labelled Regular Stratified Manifolds}
\subsection{Piecewise Linear Manifolds}

In this section, we briefly recall the definition of piecewise linear (PL) manifolds and some basic results. We refer readers to textbooks \cite{Hud69,RouSan72} for the general theory of PL manifolds.

 \begin{definition}
     A linear $n$-simplex $\Delta^n=[e_0,e_1,\cdots,e_n]$ is the convex hull of $n+1$ points $\{e_i\}_{0 \leq  i \leq n}$ in the Euclidean space $\mathbb{R}^n$, such that the vectors $\{e_i-e_0\}_{1 \leq  i \leq n}$ are linearly independent. Its orientation is the sign of the determinant of the matrix $\{e_i-e_0\}_{1 \leq  i \leq n}$.
 \end{definition}

\begin{definition}
The boundary of an oriented linear $n$-simplex is
$$\partial[e_0,e_1,\cdots,e_n]=\Sigma_{i=0}^{n} (-1)^i  [e_0,e_1,\cdots \hat{e_i} , \cdots e_n],$$
where the $\pm$ sign indicates the induced orientation of the $(n-1)$-simplex on the boundary. Its sub $(n-1-k)$-simplex is called a $k$-face, $0\leq k\leq n-1$.
\end{definition}

\begin{definition}
For convex sets $A$ and $B$ in $\mathbb{R}^n$, their convex sum is define as
$$A +_c B:=\{\lambda a+(1-\lambda) b : a\in A, ~b \in B, 0 \leq \lambda \leq 1\}.$$
\end{definition}

When we consider $\Delta^n$ as one-point suspension of $e_0$, i.e., $\{e_0\}+_c[e_1,e_2\cdots, e_n]$, we have:
$$\partial (\{e_0\}+_c[e_1,e_2\cdots, e_n])=[e_1,e_2\cdots, e_n]-(\{e_0\} +_c \partial [e_1,e_2\cdots, e_n]).
$$
When we consider $\Delta^n$ as one-point suspension of $e_n$, i.e., $[e_0,e_1\cdots, e_{n-1}] +_c \{e_n\}$, we have:
$$\partial ([e_0,e_1\cdots, e_{n-1}] +_c \{e_n\})=(\partial [e_0,e_1\cdots, e_{n-1}] +_c \{e_n\}) + (-1)^{n} [e_0,e_1\cdots, e_{n-1}].
$$

\begin{definition}
A polytope is the convex hull of finitely many points in $\mathbb{R}^n$. It is $k$-dimensional, if it is a $k$-dimensional topological manifold.
\end{definition}

A polytope is the bounded intersection of finitely many half-spaces in $\mathbb{R}^n$.

\begin{definition}
A transformation on $\mathbb{R}^n$ is called piecewise linear, if it is affine on every linear simplex of a triangulation.     
\end{definition}

\begin{definition}
A topological $n$-manifold $M$ with boundary is called PL, if it is equipped with an open covering $(U_i)_{i\in I}$ and homeomorphisms $\phi_i: U_i \to \mathbb{R}^{n-1}\times [0,\infty)$ onto their images, such that the transition map $\phi_j \phi_i^{-1}$ is PL on $\phi_i^{-1}(U_i\cap M_j)$, for any $i,j \in I$.
If all $\phi_j \phi_i^{-1}$ are oriented, then $M$ is oriented.
\end{definition}
A point $p$ is on the boundary $\partial(M)$, iff $\phi_i(p)\in \mathbb{R}^{n-1}\times \{0\},$ for some $i\in I$. 
A point $p$ is in the interior $\mathring{M}$, iff $\phi_i(p)\in \mathbb{R}^{n-1}\times (0,\infty),$ for some $i\in I$. 
The pair $(U_i,\phi_i)$ is called a chart of $M$. Two charts are compatible, if the transition map is PL. For convenience, we assume the collection $\{(U_i,\phi_i) : i\in I\}$ to be a maximal atlas. That means any chart $(U,\phi)$ compatible with all $(U_i,\phi_i)$ is in the atlas. The existence of a maximal atlas is guaranteed by Zorn's lemma.

\begin{definition}
A homeomorphism $\phi: M \to N$ of PL manifolds is a topological homeomorphism which is PL on $\phi_i^{-1}(U_i\cap \phi^{-1}V_j)$, for any chart $(U_i,\phi_i)$ of $M$ and  $(V_i,\phi_j)$ of $N$.
In this case, PL manifolds $M$ and $N$ are called homeomorphic, denoted as $M\sim N$.
\end{definition}

For a vector $x=(x_1,x_2,\ldots,x_n) \in \mathbb{R}^n$, its infinity norm is
$$\displaystyle \|x\|=\max_{1\leq i \leq n} |x_i|.$$

\begin{notation}
With the infinity norm, we denote the closed unit $n$-disc as $D^n=[-1,1]^n$, the open unit $n$-disc as $\mathring{D}^n=(-1,1)^n$, the boundary unit sphere as $S^{n-1}=\partial D^{n+1}$, the half $n$-discs as $D^n_+=D^{n-1}\times [0,1]$ and $D^n_-=D^{n-1}\times [-1,0]$. We take the convention that $S^{-1}=\emptyset.$
For those with radius $\lambda>0$, we denote them as $\lambda D^n$, $\lambda\mathring{D}^n$, $\lambda S^{n-1}$, $\lambda D^n_\pm$ respectively.
\end{notation}

An $n$-dimensional polytope $P$ is PL homeomorphic to $D^n$ and its boundary $\partial P$ is called a combinatorial $(n-1)$-sphere, which is \PL homeomorphic to $S^{n-1}$.
When the origin $O_n$ is in the interior of the polytope $P$, we construct the homeomorphism $P\sim D^n$ as
\begin{align}\label{Equ: polytope to Dn}
\phi(rx)= \|x\|^{-1} rx ~\forall~ x \in \partial P, r \in [0,1].    
\end{align}

\begin{definition}
Suppose $\Delta^n=[e_0,e_1,\cdots,e_n]$ is a linear $n$-simplex and $p$ is a point in the interior of $\Delta^k=[e_0,e_1,\cdots,e_k]$,
The $p$-subdivision of $\Delta^n$ is the union of $n$-simplices
$$\bigcup_{i=0}^k [e_0,e_1,\cdots \hat{e_i} , \cdots e_k] +_c \{p\} +_c [e_{k+1},e_{k+2},\cdots \cdots e_n].$$
\end{definition}

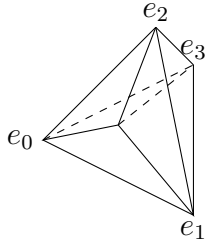
\begin{figure}[H]
    \centering
    \begin{tikzpicture}
\draw[dashed] (0,0.2) -- (1,1);
\draw (1,1)-- (.5,1.5);
\draw (0,0.2) -- (1,-1) -- (.5,1.5);
\draw (0,0.2) -- (-1,0) -- (.5,1.5);
\draw (0,0.2) -- (.5,1.5);
\draw (1,1)--(1,-1);
\draw (-1,0) --(1,-1); 
\draw[dashed] (-1,0) -- (1,1);
\node at (-1.3,0) {$e_0$};
\node at (1,-1.2) {$e_1$};
\node at (1,1.2) {$e_3$};
\node at (.5,1.7) {$e_2$};
\end{tikzpicture}
    \caption{The subdivision for $n=3$, $k=2$.}
    \label{fig:subdivision}
\end{figure}

\begin{definition}
Suppose $M$ is a PL $n$-manifold and $(U,\phi)$ is a chart of $M$. For a linear simplex $\Delta^n \subset \phi(U)$, we call $\phi^{-1}(\Delta^n)$ a simplex of $M$. 
\end{definition}

\begin{definition}
A triangulation of an oriented PL $n$-manifold $M$ is a union of simplices $M=\cup_i \sigma_i$, such that their interior are disjoint and each 0-face of a simplex is either on the boundary of $M$, or shared by two simplices with opposite orientations. 
\end{definition}

\begin{definition}
In a triangulation of an oriented PL $n$-manifold $M$,
the link of a $k$-simplex $\sigma$ is the union of $(n-k-1)$-simplices $\tau_i$, for which $\sigma \cup \tau_i$ is in an $n$-simplex and $\sigma \cap \tau_i=\emptyset$.
The triangulation is called combinatorial, if the link of every $k$-simplex is a
combinatorial $(n-k-1)$-sphere.
\end{definition}

The following result is well-known and fundamental to study PL manifolds by triangulations, see Chapter 3 in \cite{Hud69}.
\begin{theorem}\label{Thm: triangulation}
A compact PL $n$-manifold admits a combinatorial triangulation.
Moreover, two different triangulations have a common subdivision.
\end{theorem}

Two triangulations can be changed from one to another by a sequences of Pachner moves \cite{Pan91}.

\begin{definition}
Suppose $M$ and $N$ are PL manifolds.
If $\phi_t: M \to N$, $t\in [0,1]$ are homeomorphisms onto the image, such that the map $[0,1] \times M \to N$ is PL, then $\phi_0$ and $\phi_1$ are called \PL isotopic, denoted by $\phi_0\sim \phi_1$; 
Moreover, $\phi_0(M)$ and $\phi_1(M)$ are called isotopic in $N$. 
\end{definition}

\begin{definition}
For a PL manifold $M$, its mapping class group $MCG(M)$ is the quotient of the group of homeomorphisms on $M$ by the subgroup of homeomorphisms isotopic to the identity. When $M$ is oriented, the homeomorphisms in $MCG(M)$ are required to be oriented. When $M$ has a boundary, the homeomorphisms are required to be the identity on the boundary.
\end{definition}

It is well-known that the mapping class group of $D^n$ and $S^n$ are trivial. For readers' convenience, we give a proof here.

\begin{proposition}\label{Prop: Alexander}
The mapping class group of the PL manifold $D^n$ is trivial.
\end{proposition}
\begin{proof}
Suppose $\phi$ is a \PL homeomorphism of $D^n$ which is the identity on the boundary $S^{n-1}$. Then it is isotopic to identity by Alexander's trick,
\begin{align}\label{Equ: Alexander}
\phi_t(x) :=
\left\{
\begin{aligned}
 & t \phi(x/t), ~ 0 \leq \|x\| \leq t, \\
 & x, ~  t \leq \|x\| \leq 1.
\end{aligned}
\right.
\end{align}
\end{proof}

\begin{notation}
Note that $S^1$ contains four intervals, a clockwise translation of length $\varepsilon$ along $S^1$ of is a \PL homeomorphism on $S^1$.
The $180^{\circ}$ rotation $\rho$ on $S^1$ is a translation of length 2.
\end{notation}

\begin{proposition}
The mapping class group of the PL manifold $S^n$ is trivial.
\end{proposition}

\begin{proof}
Suppose $\phi$ is a \PL homeomorphism on $S^n$.
When $n=0$, $S^0$ has two points with opposite orientations. The only orientation preserving homeomorphism is the identity.

When $n\geq 1$, we take a point $p\in S^n$, such that $\phi$ is linear from a neighbourhood of $p$ to a neighbourhood of $\phi(p)$. We rotate $\phi(p)$ to $p$ by PL isotopy. So we may assume that $\phi$ preserves $p$ and it is linear in a neighbourhood of $p$.
Let $U_\varepsilon$ be the closed $\varepsilon$-neighbourhood of $p$.
There are $\varepsilon> \varepsilon'>0$, such that $(\phi')^{-1}(U_{\varepsilon'})$ is in the interior of $U_\varepsilon$.

Now we show that there is a homeomorphism $\phi'$ on $U_{\varepsilon'}$ which is the identity outside $U_{\varepsilon}$ and $\phi\phi'$ is the identity on $U_{\varepsilon'}$.
When $n=1$, it is obvious.
When $n=2$, we construct the homeomorphism $\phi'$ illustrated below.
(Here is the example when $\phi'$ is the $90^{\circ}$ on $|U_\varepsilon|\sim D^2$.)
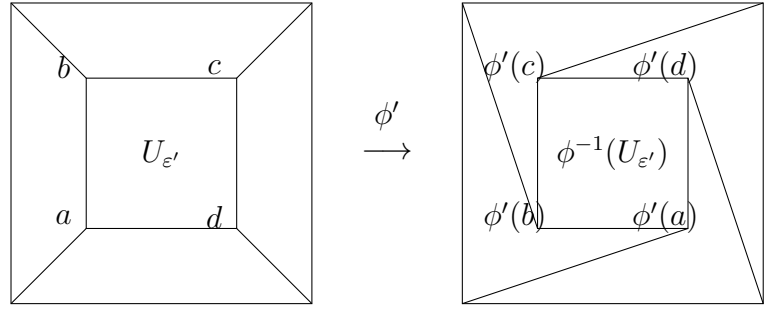
\begin{center}
\begin{tikzpicture}
\begin{scope}[shift={(-.3,.15)}]
\node at (-1,-1) {$a$};
\node at (-1,1) {$b$};
\node at (1,1) {$c$};
\node at (1,-1) {$d$};
\end{scope}
\draw (-2,-2) rectangle (2,2);
\draw (-1,-1) rectangle (1,1);
\draw (-1,-1) -- (-2,-2);
\draw (-1,1) -- (-2,2);
\draw (1,1) -- (2,2);
\draw (1,-1) -- (2,-2);
\node at (3,.5) {$\phi'$};
\node at (3,0) {$\longrightarrow$};
\node at (0,0) {$U_{\varepsilon'}$};
\begin{scope}[shift={(6,0)}]
\begin{scope}[shift={(-.3,.15)}]
\node at (-1,-1) {$\phi'(b)$};
\node at (-1,1) {$\phi'(c)$};
\node at (1,1) {$\phi'(d)$};
\node at (1,-1) {$\phi'(a)$};
\end{scope}
\draw (-2,-2) rectangle (2,2);
\draw (-1,-1) rectangle (1,1);
\draw (-1,-1) -- (-2,2);
\draw (-1,1) -- (2,2);
\draw (1,1) -- (2,-2);
\draw (1,-1) -- (-2,-2);
\node at (0,0) {$\phi^{-1}(U_{\varepsilon'})$};
\end{scope}
\end{tikzpicture}
\end{center}
The general case for $n\geq 2$ reduces to the case $n=2$ as $GL(n,\mathbb{R})$ is generated by $GL(2,\mathbb{R})$ for pairs of coordinates.

By Prop.~\ref{Prop: Alexander}, both $\phi'$ and $\phi\phi'$ are isotopic to the identity. So $\phi$ is isotopic to the identity.
\end{proof}

\subsection{Stratified PL Manifolds}

In this section, we introduce regular stratified PL $n$-manifolds. It share ideas in the study of stratified smooth manifolds. We refer the readers to the textbook \cite{Pfl01} for the general theory of stratified smooth manifolds. Stratified smooth manifolds have been used to study defects in TQFT, see \cite{CRS19} and further reference therein.

\begin{definition}
Suppose $M^n$ is a closed PL manifold. A stratified PL $n$-manifold $\mathcal{M}$ with support $M^n$ is a stratification $\{M^{k} \supseteq M^{k-1}, 0\leq k \leq n \}$, $M^{-1}=\emptyset$, such that $M^k\setminus M^{k-1}$ is an open PL $k$-manifold with closure in $M^{k}$; and any point $p$ has a neighbourhood $U$, whose intersection with $M^k\setminus M^{k-1}$ has finite connected components, for all $k$.
\end{definition}

\begin{definition}
When $M^n$ has a boundary, a stratified PL $n$-manifold $\mathcal{M}$ with support $M^n$ is a stratification $\{M^{k} \supseteq M^{k-1}, 0\leq k \leq n \}$, $M^{-1}=\emptyset$, such that $(M^k\setminus M^{k-1})\cap \mathring{M}^n$ is an open PL $k$-manifold;
and the boundary of $\partial(\mathcal{M})$ is a closed stratified PL $(n-1)$-manifold with a stratification  $\{\partial(\mathcal{M})^{k} \supseteq \partial(\mathcal{M})^{k-1}, 1\leq k \leq n-1 \}$, where $\partial(\mathcal{M})^{k-1}$ is the transversal intersection of $\partial(M^n)$ and $M^{k}$. 
\end{definition}

\begin{remark}
Sometimes $M^{k-1}$ are considered as defects of $M$ or domain walls between connected components of $M^{k} \setminus M^{k-1}$ in topological quantum field theory.
\end{remark}

\begin{definition}
When $\mathcal{M}$ is a stratified manifold with support $S^{n-1}$, we denote $\Lambda \mathcal{M}$ as its one-point suspension with the origin $O_n$. Then $\Lambda \mathcal{M}$ has support $D^n$, boundary $\partial(\Lambda \mathcal{M})= \mathcal{M}$, and interior $\mathring{\Lambda} \mathcal{M}:= \Lambda \mathcal{M}\setminus \partial\mathcal{M}$.
It has a stratification $(\Lambda \mathcal{M})^{k+1}:= \Lambda M^{k} \cup M^{k+1}$, $0 \leq k \leq n-1$, and $(\Lambda \mathcal{M})^{0}:=\{O_n\}\cup M^0$.
Moreover $(\Lambda \mathcal{M})^{k+1} \setminus (\Lambda \mathcal{M})^{k}$ is homeomorphic to $((M^{k}\setminus M^{k-1})\times (0,1)) \cup (M^{k+1}\setminus M^k)$.
\end{definition}

\color{black}

\begin{definition}
Suppose $\mathcal{M}$ is a stratified manifold.
A point $p$ in $(\mathcal{M}\setminus \partial\mathcal{M}) \cap (M^{k} \setminus M^{k-1})$ is called regular, if there is a closed neighbourhood $U$ of $p$ in $M^n$, a stratified PL manifold $\mathcal{S}$ with support $S^{n-k-1}$, and a homeomorphism $\phi: \mathcal{M}|_U \to \Lambda \mathcal{S} \times D^{k}$;
such that 
\begin{enumerate}
    \item $\phi(M^k \cap U)=O_{n-k} \times  D^{k}$;
    \item $\phi(p)=O_{n}$;
    \item $U \cap (M^{j} \setminus M^{j-1})$ has finite connected components $\{V_{j,i} : i \in I_j\}$, for all $j$;
    \item and for each connected component $V_{j,i}$, $\phi(V_{j,i})$ is in a $j$-dimensional subspace.
\end{enumerate}
A point $p$ in $\partial \mathcal{M} \cap (M^{k} \setminus M^{k-1})$ is called regular, if there is a closed neighbourhood $U$ of $p$ in $M^n$, a stratified PL manifold $\mathcal{S}$ with support $S^{n-k-2}$, and a homeomorphism $\phi: \mathcal{M}|_U \to \Lambda \mathcal{S} \times D^{k+1}_+$, such that
\begin{enumerate}
    \item $\phi(M^k \cap U)=O_{n-k-1} \times  D^{k+1}_+$;
    \item $\phi(p)=O_{n}$;
    \item $U \cap (M^{j} \setminus M^{j-1})$ has finite connected components $\{V_{j,i} : i \in I_j\}$, for all $j$;
    \item and for each connected component $V_{j,i}$, $\phi(V_{j,i})$ is in a $j$-dimensional subspace.
\end{enumerate}
We call the triple $(U,\phi,\mathcal{S})$ a regular chart of $p$, $\mathcal{S}$ the link boundary of $p$, and $U$ the normal microbundle of the $k$-simplex $M^k \cap U$ near $p$.
\end{definition}
A boundary point $p$ is regular, iff $p$ has a neighbourhood $U$ and a regular chart $(\partial\mathcal{M} \cap U,\phi,\mathcal{S})$ in $\partial\mathcal{M}$, such that $U \sim (\partial\mathcal{M} \cap U)\times [0,1]$.
The link boundary of a regular point is well-defined up to homeomorphism.

\begin{proposition}\label{Prop: link boundary}
If a point $p$ has regular charts $(U_i,\phi_i,\mathcal{S}_i)$, $i=1,2$, then
$\mathcal{S}_1\sim \mathcal{S}_2$. 
\end{proposition}

\begin{proof}
For $p$ in $(\mathcal{M}\setminus \partial\mathcal{M}) \cap (M^{k} \setminus M^{k-1})$, the transition map $\phi_2\phi_1^{-1}$ is PL on $V=\phi_1(U_1\cap U_2)$.
Take a triangulation of the closure $V$, such that the transition map is linear on every simplex. By subdivisions, we may assume that $V\cap O_{n-k}\times D^k$ contains a $k$-simplex $E$. Take all $n$-simplices $\{\sigma_j\}_{j\in J}$ containing $E$. Let $L$ be the link of $E$. Take the projection from $\pi: D^n \to D^{n-k} \times \{0\}$. Then $L\sim \pi(L)$. By the trick in Equ.\ref{Equ: polytope to Dn}, $\pi(L) \sim S^{n-k-1}$. So $\phi_1^{-1}(L) \sim \mathcal{S}_1$.
In $\phi_2(U_2)$, $\{\phi_2\phi_1^{-1}(\sigma_j)\}_{j\in J}$ are the simplices containing $\phi_2\phi_1^{-1}(E)$, $\phi_2\phi_1^{-1}(E) \subseteq D^{n-k} \times 0$ and it has the link $\phi_2\phi_1^{-1}(L)$. Similarly $\phi_1^{-1}(L)\sim \mathcal{S}_2$. So $\mathcal{S}_2 \sim \mathcal{S}_1$.

For $p$ in $(\partial \mathcal{M}) \cap (M^{k} \setminus M^{k-1})$, we consider it as a regular point in $\partial \mathcal{M}$. So $\mathcal{S}_2 \sim \mathcal{S}_1$ as well.
\end{proof}

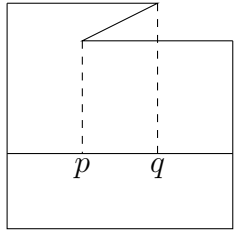
\begin{figure}[h]
    \centering
\begin{tikzpicture}
\draw (0,0) -- (2,0) --(1,-.5) -- (3,-.5);
\draw (0,-2) -- (3,-2);
\draw (0,0) -- (0,-3) -- (3,-3) -- (3,-.5);
\draw[dashed] (1,-.5) --(1,-2);
\draw[dashed] (2,0)--(2,-2);
\node at (1,-2.2) {$p$};
\node at (2,-2.2) {$q$};
\end{tikzpicture}
    \caption{Irregular Points}
    \label{fig:irregular points}
\end{figure}
\begin{center}

\end{center}

Fig.\ref{fig:irregular points} is $M^2$ of a stratified PL manifold $\mathcal{M}$. The horizontal line is $M^1$, which separates $M^2$ into two components, and $M_0=\emptyset$. The points $p,q \in M^{1}$ are not regular.

We describe the local shape of lattice models by link boundary. 
The local shapes shall be good, so that all points are regular in all dimensions.  
We want to describe such lattices by good stratified manifolds based on regular conditions. 
Now let us define regular stratified PL $n$-manifolds inductively on $n$. 
\begin{definition}
A regular stratified PL $0$-manifold is a $0$-manifold.
A stratified PL $n$-manifold $\mathcal{M}$ is called regular, if every point $p$ is regular and it has a regular chart $(U,\phi, \mathcal{S})$ for a regular PL manifold $\mathcal{S}$.
\end{definition}

\begin{definition}
A homeomorphism $\phi: \mathcal{M} \to \mathcal{N}$ of stratified PL manifolds is an homeomorphism of PL manifolds $M^n \to N^n$, and the restriction on every $M^{k} \setminus M^{k-1}$ is a homeomorphism.
In this case, we say $\mathcal{M}$ and $\mathcal{N}$ are homeomorphic, denoted by $\mathcal{M} \sim \mathcal{N}$. 

\end{definition}

Let us prove that homeomorphisms of PL manifolds can be extended to homeomorphisms of regular stratified PL manifolds. Based on it, we generalize good properties of PL manifolds, such as Theorem ~\ref{Thm: triangulation}, to regular stratified PL manifolds.

\begin{proposition}\label{Prop: Hom of PL manifolds}
Suppose $\mathcal{M}$ is a regular stratified PL $n$-manifold and $\phi: M^n \to N^n$ is a homeomorphism PL $n$-manifold. Then $\phi: \mathcal{M} \sim \mathcal{N}$ is a homeomorphism of regular stratified PL manifolds, where $\mathcal{N}$ has a stratification $\{\phi(M^{k}) \supseteq \phi(M^{k-1}), 0\leq k \leq n \}$.
\end{proposition}

\begin{proof}
For any point $p$ in $M^{k} \setminus M^{k-1}$, it has a regular chart $(U,\phi',\mathcal{S})$ for a regular $\mathcal{S}$ with support $S^{n-k-1}$, such that $\phi'(\mathcal{M}|_U)=\Lambda \mathcal{S} \times  D^{k}$ and $\phi'(M_k \cap U)=O \times D^{k}$.
Then the point $\phi(p)$ has a regular chart $(\phi(U),\phi\phi',\mathcal{S})$.
In particular, $\phi(M^{k}) \supseteq \phi(M^{k-1})$ is an open PL $k$-manifold with closure $\phi(M^{k})$, its intersection with $\phi(U)$ has finite connected components.
So $N^n$ is regular with a stratification $\{\phi(M^{k}) \supseteq \phi(M^{k-1}), 0\leq k \leq n \}$. Moreover, $\phi$ is a homeomorphism of regular stratified PL manifolds.
\end{proof}

\begin{definition}
Suppose $\phi_t: M \to N$, $t\in \{0,1\}$ are isotopic homeomorphisms onto the image, and $\mathcal{M}$ is a stratified manifold with support $M$.
We say their images $\phi_0(\mathcal{M})$ and $\phi_1(\mathcal{M})$ are isotopic in $N$. 
\end{definition}

\begin{definition}
Suppose $\mathcal{M}$ is a regular stratified PL $n$-manifold, and $N^j$ is sub PL $j$-manifold of the PL manifold $M^n$. 
For a point $p$ in $N^j \cap (M^k \setminus M^{k-1})$, we say $\mathcal{M}$ and $N^j$ are transversal at $p$, if 
$p$ has a regular chart $(U,\phi,\mathcal{S})$, such that 
$$\phi(U\cap N^j)= D^j \times O_{n-j} ~\text{or}~ D^{j}_+\times O_{n-j}.$$
We call $p$ a transversal point.
We say $\mathcal{M}$ is transversal to $N^j$, if they are transversal at all points of intersection.   
In this case, $\mathcal{N}:=\{ N^{j+k-n}=N^j \cap M^k : j+k\geq n\}$ is a regular stratified PL $j$-manifold with regular charts $(U\cap N^j, \phi|_{U\cap N^j}, \mathcal{S})$, called the sub $j$-manifold of $\mathcal{M}$.
\end{definition}

Note that if $\mathcal{M}$ and $N^j$ are transversal at $p$, then $j+k\geq n$.

\begin{definition}
Suppose $\mathcal{M}$ is a regular stratified PL $n$-manifold. A triangulation of $M^n$ is called transversal to $\mathcal{M}$, if every simplex is transversal to $\mathcal{M}$.      
\end{definition}

We state a special case of the transversal theorem for PL manifolds of Williamson in Theorem 3.3.1 in \cite{Wil66}. We try to follow his symbols in that theorem.
\begin{lemma}\label{Lem: transversality}
Suppose $B$ is a $j$-simplex in a PL $n$-manifold $T$. 
Suppose $S$ is a closed PL $k$-manifold in $T$ with boundary $Q$, and $Q$ has a neighbourhood $X$, such that $X$ is transversal to $B$, then in any given neighborhood of $S$, there are isotopy $H_t$ fixing $Q$, $0\leq t\leq 1$, such that $H_0$ is identity and $H_1(S)$ is transversal to $B$.    
\end{lemma}

\begin{theorem}\label{Thm: transversality}
Suppose $\Delta$ is a combinatorial triangulation of a compact PL $n$-manifold $M^n$. 
Suppose $\mathcal{M}$ is a regular stratified PL $n$-manifold with support $M^n$, which is transversal to $\Delta$ on the boundary $\partial M^n$. Then there is an ambient isotopy $\phi_t$ fixing the boundary, such that
$\phi_1(\mathcal{M})$ is transversal to $\Delta$. 
\end{theorem}

\begin{proof}
Let us prove the statement by induction on $n$.
The statement is trivial for $n=0$.
Suppose the statement is true for $(n-1)$-manifolds.

Take an $n$-simplex $\sigma$ of the triangulation of $M^n$. Denote the support of $\sigma$ by $S$ and the boundary by $B$. Take a neighbourhood $U$ of $\sigma$. 
We first ambient isotope $M_0$ away from $S$.

By induction on $k=1,\cdots, n-1$, we prove that we can ambient isotope $\mathcal{M}$ in $U$, such that $\mathcal{M}$ and $S$ are transversal at all points in $M^k$.

Suppose $\mathcal{M}$ and $S$ are transversal at all points in $S\cap M^{k-1}$.
Then the intersection of $M^{k-1}$ and the $n-k$-skeletons of $\sigma$ is empty.
By the compactness of $S\cap M^{k-1}$, it has a neighbourhood of transversal points.
We ambient isotope $M^k\setminus M^{k-1}$ away from the $n-k-1$-skeletons fixing a neighbourhood of $M^{k-1}$. By the regularity, for every connected component $C^k$ of $M^k\setminus M^{k-1}$, its closure is a PL sub $k$-manifold and its boundary is in $M^{k-1}$, which has a neighbourhood transversal to $S$. By Lemma \ref{Lem: transversality}, we ambient isotope every $C^k$ fixing a neighbourhood of $M^{k-1}$, so that $C^k$ is transversal to $S$. Moreover ambient isotope of different $C^k$'s are independent of each other. Then  the intersection of $M^{k}\setminus M^{k-1}$ and $S$ are transversal. It completes the induction on $k$ and then $\mathcal{M}$ is transversal to $S$.

By the compactness of $S$, $\mathcal{M}|_{S}$ has a neighbourhood homeomorphic to $\mathcal{M}|_{S} \times [-1,1]$. 
Note that $S$ has a triangulation $\partial \sigma$. By induction on $n$, we can ambient isotope $\mathcal{M}|_{S}$, so that $\mathcal{M}|_{S}$ is transversal to $\partial \sigma$. 
Then we extend the ambient isotope to the neighbourhood. 
We end up with $\mathcal{M}$ transversal to the $n$-simplex $\sigma$.
Similarly we ambient isotope $\mathcal{M}$ for other $n$-simplices without changing the neighbourhood of the previous simplices. Eventually, $\mathcal{M}$ is transversal to the triangulation.
\end{proof}

\begin{definition}
A regular stratified PL $n$-manifold $\mathcal{N}$ is called a sub manifold of a regular stratified PL $n$-manifold $\mathcal{M}$, if $N^{n}$ is a sub PL manifold of $M^n$ and $N^k=N^n\cap M^k$, for any $0\leq k \leq j$. 
\end{definition}

\begin{figure}[H]
    \centering
    \begin{tikzpicture}
\begin{scope}[yscale=.5]
\draw (-2.75,-3) --++ (4,0) --++(1.5,6) --++ (-4,0)--++(-1.5,-6);
\draw[dashed] (0,0)--(-1,0);
\draw[dashed] (0,0)--(1,1);
\draw[dashed] (0,0)--(.5,-1);
\draw (0,-2)--(0,2);
\draw (-1,-2)--(-1,2);
\draw (.5,-3)--(.5,1);
\draw (1,-1)--(1,3);
\begin{scope}[shift={(0,2)}]
\draw (0,0)--(-1,0);
\draw (0,0)--(1,1);
\draw (0,0)--(.5,-1);    
\end{scope}
\begin{scope}[shift={(0,-2)}]
\draw (0,0)--(-1,0);
\draw (0,0)--(1,1);
\draw (0,0)--(.5,-1);    
\end{scope}
\end{scope}
\end{tikzpicture}    
    \caption{An example of a regular stratified manifold}
    \label{fig:enter-label}
\end{figure}
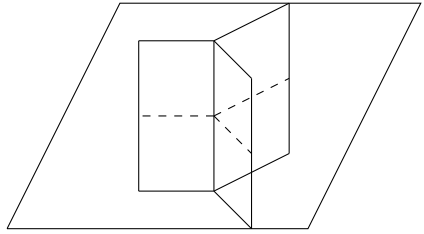

We mainly use the following two operations of PL $n$-manifolds: (1) disjoint union (2) gluing. 
We show that both operations extend to regular stratified PL $n$-manifolds.
Obviously the disjoint union of regular PL $n$-manifolds is a regular PL $n$-manifold.
We need to deal with the charts at the corner when we glue the boundary. 

\begin{lemma}\label{Lem: corner}
There is a PL homeomorphism $\phi$ from $[-1,1]\times [0,1]$ to $[0,1]\times [0,1]$ which preserves $[0,1]\times \{0\}$ and maps $[-1,0] \times \{0\}$ to $\{0\}\times [0,1]$.    
\end{lemma}

\begin{proof}
\begin{tikzpicture}
\draw (-1,0) rectangle (1,1);
\begin{scope}[shift={(-.2,-.2)}]
\node at (-1,0) {$A$};
\node at (0,0) {$B$};
\node at (1,0) {$C$};   
\end{scope}
\begin{scope}[shift={(4,0)}]
\begin{scope}[shift={(-.2,-.2)}]
\node at (-1,0) {$A$};
\node at (0,0) {$B$};
\node at (1,0) {$C$};   
\end{scope}
\node at (-2,.5)  {$\rightarrow$};
\draw (-1,0) -- (1,1)--(1,0)--(-1,0);
\end{scope}
\begin{scope}[shift={(8,0)}]
\begin{scope}[shift={(-.2,-.2)}]
\node at (0,.5) {$A$};
\node at (0,0) {$B$};
\node at (1,0) {$C$};   
\end{scope}
\node at (-2,.5)  {$\rightarrow$};
\draw (0,.5) -- (1,1)--(1,0)--(0,0)--(0,.5);
\end{scope}
\begin{scope}[shift={(12,0)}]
\begin{scope}[shift={(-.2,-.2)}]
\node at (0,1) {$A$};
\node at (0,0) {$B$};
\node at (1,0) {$C$};    
\end{scope}
\node at (-2,.5)  {$\rightarrow$};
\draw (0,0) rectangle (1,1);    
\end{scope}
\end{tikzpicture}    
\end{proof}

\begin{proposition}
Suppose $\mathcal{M}$ is a regular stratified PL $n$-manifold. Suppose $\mathcal{P}$ and $\mathcal{Q}$ are regular stratified closed PL sub $(n-1)$-manifolds of $\partial \mathcal{M}$, $\mathring{\mathcal{P}} \cap \mathring{\mathcal{Q}}=\emptyset$, $\partial\mathcal{P}\cap \partial\mathcal{Q}$ is a sub $(n-2)$-manifold $\partial \mathcal{M}$, and $\phi: \mathcal{P}\to \mathcal{Q}$ is an orientation reversing homeomorphism. 
Then we obtain a regular stratified PL $n$-manifold $\mathcal{M}/\phi$ by gluing $p\in \mathcal{P}$ with $\phi(p)\in \mathcal{Q}$.     
\end{proposition}

\begin{proof}
When $p$ is in the interior of $P$, take a regular chart $(U,\phi_p,\mathcal{S})$ of $p$ in $\partial M$, then $(\phi U, \phi_p\phi^{-1}, \mathcal{S})$ is a regular chart of $q=\phi(p)$ in $\partial M$. Their neighbourhoods in $\mathcal{M}$ are normal bundles of $U_p$ and $\phi U_p$. The chart of the union the two neighbourhood is the union of the charts on the normal bundles. 

When $p$ is on the boundary of $\mathcal{P}$, but not in $\partial\mathcal{P}\cap \partial\mathcal{Q}$, we apply the homeomorphism in Lemma \ref{Lem: corner} to the charts of $p$ and $q$ along $\partial (P\cup Q)$, and then glue the charts. 

When $p$ is an inner point of $\partial\mathcal{P}\cap \partial\mathcal{Q}$, we apply the inverse homeomorphism in Lemma \ref{Lem: corner} to the charts of $p$ and $q$ along $\partial P\cap \partial Q)$, and then glue the charts. 

When $p$ is on the boundary of $\partial\mathcal{P}\cap \partial\mathcal{Q}$, we apply the homeomorphism in Lemma \ref{Lem: corner} to the charts of $p$ and $q$ along $\partial (P\cup Q)$, and then
we apply the inverse homeomorphism in Lemma \ref{Lem: corner} to the charts along $\partial P\cap \partial Q$
and then glue the charts. 
\end{proof}

\subsection{Labelled Stratified Manifolds}
In the rest of the paper, the manifold in the top dimension $n$ is always assumed to be oriented, compact, PL. All stratified manifolds and charts are regular. We will ignore the terminology oriented, compact, PL and regular for simplicity, if there is no confusion.

For a`lattice on an oriented compact $n$-manifold $M^n$,
we describe its lattice shape as a stratified $n$-manifold $\mathcal{M}$ with support $M^n$. The local lattice shape is described by the charts $(U,\phi,\mathcal{S})$. Usually there are only finitely many local lattice shapes, so there are finitely many choices of $\mathcal{S}$.
We consider a set $\PP$ as the positions to assign vectors, such as spins.
For points in $\PP$, if they have the same local shape $\Lambda\mathcal{S}$, $|\mathcal{S}|=S^{n-1}$, they shall be assigned with vectors in the same vector space $V_{\mathcal{S}}$.
If the local shape of a position point is $ \Lambda\mathcal{S}^{n-k-1} \times \mathring{D}^{k}$, $|\mathcal{S}^{n-k-1}|=S^{n-k-1}$, an assignment is a replacement of
the closed stratified manifold $\mathcal{S} \times D^{k}$ by a spin vector.
A configuration of the lattice is an assignment of vectors to points in $\PP$.

Now let us formulate their mathematical definitions.

\begin{definition}\label{Def:LS}
Suppose $\LS_k$ is a set of stratified PL manifolds with support $S^{n-k-1}$, $0\leq k\leq n$ and $\LS_{\bullet}=\cup_{k=0}^n \LS_k$.
We say a PL manifold $\mathcal{M}$ has local shape $\LS_{\bullet}$, if for any $0\leq k \leq n$ and $p\in M^{k} \setminus M^{k-1}$, $p$ has a chart $(U,\phi,\mathcal{S})$, $\mathcal{S}\in \LS_k$.
We call $\LS_{\bullet}$ the local shape set.
\end{definition}

We fix the local lattice shape first and we only consider stratified manifolds of a given type $\LS_{\bullet}$. For example, when $n=2$, $k=0$, the local shape of a point $p\in M^0$ is $S^1$ with a stratification of $m$ points. Then $p$ is a $m$-valent point. Usually $m$ is $2,3,4,5,6$ for 2D lattices.

\begin{definition}\label{Def:LL}
For a local shape set $\LS$, a label space $L$ is a set of vector spaces  $L=\{L_{\mathcal{S}} : \mathcal{S} \in \LS_{\bullet}\}$  over the field $\mathbb{K}$.
\end{definition}

\begin{definition}
For a label space $L$, an $L$-labelled stratified manifold $\mathcal{M}$ consists of
\begin{enumerate}
    \item a set of points $\PP$, called position points;
    \item a regular chart $(U_p, \phi_p, \mathcal{S}_p)$, for every $p\in \PP$; 
    \item and a vector $\otimes_{p\in \PP} v_p \in \otimes_{p\in \PP} L_{\mathcal{S}_p}$. 
\end{enumerate}
such that $M^0\subseteq \PP \subseteq  \mathring{\mathcal{M}}$; $\phi_p$ is orientation preserving; $\partial(\mathcal{M})$ and $U_p$, $p\in\PP$, are pairwise disjoint.
\end{definition}

\begin{notation}
We consider the $L$-labelled stratified manifold $\mathcal{M}$ as
a replacement of $U_p$ in $\mathcal{M}$ by the label $\phi_p^{-1}(v_p)$, for all $p\in \PP$.
Diagrammatically, we draw $\phi^{-1}(v_p)$ at $p$ of the stratified manifold.
\end{notation}

\begin{definition}\label{Def: condensation}
For a stratified manifold $\mathcal{M}$, we define its configuration space $\mathcal{M}(L)$ to be the vector space spanned by $L$-labelled stratified manifolds $\mathcal{M}$.
For a \PL manifold $M$ with boundary $S$ and a stratified manifold $\mathcal{S}$ with support $S$, 
We define the condensation space on $M$ with boundary $\mathcal{S}$ as
$$M_{\mathcal{S}}(L):=\{\ell\in \mathcal{M}(L): |\mathcal{M}|=M, \partial \mathcal{M}=\mathcal{S}\}.$$
When $\partial M=\emptyset$, we write $M(L)$ for short.
\end{definition}
The two operations disjoint union and gluing extend linearly to labelled stratified manifold.

By Prop.~\ref{Prop: Hom of PL manifolds}, a homeomorphism of manifolds $\phi: M\to N$,
induces a homeomorphism of stratified manifolds $\phi: M(L)\to N(L)$.

\begin{definition}
For an orientation preserving homeomorphism of manifolds $\phi: M\to N$,
and an $L$-labelled stratified manifold $\ell \in M(L)$, with positions points $\PP$, charts $\{(U_p, \phi_p, \mathcal{S}_p): p\in \PP\}$, and vectors $\{v_p \in L_{\mathcal{S}_p} : p\in \PP\}$,
we define a labelled stratified manifold $\phi(\ell)\in N(L)$ with
\begin{enumerate}
    \item a set of position points $\phi(\PP)$;
    \item every $\phi(p)\in \phi(\PP)$ has a chart $(\phi(U_p), \phi_p\phi^{-1}, \mathcal{S}_p)$, and a vector $v_p$ in $L_{\mathcal{S}_p}$.
\end{enumerate}
\end{definition}

\section{Hyper-Sphere Functions}\label{Sec: Hyper-Sphere Functions}

\subsection{Hyper-Sphere Functions}
\begin{definition}\label{Def: Z HI}
For a local shape set $\LS$ and label spaces $L=\{L_{\mathcal{S}} : \mathcal{S} \in \LS_{\bullet}\}$ over a field $\mathbb{K}$,
we call a linear functional $Z: S^n(L) \to \mathbb{K}$ homeomorphic invariant (HI), if
$$Z(\ell)=Z(\phi(\ell)),$$
for any vector $\ell\in S^n(L)$ and any orientation preserving \PL homeomorphism $\phi$ on $S^n$.
In this case, we call $Z$ an $S^n$ functional for short.
\end{definition}

\begin{remark}
Recall that the mapping class group of $S^n$ is trivial, so the isotopic invariance of $Z$ on $S^n(L)$ is equivalent to the homeomorphic invariance.
\end{remark}

We can extend the linear functional $Z$ to $M(L)$, $|M|\sim S^n$ by homeomorphisms.
To simplify the terminology, we first focus on the stratified manifolds with support $S^n$.
We will extend the linear functional $Z$ to $M(L)$ for a general $n$-manifold $M$ when we construct $n+1$ alterfold TQFT.

Recall that $D^{n+1}=[-1,1]^{n+1}= D^n \times [-1,1]$ and $S^n=\partial D^{n+1}$.
The equator of $S^n$ is $ S^{n-1} \times \{0\}$, which separates $S^n$ into two $n$-discs
\begin{align*}
D^n_{+}&=D^{n} \times \{1\} \cup S^{n-1} \times [0,1];\\
D^n_{-}&=D^{n} \times \{-1\} \cup S^{n-1} \times [-1,0].
\end{align*}
And $\partial D^n_{\pm}=\pm (-1)^n (S^{n-1} \times \{0\})$, where the sign $\pm (-1)^n$ indicates the orientation of the boundary.

\begin{remark}
As the stratified manifold contains more data, we would like to keep the stratified manifold as the first component in the product $\mathcal{D}^{n-k}\times D^k$. Then the new coordinate will be the last, and we obtain a global factor $(-1)^n$.

If we decompose $D^{n+1}=[-1,1]^{n+1}= [-1,1] \times [-1,1]^n$
into two part $\partial D^n_{\pm}$ according to the first coordinate. Then $\partial D^n_{\pm}=\pm (\{0\} \times S^{n-1})$.
In this case, the new coordinate will be the first. Then the stratified manifold will appear as the last component in the product $D^k \times \mathcal{D}^{n-k}$.
\end{remark}

Now let us study $L$-labelled stratified manifolds with boundary support $S^{n-1} \times \{0\}$.

\begin{definition}
Suppose $\mathcal{S}$ is an oriented stratified manifold with support $S^{n-1} \times \{0\}$.
We define
$V_{\mathcal{S},\pm}$ to be the vector space spanned by $L$-labelled stratified manifold with support $D^n_{\pm}$ and boundary $\pm(-1)^n\mathcal{S}$.
\end{definition}

\begin{definition}
The $S^n$ functional $Z$ defines a bi-linear form on $V_{\mathcal{S},-}\times V_{\mathcal{S},+}$,
$$Z(\ell_- \times \ell_+):=Z(\ell_-\cup\ell_+),$$
where $\ell_\pm \in V_{\mathcal{S},\pm}$ are $L$-labelled stratified manifolds.
\end{definition}

\begin{definition}
We call the rank of the bilinear form the entanglement rank of $Z$ over $\mathcal{S}$, denoted by  $r_Z(\mathcal{S})$.  
We say $Z$ has finite entanglement rank, if $r_Z(\mathcal{S})<\infty$ for any $\mathcal{S}$.
\end{definition}

\begin{definition}
A vector $v \in V_{\mathcal{S},\pm}$ is called a null vector, if
$$Z(v \cup w)=0, \forall w \in V_{\mathcal{S},\mp}.$$
The subspace of all null vectors are denoted by $K_{\mathcal{S},\pm}$.
Their quotient spaces are denoted by $\tilde{V}_{\mathcal{S},\pm}:=V_{\mathcal{S},\pm}/K_{\mathcal{S},\pm}$.
\end{definition}

Note that both $\dim \tilde{V}_{\mathcal{S},+}$ and $\dim \tilde{V}_{\mathcal{S},-}$ are equal to the rank of the bilinear form. When $Z$ has finite entanglement rank, $\tilde{V}_{\mathcal{S},\pm}$ are finite dimensional.

\begin{proposition}\label{Prop:phi-inv}
Suppose $\phi$ is a \PL homeomorphism on $D^n_\pm$, which is the identity on the boundary, then it induces the identity action on $\tilde{V}_{\mathcal{S},\pm}$.
\end{proposition}

\begin{proof}
For any $v \in V_{\mathcal{S},\pm}$ and $w \in V_{\mathcal{S},\mp}$.
$$Z(v\cup w)=Z(\phi(v)\cup w).$$
So $\phi(v)=v$ in $\tilde{V}_{\mathcal{S},\pm}$.
\end{proof}

\begin{corollary}\label{Cor: homeomorphism on Dk}
A homeomorphism $\phi$ on $S^{k-1}$ extends to a   homeomorphism $\Lambda\phi$ on $D^{k}$ by the one-point suspension. It extends to a homeomorphism $I_{D^{n-k}} \times \Lambda\phi$ on $\mathcal{D}^{n-k}\times D^{k}$. The corresponding action on $\tilde{V}_{\partial(\mathcal{D}^{n-k}\times D^{k}),\pm}$ is the identity map.        
\end{corollary}
   
\begin{proof}
It follows from Prop.~\ref{Prop: Alexander} and \ref{Prop:phi-inv}.  
\end{proof}

\begin{definition}
We define $\rho$ to be the $180^{\circ}$ rotation along $S^1$ of the last two coordinates.
It induces a linear map from $V_{\mathcal{S}_\pm}$ to $V_{\rho(\mathcal{S})_\mp}$.

More precisely, for an $L$-labelled stratified manifold $\ell$ with support $D^n_{\pm}$, position points $\PP$, charts $\{(U_p, \phi_p, \mathcal{S}_p): p\in \PP\}$, and vectors $\{v_p \in L_{\mathcal{S}_p} : p\in \PP\}$,
we define its rotation $\rho(\ell)$ as an $L$-labelled stratified manifold with support $D^n_{\mp}$,
\begin{enumerate}
    \item position points $\rho(\PP)$;
    \item a regular chart $(\rho(U_p), \phi_p \rho, \mathcal{S}_p)$, for every $\rho(p)\in \rho(\PP)$; 
    \item and a vector $\otimes_{p\in \PP} v_p$.
\end{enumerate}
\end{definition}

\begin{proposition}
The rotation $\rho$ is well-defined from $\tilde{V}_{\mathcal{S},\pm}$ to $\tilde{V}_{\mathcal{S},\mp}$.
\end{proposition}

\begin{proof}
If $v\in \tilde{V}_{\mathcal{S},\pm}$ is a null vector, then for any $w\in \tilde{V}_{\mathcal{S},\pm}$,
$$Z(\rho(v)\cup w)=Z(v\cup \rho(w))=0.$$
So $\rho(v)$ is a null vector.
So the rotation $\rho$ is well-defined on the quotient $\tilde{V}_{\mathcal{S},\pm}$.
\end{proof}

Suppose $M$ is an $n$-manifold with boundary $S$ and $\mathcal{S}$ is a stratified manifold with support $S$.
Recall that $M_{\mathcal{S}}(L)$ is the vector space spanned by $L$-labelled stratified manifold with support $M$ and boundary $\mathcal{S}$.

\begin{definition}
For an $S^n$ functional $Z$, we define the kernel $K_Z$ to be the vector space spanned by stratified manifolds labelled with a null vector in a local $n$-disc.
We define quotient space $\tilde{M}_{\mathcal{S}}(L):=M_{\mathcal{S}}(L)/K_Z$.
\end{definition}

By Prop.~\ref{Prop:phi-inv}, the vector in $\tilde{M}_{\mathcal{S}}(L)$ is well-defined up to isotopy in a local $n$-disc.

\begin{notation}
The two operations disjoint union and gluing preserve $K_Z$, so they are well-defined on the quotient spaces $\{\tilde{M}_{\mathcal{S}}(L)\}$, called the tensor and contraction respectively.    
\end{notation}

\section{Hyper Disc Algebras and Sphere Algebras}\label{Sec: Hyper Disc Algebras and Sphere Algebras}

Given an $S^n$ functional $Z$ on $L$-labelled stratified manifolds, the two operations disjoint union and gluing induce tensor product and contraction on the quotient spaces $\{\tilde{M}_{\mathcal{S}}(L)\}$. They provide fruitful algebraic structures according to the rich boundary conditions and the algebraic relations captured by the kernel $K_Z$.
In particular, we can glue two $n$-discs into one $n$-disc up to isotopy. We first study the $n$-category emerging from the local $n$-disc by idempotent completion.

Note that $Z$ is homeomorphic invariance on $S^n$, so for any homeomorphism $\phi$ on $\mathbb{R}^{n+1}$  and $\ell \in S^n(L)$, the value 
$$Z(\phi(\ell)):=Z(L)$$
is independent of $\phi$.   
Furthermore, if $\phi$ preserves $\mathcal{S}$, $|\mathcal{S}|=S^{n-1}\times 0$,
by Prop.~\ref{Prop:phi-inv}, we can identity $\phi(v)$ as $v$ in $\tilde{V}_{\mathcal{S},\pm}$.

\begin{notation}
Let $G$ be the group generated by affine transformations on $\mathbb{R}$ for every coordinate of $\mathbb{R}^{n+1}$.
If an element $g$ in $G$ preserves $S^n$, then it is the identity map on $S^n$.
So there is no ambiguity to identify a vector $v \in \tilde{V}_{\mathcal{S},\pm}$ with $g(v)$.
We will write $g(v)$ as $v$ to simplify the notation.
\end{notation}

\begin{notation}
We say we label $v_p$ at $p$ of a stratified manifold, if the labelled stratified manifold $\ell$ has a position point $p$, a regular chart $(U_p,\phi_p,\mathcal{S}_p)$ and a vector $v_p\in L_{\mathcal{S}_p}$, such that $\phi_p=g|_{U_p}$, for some $g\in G$, 
Diagrammatically, we draw the label $v_p$ at $p$ of the stratified manifold.
\end{notation}

\subsection{Hyper Disc Algebras}
We introduce an $n$-dimensional hyper disc algebra, or a $D^n$-algebra, to study the algebraic properties of 
$V_{\mathcal{S}_\pm}$ for all $\mathcal{S}=S^{n-1}$. It generalizes the notion of planar algebras introduced by Jones \cite{Jon21}.  
The algebraic operations come from combinations of the two elementary operations: tensor product and contraction.
They are extremely rich due to the complexity of the stratified manifolds and topological isotopy.

We remove the interior of $m$ disjoint $n$-discs in an $n$-disc, which are all \PL homeomorphic to the standard $D^n$ and denoted it as $T^n(m)$. It is unique up to \PL homeomorphism.
For example, $D_0=[-1,1]^n$, $D_i= [-1/2,1/2]^{n-1} \times [\frac{-m-2+2i}{m+1}, \frac{-m+2i}{m+1}]$, $i=1,2\cdots, m$. $T^n(m)=D_0\setminus \bigcup_{i=1}^{m} \mathring{D}_i$.
We consider $\partial D_0$ as the output boundary and $\partial D_i$ as the $i^{th}$ input boundary of $T^n(m)$. The output boundary has the same orientation as the bulk $T^n(m)$ and the input boundary has orientation opposite to the bulk.

\begin{definition}
     We define an $n$-tangle as a stratified $\mathcal{T}$ in $\mathbb{R}^n$, $|\mathcal{T}|\sim T^n(m)$, such that the 0-manifold $T_0$ is empty in the interior.
\end{definition}

\begin{remark}
From the view of TQFT, we only consider tangles up to \PL isotopy in the interior.
\end{remark}

If the output disc of $\mathcal{T}$ is identical to the $i^{th}$ input disc of $\mathcal{T}'$, then we obtain a tangle $\mathcal{T}'\circ_{i}\mathcal{T}$ by gluing the boundary.

\begin{definition}
For a local shape $LS_{\bullet}$, we consider a stratified manifold $\mathcal{S}$ with support $S^{n-1}$ as an object and an $n$-tangle as a morphism from the tensor of objects to one object. Under such compositions, all $n$-tangles form an operad of stratified \PL $n$-manifolds, denoted by $\stopn(LS_{\bullet})$, or just $\stopn$ for short.
\end{definition}

\begin{remark}
The tangle of the \PL manifold $T^n(m)$ (without stratification) is symmetric w.r.t. the inputs. The tangle of the stratified \PL manifold $\mathcal{T}$ is not symmetric.
\end{remark}

\begin{definition}
   Note $|\mathcal{T}|\sim T^n(1)$ is a morphism from one object to one object. We call it an annular $n$-tangle. These objects and morphisms form a category, called the annular category of stratified \PL $n$-manifolds.
\end{definition}

\begin{definition}
For any stratified $\mathcal{S}$, we denote its \PL homeomorphic
equivalence class as $[\mathcal{S}]:=\{ \mathcal{S}': \mathcal{S}'\sim \mathcal{S} \}$.
We denote $\hom[\mathcal{S}]$ to be the groupoid of \PL homeomorphisms between elements in $[\mathcal{S}]$, and  $\hom(\mathcal{S})$ to be the group of \PL homeomorphisms on $\mathcal{S}$.
\end{definition}

\begin{definition}
    a $D^n$-algebra $V$ is a covariant representation $\pi$ of the operad $\stopn$. That means
    \begin{enumerate}
        \item for any object in $\stopn$, namely a stratified $\mathcal{S}$, $|\mathcal{S}|\sim S^{n-1}$, there is a vector space $V_{\mathcal{S}}$;
        \item for any homeomorphism $\phi: \mathcal{S} \to \mathcal{S'}$, there is an invertible linear map $\pi(\phi)$ from $V_{\mathcal{S}}$ to $V_{\mathcal{S}'}$,
        \item for any morphism in $\stopn$, namely an $n$-tangle $\mathcal{T}$, there is a multilinear map $\pi(\mathcal{T}): \bigotimes_{i=1}^{m} V_{\mathcal{S}_i} \to  V_{\mathcal{S}_0}$;
    \end{enumerate}
    such that
    \begin{enumerate}
     \item $\pi$ is a representation of the groupoid $hom[\mathcal{S}]$;
     \item for any $n$-tangle $\mathcal{T}$ and $\mathcal{T'}$, s.t. $\mathcal{T}_0=\mathcal{T'}_i$, $$\pi(\mathcal{T}'\circ_{i}\mathcal{T})=\pi(\mathcal{T}')\circ_{i}\pi(\mathcal{T});$$
     \item The action of $\mathcal{T}$ is compatible with any orientation preserving \PL homeomorphism $\phi$ on $\mathbb{R}^n$:
     $$\pi(\phi(\mathcal{T}))(\bigotimes_{i=1}^{m} \pi(\phi|_{\mathcal{S}_i}))=
     \pi_{\phi|\mathcal{S}_0}\pi(\mathcal{T}).$$
    \end{enumerate}
\end{definition}

\begin{notation}
We will ignore the symbol $\pi$ if there is no confusion, and simply write the linear map $\pi(\mathcal{T})$ as $\mathcal{T}$,

$$\mathcal{T}: \bigotimes_{i=1}^{m} V_{\mathcal{S}_i} \to  V_{\mathcal{S}_0}.$$

\end{notation}

\begin{remark}
The compatibility between $\mathcal{T}$ and $\phi$ in condition (3) is equivalent to that
for any input vectors $v_i \in V_{\mathcal{S}_i}$,
    $$\phi(\mathcal{T}(\bigotimes_{i=1}^{n} v_i))=\phi(\mathcal{T})(\bigotimes_{i=1}^{n} \phi(v_i)).$$
\end{remark}

\begin{notation}
For $|\mathcal{S}|= \Lambda S^{n-1}$ , we represent a vector $\alpha$ in $V_{\mathcal{S}}$ as a stratified manifold $\Lambda \mathcal{S}$ which is the one-point suspension of $\mathcal{S}$ with the origin, and the origin is labelled with $\alpha$.
\end{notation}

\begin{figure}
    \centering
    \begin{tikzpicture}
    \draw  (-1,-1) rectangle (1,1);
    \draw[blue] (0,0) --(.5,1);
    \draw[blue] (0,0) --(-.5,1);
    \draw[blue] (0,0) --(0,-1);
    \node at (-.5,0) {$\alpha$};
  \begin{scope}[scale=2,shift={(2,0)}]  
  \begin{scope}[shift={(.25,.25)}]    
            \fill[gray!20] (0,0)--(1,0)--(1,-1)--(0,-1);
            \end{scope}
            \begin{scope}[shift={(0,-.5)}]
            \fill[opacity=.20] (0,0)--(1,0)--(1.5,.5)--(.5,.5);
            \end{scope}
            \begin{scope}[shift={(-.5,0)}]
                \fill[opacity=.20] (1,0)--(1.5,.5)--(1.5,-.5)--(1,-1);
            \end{scope}
            \draw[gray!60] (1,0)--(1.5,.5)--(1.5,-.5)--(1,-1)--(1,0);
            \draw[gray!60] (0,0)--(1,0)--(1.5,.5)--(.5,.5)--(0,0);
            \draw[gray!60] (0,0)--(1,0)--(1,-1)--(0,-1)--(0,0);
            \draw[gray!60,dashed] (.5,.5)--(.5,-.5);
            \draw[gray!60,dashed] (.5,-.5)--(1.5,-.5);
            \draw[gray!60,dashed] (.5,-.5)--(0,-1);
            \draw[blue] (.25,-.25)--(1.25,-.25);
            \draw[blue] (.75,-.75)--(.75,.25);
            \draw[blue] (.5,-.5)--(1,0);
            \node at (.7,-.2) {$\alpha$};
    \end{scope}
\end{tikzpicture} 
    \caption{Fig: Diagrammatic presentations of vectors for $n=2,3$.}
    \label{fig:2D 3D}
\end{figure}
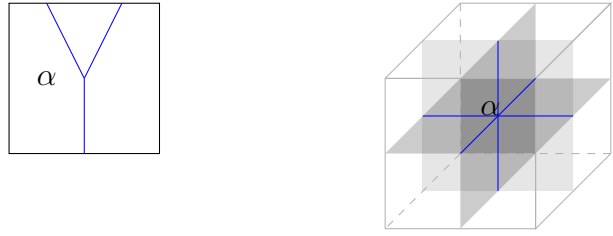

We can choose a representative $\mathcal{S}$ in $[\mathcal{S}]$ and only consider the action of its mapping class group $MCG(\mathcal{S})$.

\begin{definition}
    For a $D^n$-algebra $V$ and an $n$-tangle $\mathcal{T}$, $|\mathcal{T}|=T^n(m)$,
    we partially fill the input vectors $v_i \in V_{\mathcal{S}_i}$, for $i \in I \subset \{1,2,\cdots, m\}$, we call the result a labelled tangle which is a multilinear map on the rest inputs:
    $$\mathcal{T}(\bigotimes_{i\in I} v_i): \bigotimes_{i\notin I} V_{\mathcal{S}_i} \to  V_{\mathcal{S}_0}.$$
    If all inputs are filled by vectors, then it becomes a vector in $V$, and we call it a fully labelled tangle. 
\end{definition}

For an $n$-tangle $\mathcal{T}$ with $m$ inputs, $m\geq 2$, and the $i^{th}$ input boundary $\mathcal{S}_i$ has no stratification, we can fill the $PL$ manifold $D_i$, $D_i\sim D^n$, and we obtain an $n$-tangle with $m-1$ inputs, denoted by $\mathcal{T}(D_i)$.
For a vector $v \in V(S^{n-1})$, we can fill the $i^{th}$ input of $\mathcal{T}$ by $v_i$, denoted by $\mathcal{T}(v_i)$.

\begin{definition}
For a $D^n$-algebra $V$, a vector $v \in V(\partial=S^{n-1})$ is called the unit, if
$$\mathcal{T}(v_i)=\mathcal{T}(D_i),$$
as multilinear maps for any $\mathcal{T}$ and $i$.
In this case, we call the $D^n$-algebra $V$ unital.
\end{definition}

\begin{proposition}
    The unit is unique.
\end{proposition}

\begin{proof}
If both $v$ and $w$ are units, then $v=T^n(2)(w,v)=w$.
\end{proof}

\begin{definition}
    A subspace $W$ of a $D^n$-algebra $V$ is a set of vector spaces $\{W(\partial=\mathcal{S}) \subseteq V(\partial=\mathcal{S}) : |\mathcal{S}|=S^{n-1}\}$.
    We call $W$ a sub $D^n$-algebra of $V$, if it is invariant under the action of $n$-tangles. We call $W$ an ideal of $V$, if it is invariant under the action of $V$-labelled tangles.
\end{definition}
The subspace/subalgebra/ideal $W$ is considered to be trivial, if $W=0$ or $V$.

\begin{notation}
   For a subspace $W$ of a $D^n$ algebra $V$, we denote $A(W)$ to be the sub algebra generated by $W$ and $I(W)$ to be the ideal generated by $W$.
\end{notation}

\begin{proposition}
    The quotient $V/W$ for an ideal $W$ of a $D^n$ algebra $V$ is a $D^n$ algebra.
\end{proposition}

\begin{proof}
    It follows from the definitions.
\end{proof}

\begin{definition}
A $D^n$-algebra $V$ is called decomposible, if $V=V_1\oplus V_2$ for two sub  $D^n$-algebras $V_1$ and $V_2$.
\end{definition}

\begin{definition}
A $D^n$-algebra $V$ is called reducible, if it has a non-trivial ideal.
\end{definition}

\begin{definition}
For two $D^n$-algebras with representations (or functors) $\pi_V$ and $\pi_W$, a homomorphism $\Phi: V\to W$ of $D^n$-algebras is a natural transformation between the functors $\pi_V$ and $\pi_W$.
\end{definition}

\begin{definition}
   A $D^n$ algebra is called finite dimensional, if every vector space $V(\partial=\mathcal{S})$, $|\mathcal{S}|=S^{n-1}$, is finite dimensional.
\end{definition}

\begin{definition}
For a given local shape set $LS_{\bullet}$, all $L$-labelled stratified manifolds $S^n(L)$ form a $D^n$ algebra, called the free $D^n$ algebra.    
\end{definition}

\subsection{Hyper Sphere Algebras}

\begin{definition}
For a $D^n$ algebra $V$ and an $S^n$ functional on $S^n(V)$, we call the pair $(V,Z)$ a hyper sphere algebra or an $S^n$ algebra. An $S^n$ algebra $(V,Z)$ is called non-degenerate, if $K_Z=0$. 
\end{definition}

\begin{proposition}\label{Prop: Kernel ideal}
For an $S^n$ functional on $S^n(L)$, the kernel $K_Z$ is an ideal of the free $D^n$ algebra.
The quotient spaces $\tilde{V}:=\{\tilde{V}_{\mathcal{S},-} : |\mathcal{S}|=S^{n-1}\}$ from a $D^n$ algebra, and $(\tilde{V},Z)$ is a non-degenerate $S^n$-algebra.
\end{proposition}

\begin{proof}
If an input vector is a stratified manifold labelled by a null vector, then the output is also labelled by a null vector. So $K_Z$ is an ideal.   
So the quotient $\tilde{V}$ is a $D^n$-algebra. 
The $S^n$ functional $Z$ is zero on $K_Z$, so it is well-defined on $S^n(\tilde{V})$.
Moreover, a null vector in $V_{\mathcal{S}}$ is in $K_Z$. So $(\tilde{V},Z)$ is a non-degenerate
$S^n$-algebra.
\end{proof}

Recall that $\tilde{V}_{\mathcal{S},\pm}$ are spanned by $L$-labelled stratified manifolds modulo the kernel $K_Z$.
We can update the label space $L$ by the $D^n$ algebra $\tilde{V}=\{\tilde{V}_{\mathcal{S},-}: |\mathcal{S}|=S^{n-1}\}$.

\begin{definition}\label{Def: update label space to Dn algebra}
For a $D^n$ algebra $V$, a $V$-labelled stratified manifold $\mathcal{M}$ consists of
\begin{enumerate}
    \item a set of points $\PP$, called position points;
    \item a regular chart $(U_p, \phi_p, \mathcal{S}_p)$, for every $p\in \PP$; 
    \item and a vector $\otimes_{p\in \PP} v_p \in \otimes_{p\in \PP} V_{\mathcal{S}_p}$. 
\end{enumerate}
such that $M^0\subseteq \PP \subseteq  \mathring{\mathcal{M}}$; $\phi_p$ is orientation preserving; $\partial(\mathcal{M})$ and $U_p$, $p\in\PP$, are pairwise disjoint.
\end{definition}

\begin{proposition}\label{Prop: Update Label Space}
For the non-degenerate $S^n$ algebra $(\tilde{V},Z)$, $\tilde{V}$-labelled stratified manifold $\mathcal{M}$ with boundary $\mathcal{S}$ and support $D^n_{\pm}$ is a vector in $\tilde{V}_{\mathcal{S},\pm}$.    
\end{proposition}

\begin{proof}
We can regard $\mathcal{M}$ as an $L$-labelled stratified manifold, by replacing every $v_p \in V_{\mathcal{S}_p}$ as an $L$-labelled stratified manifold.
By Prop.~\ref{Prop: Kernel ideal},  $\mathcal{M}$ is a null vector, if a label $v_p$ is a null vector. So a $\tilde{V}$-labelled stratified manifold $\mathcal{M}$ with boundary $\mathcal{S}$ and support $D^n_{\pm}$ is a vector in $\tilde{V}_{\mathcal{S},\pm}$.    
\end{proof}

\subsection{Reflection Positivity}

\begin{definition}
Suppose * is an order-two automorphism of the field $\mathbb{K}$. For vector spaces $V$ and $W$, a map $T:V\to W$ is *-linear if
$$T(ax+by)=a^*T(x)+b^*T(y),~\forall a,b\in \mathbb{K}, ~x,y\in V.$$
\end{definition}

\begin{definition}
We define a reflection $\theta$ on an $S^n$ algebra $(V,S)$ as order-two *-linear maps $\theta: V_{\mathcal{S},\pm} \to V_{\mathcal{S},\mp}$ for every boundary stratified manifold $\mathcal{S}$ with support $S^{n-1}$, such that and
$$\theta \mathcal{T}(\bigotimes_{i\in I} v_i)=\mathcal{T}(\bigotimes_{i\in I} \theta(v_i) ),$$
for any fully labelled tangle $\displaystyle \mathcal{T}(\bigotimes_{i\in I} v_i)$.
\end{definition}

\begin{remark}
For an $S^n$ functional $Z$, the reflection $\theta$ on the $S^n$ algebras $(V,S)$ (or $(\tilde{V},S)$ is determined by the action of the label space $L$. 
We may add the $180^\circ$ rotation $\rho$ of $\theta(L)$ to the label space $L$ for convenience.
\end{remark}

\begin{definition}\label{Def: Z Hermitian}
We call an $S^n$ functional $Z$ Hermitian w.r.t. a reflection on the $S^n$-algebra $(V,Z)$, if for any $v_{\pm} \in V_{\mathcal{S},\pm}$,
\begin{align}\label{Equ: Z Hermitian}
Z(v_-\cup v_+)^*&=Z(\theta(v_+)\cup \theta v_-).    
\end{align}

\end{definition}

\begin{proposition}
If an $S^n$ functional $Z$ Hermitian, then the reflection is well-defined on the quotient spaces $\tilde{V}_{\mathcal{S},\pm}$.   
\end{proposition}
\begin{proof}
For any null vector $v_\pm \in V_{\mathcal{S},\pm}$, by Equ.~\ref{Equ: Z Hermitian}, we have $\theta(v_\pm)$ is a null vector. So the reflection $\theta$ is well-defined on the quotient spaces $\tilde{V}_{\mathcal{S},\pm}$.    
\end{proof}

\begin{definition}\label{Def: Z RP}
A Hermitian $S^n$ functional $Z$ is called reflection positive (RP) over a number field $\mathbb{K}$, if 
$$Z(v\cup \theta (v))\geq 0, ~\forall~ v \in V_{\mathcal{S},\pm}.$$
\end{definition}

Usually the involution $*$ on $\mathbb{C}$ is the complex conjugate. When $Z$ is reflection positive, the quotient spaces $\tilde{V}_{\mathcal{S},\pm}$ are (pre) Hilbert spaces, which are called the physical Hilbert space in quantum field theory.    
In particular, $v$ is a null vector iff
$Z(v \cup \theta(v))=0$, by Cauchy-Schwarz inequality.

In Def.~\ref{Def: update label space to Dn algebra}, we update the label space from $L$ to $\tilde{V}_{\mathcal{S},-}$ of the $D^n$ algebra $\tilde{V}$ for orientation preserving homeomorphisms in regular charts. 
When $Z$ is Hermitian, we can further extend the label space to $\tilde{V}_{\mathcal{S},\pm}$. A stratified manifold may have a label $\phi^{-1}(v)$ for orientation reversing homeomorphism $\phi$ in a regular chart $(U, \phi, \mathcal{S})$ and a vector $v\in \tilde{V}_{\mathcal{S},+}$.

\begin{definition}
For an $S^n$ algebra $(\tilde{V},Z)$, a $\tilde{V}$-labelled stratified manifold $\mathcal{M}$ consists of
\begin{enumerate}
    \item a set of points $\PP$, called position points;
    \item a regular chart $(U_p, \phi_p, \mathcal{S}_p)$, for every $p\in \PP$; 
    \item and a vector $\otimes_{p\in \PP} v_p \in \otimes_{p\in \PP} \tilde{V}_{\mathcal{S}_p,\pm}$. 
\end{enumerate}
such that $M^0\subseteq \PP \subseteq  \mathring{\mathcal{M}}$; the $\pm$ sign of $\tilde{V}_{\mathcal{S}_p,\pm}$ depends on the orientation of $\phi_p$; $\partial(\mathcal{M})$ and $U_p$, $p\in\PP$, are pairwise disjoint.
\end{definition}

We generalize Prop.~\ref{Prop:phi-inv} for orientation reversing \PL homeomorphisms as follows.
\begin{proposition}
Suppose $\phi$ is an orientation reversing \PL homeomorphism from $D^n_{\pm}$ to $D^n_{\mp}$, which is the identity on the boundary, then $\phi=\theta$ as an action on $\tilde{V}_{\mathcal{S},\pm}$.
\end{proposition}
\begin{proof}
By Prop.~\ref{Prop:phi-inv}, both $\phi\theta$ and $\theta^2$ are the identity on $\tilde{V}_{\mathcal{S},\pm}$. So $\phi=\theta$.
\end{proof}

\section{Idempotent Completion and Spherical n-category}\label{Sec: Idempotent Completion and Spherical n-category}

In this paper, we study the higher representation theory of the $D^n$-algebra and construct an $n$-category. We consider it such an $n$-category from an $S^n$ function as a spherical $n$-category. If the  $S^n$ function us reflection positive, then  consider it as a unitary $n$-category.  

\begin{definition}
From the $S^n$ functional $Z$, we obtain a non-degenerate $S^n$-algebra $(\tilde{V},Z)$.
We call it finite dimensional, if $\tilde{V}_{\mathcal{S},-}$ is finite dimensional for any $\mathcal{S}$.
\end{definition}
In this paper, we focus on finite dimensional $S^n$ algebras.
A finite dimensional semisimple algebra over a field $\mathbb{K}$ is a direct sum of matrix algebras, $\bigoplus_i M_{n_i}(\mathbb{K})$.
If it is commutative, then $n_i=1$.
It is well-known that a finite dimensional $C^*$ algebra is semisimple.

We show that for any $|\mathcal{D}^{n-1}|=D^{n-1}$, $\tilde{V}_{\partial(\mathcal{D}^{n-1} \times D^1)}$ forms an associative algebra under the composition along the last coordinate.
If $Z$ is reflection positive, then the algebra is a $C^*$-algebra. 
Over a general field $\mathbb{K}$, we assume that the algebra is semisimple. 
We consider a (minimal) idempotent of the algebra as an (indecomposible) representation.
We will study the idempotent completion of these semisimple algebras and study the $n$-category structure of these idempotents as a higher representation category.

\subsection{Idempotent Completion}
In the following, we assume that $D^n$ and $D^n_-$ have the same orientation. (We ignore the global sign $(-1)^{n+1}$ of the orientation for simplicity.) 
Suppose $V_{\mathcal{S}}$ is the vector space spanned by $L$-labelled stratified manifold with support $D^n$ and the same boundary orientation as $V_{\mathcal{S},-}$. 
For a vector $x\in V_{\mathcal{S}}$,
we construct the vector $x_{-}$ in $V_{\mathcal{S},-}$ as
\begin{align*}
x_{-}&:=x\times \{-1\} \cup \mathcal{S} \times [-1,0].
\end{align*}

We study the algebraic structure of $\tilde{V}_{\mathcal{S},-}$ by that of the $D^n$ algebra $V_{\mathcal{S}}$. We rescale two $D^n$ and then compose them to one $D^n$ along the last coordinate, and consider their composition as a multiplication. We first study its idempotent completion which encodes an $n$-category.

For an $L$-labelled stratified manifold $\mathcal{S}^{n}$ with support $S^{n}$, it is a scalar according to $Z$.
For an $L$-labelled stratified manifold $\mathcal{D}^{n}$ with support $D^{n}$, it is a vector.
For a stratified manifold $\mathcal{S}^{n-1}$ with support $S^{n-1}$, it corresponds to a vector space $V_{\mathcal{S}^{n-1}}$.
We will explain the categorical meaning of stratified manifolds with support $D^{k}$ and $S^{k-1}$ for all $k$.

Note that when $n=1$, the vector space $\tilde{V}_{S^0,-}$ forms an associative algebra with a trace from $Z$. For $x,y \in V_{S^0}$, we label $x$ at $[-1,0]$ and label $y$ at $[0,1]$.
More precisely, the position point $-1/2$ has a vector $x$ with a regular chart $(U,\phi,\mathcal{S})$, $U=[-1,0]$, $\phi:U\to D^1$, $\mathcal{S}=S^0$.

We consider the result as their multiplication, denoted by $xy$, which is a vector in $V_{S^0}$.
Moreover, $V_{S^0}$ is an associative algebra.
The associativity follows from the following \PL homeomorphism $\phi$ on $[-1,1]$:
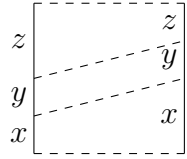
\begin{center}
\begin{tikzpicture}
\draw (-1,-1) --  (-1,1);
\node at (-1.2,-.75) {$x$};
\node at (-1.2,-.25) {$y$};
\node at (-1.2,.5) {$z$};
\begin{scope}[shift={(2,0)}]
\draw (-1,-1) --  (-1,1);
\node at (-1.2,-.5) {$x$};
\node at (-1.2,.25) {$y$};
\node at (-1.2,.75) {$z$};
\end{scope}
\draw[dashed] (-1,-1) -- (1,-1);
\draw[dashed] (-1,-.5) -- (1,0);
\draw[dashed] (-1,0) -- (1,.5);
\draw[dashed] (-1,1) -- (1,1);
\end{tikzpicture}
\end{center}
Note that $D^1$ is the identity $I$ of $V_{S^0}$.
We define the trace of $x$ as
$$Tr(x):=Z(x_- \cup \rho(D^1_-)).$$
By \PL isotopy, we have 
$$Tr(xy)=Z(x_- \cup \rho(y_-))=Tr(yx).$$

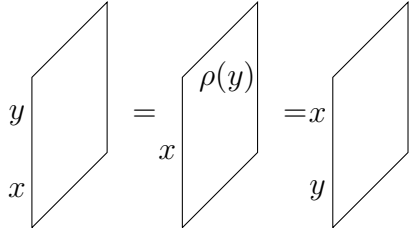
\begin{center}
\begin{tikzpicture}
\draw (-1,-1) -- (-1,1) -- (0,2)-- (0,0)-- (-1,-1);
\node at (-1.2,-.5) {$x$};
\node at (-1.2,.5) {$y$};
\node at (0.5,.5) {$=$};
\node at (2.5,.5) {$=$};
\begin{scope}[shift={(4,0)}]
\draw (-1,-1) -- (-1,1) -- (0,2)-- (0,0)-- (-1,-1);
\node at (-1.2,-.5) {$y$};
\node at (-1.2,.5) {$x$};
\end{scope}
\begin{scope}[shift={(2,0)}]
\draw (-1,-1) -- (-1,1) -- (0,2)-- (0,0)-- (-1,-1);
\node at (-1.2,0) {$x$};
\node at (-.4,1) {$\rho(y)$};
\end{scope}
\end{tikzpicture}
\end{center}

In the rest, we draw the last second coordinate as the the vertical one, and consider it as the multiplication direction.

If $x_-$ is a null vector, then $Tr(x)=0$ and for any $a,a',y \in V_{S^0}$,
$$Z((axa')_-\cup \rho(y_-))=Z(x \cup \rho(a_-)\rho(y_-)\rho(a'_-))=0.$$
So $(axa')_-$ is a null vector. So the multiplication and the trace are well-defined on $\tilde{V}_{S^0,-}$.

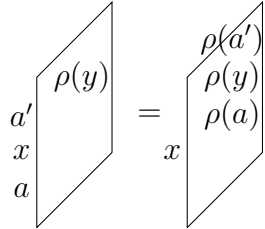
\begin{center}
\begin{tikzpicture}
\draw (-1,-1) -- (-1,1) -- (0,2)-- (0,0)-- (-1,-1);
\node at (-1.2,-.5) {$a$};
\node at (-1.2,.5) {$a'$};
\node at (-1.2,0) {$x$};
\node at (-.4,1) {$\rho(y)$};
\node at (0.5,.5) {$=$};
\begin{scope}[shift={(2,0)}]
\draw (-1,-1) -- (-1,1) -- (0,2)-- (0,0)-- (-1,-1);
\node at (-.4,.5) {$\rho(a)$};
\node at (-.4,1.5) {$\rho(a')$};
\node at (-1.2,0) {$x$};
\node at (-.4,1) {$\rho(y)$};
\end{scope}
\end{tikzpicture}
\end{center}

Now let us study the algebra for a general $n$.
For a stratified manifold $\mathcal{D}^{n-1}$ with support $D^{n-1}$, we construct
$$\mathcal{S}=\partial(\mathcal{D}^{n-1} \times D^1)=\partial\mathcal{D}^{n-1} \times D^{1} \cup \mathcal{D}^{n-1} \times S^{0}.$$
Then $|\mathcal{S}|=S^{n-1}$. The vector space $V_{\mathcal{S},-}$ forms an algebra in the following sense.

\begin{definition}
For two vectors $x$, $y$ in $V_{\mathcal{S}}$, we label $x$ at $D^{n-1}\times [-1,0]$ and $y$ into $D^{n-1}\times [0,1]$ after scaling their last coordinates by $1/2$. We define the result as the their multiplication, denoted by $xy$. Note that the boundary of $xy$ is $\mathcal{S}$, so $xy$ is a vector in $V_{\mathcal{S}}$.    
\end{definition}

\begin{definition}
We denote $\rho$ to be the $180^{\circ}$ rotation of the last second coordinate.
Then $$\rho(xy)=\rho(y)\rho(x), ~\forall~x,y \in \tilde{V}_{\mathcal{S},\mp}.$$
\end{definition}

\begin{definition}
For any $x\in V_{\mathcal{S}}$ above, we define its trace as
$$Tr(x):=Z(x_-\cup \rho((\mathcal{D}^{n-1}\times D^1)_+)).$$
\end{definition}

\begin{proposition}\label{Prop: Trace}
Then for any $x,y \in V_{\mathcal{S}}$, we have
$$Tr(xy)=Z(x_-\cup \rho(y_-))=Tr(yx).$$
\end{proposition}
\begin{proof}

It follows from \PL isotopy around $S^1$, similar to the case $n=1$.
\begin{center}
\begin{tikzpicture}
\node at (2.5,0) {$=$};
\node at (7,0) {$=$};
\draw (-1,-1) -- (-1,1) -- (1,1) -- (1,-1) -- (-1,-1);
\draw[dashed] (-1,-1) --++(.8,.8);
\draw (1,-1) --++(.8,.8);
\draw (-1,1) --++(.8,.8);
\draw (1,1) --++(.8,.8);
\draw[dashed] (-.2,-.2) --++(0,2);
\draw[dashed] (-.2,-.2) --++(2,0);
\draw (1.8,1.8) --++(0,-2);
\draw (1.8,1.8) --++(-2,0);
\node at (0,.5) {$y$};
\node at (0,-.5) {$x$};
\node at (-.8,1.5) {$[-1,1]$};
\node at (-1.3,0) {$D^1$};
\node at (0,-1.3) {$\mathcal{D}^{n-1}$};
\begin{scope}[shift={(4.5,0)}]
\draw (-1,-1) -- (-1,1) -- (1,1) -- (1,-1) -- (-1,-1);
\draw[dashed] (-1,-1) --++(.8,.8);
\draw (1,-1) --++(.8,.8);
\draw (-1,1) --++(.8,.8);
\draw (1,1) --++(.8,.8);
\draw[dashed] (-.2,-.2) --++(0,2);
\draw[dashed] (-.2,-.2) --++(2,0);
\draw (1.8,1.8) --++(0,-2);
\draw (1.8,1.8) --++(-2,0);
\node at (.8,.8) {$\rho(y)$};
\node at (0,0) {$x$};
\node at (-.8,1.5) {$[-1,1]$};
\node at (-1.3,0) {$D^1$};
\node at (0,-1.3) {$\mathcal{D}^{n-1}$};
\end{scope}
\begin{scope}[shift={(9,0)}]
\draw (-1,-1) -- (-1,1) -- (1,1) -- (1,-1) -- (-1,-1);
\draw[dashed] (-1,-1) --++(.8,.8);
\draw (1,-1) --++(.8,.8);
\draw (-1,1) --++(.8,.8);
\draw (1,1) --++(.8,.8);
\draw[dashed] (-.2,-.2) --++(0,2);
\draw[dashed] (-.2,-.2) --++(2,0);
\draw (1.8,1.8) --++(0,-2);
\draw (1.8,1.8) --++(-2,0);
\node at (0,.5) {$x$};
\node at (0,-.5) {$y$};
\node at (-.8,1.5) {$[-1,1]$};
\node at (-1.3,0) {$D^1$};
\node at (0,-1.3) {$\mathcal{D}^{n-1}$};
\end{scope}
\end{tikzpicture}
\end{center}
\end{proof}

We keep the 3D pictorial convention for the $n+1$ coordinates of $D^{n+1}$.
The first $n-1$ coordinates are horizontal.
The last second coordinate is the vertical multiplication.

\begin{proposition}\label{Prop: D1 algebra}
For any $|\mathcal{D}^{n-1}|=D^{n-1}$, the above multiplication and trace are well-defined on $\tilde{V}_{\mathcal{S},-}$, $\mathcal{S}=\partial \mathcal{D}^{n-1}\times D^1$.
Moreover, $\tilde{V}_{\mathcal{S},-}$ forms an associative algebra, denoted by $A(\mathcal{D}^{n-1}\times D^1)$. Its identity is $\mathcal{D}^{n-1}\times D^1$.
\end{proposition}
\begin{proof}
Similar to the case $n=1$, if $x$ is a null vector in $\tilde{V}_{\mathcal{S},-}$, then $Tr(x_-)=0$ and $(axb)_-$ is a null vector, by the \PL isotopy invariance of $Z$. So the trace and the multiplication are well-defined on $\tilde{V}_{\mathcal{S},-}$.
The associativity is similar to the case $n=1$.
\end{proof}

\begin{definition}
We call the non-degenerate $S^n$ algebra $(\tilde{V},Z)$ or the $S^n$ functional $Z$ to be semisimple, if $\tilde{V}_{\mathcal{S},-}$ is semisimple for any $|\mathcal{D}^{n-1}|=D^{n-1}$.
\end{definition}

\begin{proposition}
For $2 \leq k \leq n$, suppose $\mathcal{D}^{n-k}$ is a stratified manifold with support $D^{n-k}$. We construct
$$\mathcal{S}=\partial(\mathcal{D}^{n-k}\times D^{k})=\partial\mathcal{D}^{n-k} \times D^{k}   \cup \mathcal{D}^{n-k} \times S^{k-1}.$$
Then $\tilde{V}_{\mathcal{S},-}$ is a commutative associate algebra, denoted by $A(\mathcal{D}^{n-k}\times D^k)$, with unit $\mathcal{D}^{n-k}\times D^k$.
\end{proposition}

By Corollary~\ref{Cor: homeomorphism on Dk}, it does not matter which direction in the second component $D^k$ is for the multiplication, as its boundary $\partial D^k$ has no stratification.  

\begin{proof}
We consider $\mathcal{D}^{n-k}\times D^{k-1}$ as $\mathcal{D}^{n-1}$, then by Prop.~\ref{Prop: D1 algebra}, $\tilde{V}_{\mathcal{S},-}$ is an associate algebra.
The commutativity for $k\geq 2$ follows from switching two smaller $k$-discs in $D^k$ by \PL isotopy in Fig.~\ref{fig: commutativity}.
\end{proof}

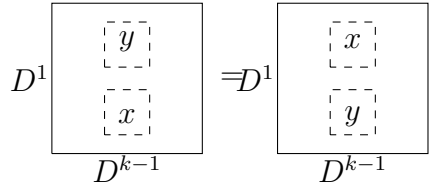
\begin{figure}
    \centering
    \begin{tikzpicture}
\draw (-1,-1) -- (-1,1) -- (1,1) -- (1,-1) -- (-1,-1);
\node at (0,-.5) {$x$};
\node at (0,.5) {$y$};
\node at (1.4,0) {$=$};
\node at (-1.3,0) {$D^1$};
\node at (0,-1.2) {$D^{k-1}$};
\begin{scope}[scale=.3,shift={(0,1.5)}]
\draw[dashed] (-1,-1) -- (-1,1) -- (1,1) -- (1,-1) -- (-1,-1);
\end{scope}
\begin{scope}[scale=.3,shift={(0,-1.5)}]
\draw[dashed] (-1,-1) -- (-1,1) -- (1,1) -- (1,-1) -- (-1,-1);
\end{scope}
\begin{scope}[shift={(3,0)}]
\draw (-1,-1) -- (-1,1) -- (1,1) -- (1,-1) -- (-1,-1);
\node at (0,-.5) {$y$};
\node at (0,.5) {$x$};
\node at (-1.3,0) {$D^1$};
\node at (0,-1.2) {$D^{k-1}$};
\begin{scope}[scale=.3,shift={(0,1.5)}]
\draw[dashed] (-1,-1) -- (-1,1) -- (1,1) -- (1,-1) -- (-1,-1);
\end{scope}
\begin{scope}[scale=.3,shift={(0,-1.5)}]
\draw[dashed] (-1,-1) -- (-1,1) -- (1,1) -- (1,-1) -- (-1,-1);
\end{scope}
\end{scope}
\end{tikzpicture}
    \caption{Fig: Commutativity}
    \label{fig: commutativity}
\end{figure}
This commutative multiplication is usually considered as a commutative Frobenius algebra in TQFT.

\begin{remark}
In general, if we fix the first component $\mathcal{D}^{n-k}$ and change the second component $D^k$ to all possible $\mathcal{D}^{k}$,
then the boundary is
$$\mathcal{S}=\partial(\mathcal{D}^{n-k}\times \mathcal{D}^{k})=\partial\mathcal{D}^{n-k} \times \mathcal{D}^{k}   \cup \mathcal{D}^{n-k} \times \partial \mathcal{D}^{k}.$$
These vector spaces form a $D^k$ algebra.
\end{remark}

\begin{definition}
For $0 \leq k \leq n-1$, $|\mathcal{D}^{k}|=D^{k}$, and
we call a minimal idempotent $\alpha$ of $A(\mathcal{D}^{k}\times D^{n-k})$
an indecomposible $k$-morphism of type $\mathcal{D}^k$.
\end{definition}

\begin{definition}
Suppose $\phi$ is an homeomorphism on $D^{k}$.
Then for any $k$-morphism of type $\mathcal{D}^k$, $\phi(\alpha)$ is a $k$-morphism of type $\phi(\mathcal{D}^k)$.
We call $\alpha$ and $\phi(\alpha)$ homeomorphic equivalent.
If $\phi$ fixes the boundary, then we call $\alpha$ and $\phi_{\alpha}$ interior homeomorphic equivalent.
\end{definition}

In principle, we are only interested in $k$-morphisms up to (interior) homeomorphic equivalence, as the $S^n$ functional is homeomorphic invariant.

\begin{definition}
For a $k$-morphism $\alpha$ of type $\mathcal{D}^k$, we define its quantum dimension as $Tr(\alpha)$.
\end{definition}

\begin{proposition}\label{Prop: non-zero trace}
If $A(\mathcal{D}^k\times D^{n-k})$ is semisimple, then for any indecomposible $k$-morphism $\alpha$ of type $\mathcal{D}^k$, we have $Tr(\alpha)\neq 0.$
\end{proposition}

\begin{proof}
Assume that $Tr(\alpha)=0$. By the semisimplicity,
for any $x \in A(\mathcal{D}^k\times D^{n-k})$, $\alpha x\alpha=c\alpha$ for some $c\in \mathbb{K}$. So
$Tr(x\alpha)=Tr(\alpha x\alpha)=Tr(c\alpha)=0$,
By Prop.~\ref{Prop: Trace}, $\alpha$ is a null vector, which is a contradiction. So $Tr(\alpha)\neq 0.$
\end{proof}

\begin{definition}
Suppose $\theta$ is a reflection. Then $\rho\theta=\theta\rho$ on $V_{\mathcal{S}_\pm}$.
\end{definition}

\begin{proof}
Under the action of $\rho\theta$, the chart $(U_p, \phi_p, \mathcal{S}_p)$ and vector $v_p$ of a labelled stratified manifold becomes $(\rho\theta(U_p), \phi_p\theta\rho, \rho\theta\mathcal{S}_p)$ and $\theta(v_p)$.
Under the action of $\theta\rho$, they become $(\theta\rho(U_p), \phi_p\rho\theta, \theta\rho\mathcal{S}_p)$ and $\theta(v_p)$.
The identity $\rho\theta=\theta\rho$ is a basic geometric fact that the vertical reflection is the composition of the $180^{\circ}$ rotation and the horizontal reflection.
\end{proof}

\begin{definition}
Suppose $\theta$ is a reflection. For any $x \in \tilde{V}_{\mathcal{S}_-}$, we define its adjoint as
$$x^*:=\rho\theta(x).$$ 
\end{definition}

Then $Tr(x^*)=Tr(x)^*$. So the adjoint is well-defined from $\tilde{V}_{\mathcal{S},-}$ to $\tilde{V}_{\mathcal{S}^*,-}$, where $\mathcal{S}^*$ is the vertical reflection of $\mathcal{S}$.
Then $A(\mathcal{D}^{n-k}\times D^k)$ is a *-algebra with the adjoint *.

\begin{proposition}
If further $Z$ is reflection positive and $A(\mathcal{D}^{n-k}\times D^k)$ is a finite dimensional, then it is a $C^*$ algebra.
\end{proposition}
\begin{proof}
If further $Z$ is reflection positive, then
$$Tr(xx^*)=Z(x_- \cup \rho(x^*_-))=Z(x_- \cup \theta (x_-))\geq 0.$$
So $Tr(xx^*)=0$ iff $x_-=0$. So $Tr$ is positive definite on $A(\mathcal{D}^{n-k}\times D^k)$.
If further $Z$ is finite dimensional, then $A(\mathcal{D}^{n-k}\times D^k)$ is a finite dimensional $C^*$ algebra. In particular, $A(\mathcal{D}^{n-k}\times D^k)$ is semisimple.    
\end{proof}

\subsection{Bimodules and Morita equivalence}

\begin{definition}
Suppose  $\alpha_\pm$ are $k$-morphisms of type $\mathcal{D}_\pm$ respectively, and they have the same boundary $\mathcal{S}$. Suppose $\mathcal{D}$ has support $D^{k+1}$ with boundary $\mathcal{S}^k=(\mathcal{D}_- \times \{-1\}\cup \mathcal{D}_+ \times \{1\} \cup \mathcal{S}\times D^1)$. Suppose $\beta$ is an idempotent of $A(\mathcal{D}\times D^{n-k-1})$.
We call $\beta$ an $\alpha-\alpha$ bimodule, if $P(\alpha_-) \beta Q(\alpha_+)=\beta$,
where $P(\alpha_-)$ is $\mathcal{D}\times D^{n-k-1}$ labelled by $\alpha_-$ at $\mathcal{D}_-\times [-1,-1+\varepsilon] \times D^{n-k-1}$ and  $Q(\alpha_+)$ is $\mathcal{D}\times D^{n-k-1}$ labelled by $\alpha_+$ at $\mathcal{D}_+\times [1-\varepsilon, 1] \times D^{n-k-1}$.
When $\beta$ is a minimal idempotent, we call it an indecomposible $\alpha-\alpha$ bimodule, and a $(k+1)$ morphism from $\alpha_-$ to $\alpha_+$.
\end{definition}
By isotopy, the three idempotents $P(\alpha_-)$, $\beta$, $Q(\alpha_+)$ commute. 
The equality $P(\alpha_-) \beta Q(\alpha_+)=\beta$ is equivalent to that $\beta$ is a sub idempotent of $P(\alpha_-)$ and $Q(\alpha_+)$.

\begin{definition}
Suppose $\alpha$ is a $k$-morphism of type $\mathcal{D}^k$. We can consider it as an idempotent in  $A(\mathcal{D}^k\times D^1)\times D^{n-k-1})$. It is a $\alpha-\alpha$ bimodule, called the trivial bimodule.
\end{definition}

\begin{definition}
Applying the $180^\circ$ rotation to the $k+1$ and $k+2$ coordinates, we obtain a $\alpha_+-\alpha_-$ bimodule, called the dual bimodule of $\beta$, denoted by $\overline{\beta}$. 
\end{definition}

\begin{definition}\label{Def: composition of morphisms}
Suppose $\alpha_i$ are $k$-morphisms $i=0,1,2$ and $\beta_i$ are $\alpha_i-\alpha_{i+1}$ bimodules,
We label $\beta_0$ at $D^{k+1}_-\times D^{n-k+1}$ and  $\beta_1$ at $D^{k+1}_+\times D^{n-k+1}$, and the result is a $\alpha_0-\alpha_2$ bimodule, called the fusion of $\beta_0$ and $\beta_1$ at $\alpha_1$, denoted by $\beta_0\otimes_{\alpha_1} \beta_1$.  
\end{definition}
The fusion could be considered as a composition of $(k+1)$-morphisms at the boundary $\alpha_1$.
The fusion is well defined modulo interior homeomorphic equivalence.

\begin{definition}
Two indecomposible $k$-morphisms are called Morita equivalent in the $D^n$ algebra, if they have a bimodule in the $D^n$ algebra.    
\end{definition} 

\begin{proposition}\label{Prop: Morita Equivalence}
The Morita equivalence is an equivalence relation.  
\end{proposition}

\begin{proof}
The Morita equivalence has the three properties, Reflexivity, Symmetry and Transitivity as discussed above. So it is an equivalence relation.  
\end{proof}

\begin{definition}
The global dimension of an indecomposible $(n-1)$-morphism is defined to be 1 (assuming semisimplicity).
For an indecomposible $k$-morphism $\alpha$ of type $\mathcal{D}^k$, take $\mathcal{S}^k= (\mathcal{D}^k_-\cup D^k_+)$.
The global dimension of $\alpha$ is defined inductively for $k=n-1,n-2, \cdots,0$ as
\begin{align}\label{Equ: global dimension}
\mu(\alpha)=\sum_{\beta} \frac{Tr(\beta)^2}{\mu(\beta)},    
\end{align}
summing over indecomposible $\alpha-\alpha$ bimodules $\beta$, $\beta$ is a representative in its annular equivalence class.
\end{definition}
To ensure it is well-defined, we need to assume $\mu(\beta)\neq 0$ inductively.

\begin{definition}
For a semisimple $S^n$ functional $Z$, we call it strong semisimple, if the global dimension of any $k$-morphism is nonzero, for any $0\leq k\leq n-1$.
\end{definition}

\begin{remark}
The non-zero global dimension condition is necessary when we construct $n+1$ TQFT later. If the condition only holds partially, then we can only construct a TQFT partially.    
\end{remark}

\begin{remark}
Etingof, Nikshych and Ostrik proved that the global dimension of a fusion category is positive over $\mathbb{C}$ in Theorem 2.3 in \cite{ENO05}.
That means the global dimension of the 0-morphism is non-zero. It will be interesting to see whether the strong semisimple condition can be derived from weaker conditions in higher dimensions.  
\end{remark}

\begin{definition}
We call $\frac{Tr(\beta)^2}{\mu(\beta)}$ the intrinsic dimension of the indecomposible $k$-morphism $\beta$.   
\end{definition}

Neither the quantum dimension $Tr(\beta)$ nor the global dimension $\mu(\beta)$ is invariant under Morita equivalence. 
We will prove that the intrinsic dimension $\frac{Tr(\beta)^2}{\mu(\beta)}$ is invariant under Morita equivalence in Theorem \ref{Thm: invariance of the intrinsic dimension}.

\subsection{Annular Equivalence}

\begin{definition}
Suppose $|\mathcal{S}|=S^{k-1}$, $|\mathcal{D}_i|=D^k$ and $\partial(\mathcal{D}_i)=\mathcal{S}$, for $i=0,1$.
We define $V(\mathcal{D}_{1},\mathcal{D}_{0})$ as a vector space spanned by $L$-labelled stratified manifolds with support $D^k\times (D^{n-k} \setminus \frac{1}{2}D^{n-k})$and boundary
$$\mathcal{S}\times (D^{n-k} \setminus \frac{1}{2} D^{n-k}) \cup  \mathcal{D}_0 \times S^{n-k-1} \cup  \mathcal{D}_1 \times \frac{1}{2} S^{n-k-1}.$$
\end{definition}
Here the scalar $\frac{1}{2}$ in front of $D^{k}$ and $S^{n-k-1}$ means scaling their radius by $\frac{1}{2}$, and
$D^{n-k} \setminus \frac{1}{2}D^{n-k}$ is removing the interior of $\frac{1}{2} D^{n-k}$ from $D^{n-k}$.

\begin{definition}
Let us define the multiplication
$$V(\mathcal{D}_{1},\mathcal{D}_{0}) \times  V_{\partial(\mathcal{D}_1\times D^{n-k})} \to V_{\partial(\mathcal{D}_0\times D^{n-k})}.$$
For a vector $T \in V(\mathcal{D}_{1},\mathcal{D}_{0})$ and a vector $x \in V_{\partial(\mathcal{D}_1\times D^{n-k})}$,
we scale $x$ by $1/2$ and then put it inside $T$, and denoted the result as their multiplication $Tx$. Then $Tx \in V_{\partial(\mathcal{D}_0\times D^{n-k})}$.
\end{definition}

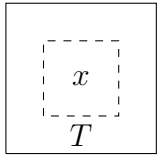
\begin{center}
\begin{tikzpicture}
\draw (-1,-1) -- (-1,1) -- (1,1) -- (1,-1) -- (-1,-1);
\node at (0,0) {$x$};
\node at (0,-.75) {$T$};
\begin{scope}[scale=.5]
\draw[dashed] (-1,-1) -- (-1,1) -- (1,1) -- (1,-1) -- (-1,-1);
\end{scope}
\end{tikzpicture}
\end{center}

If $x$ is a null vector, then $Ax$ is a null vector. So the multiplication is well defined for
$$V(\mathcal{D}_{1},\mathcal{D}_{0}) \times  \tilde{V}_{\partial(\mathcal{D}_1\times D^{n-k}),-} \to \tilde{V}_{\partial(\mathcal{D}_0\times D^{n-k}),-}.$$
If $T$ contains a null vector in its local region, then $Ax=0$.
For any interior \PL homeomorphism $\phi$ on $D^k\times (D^{n-k} \setminus \frac{1}{2}D^{n-k})$), $Tx=\phi(T)x$ in $\tilde{V}_{\partial(\mathcal{D}_0\times D^{n-k}),-}$.

\begin{definition}
We define $A(\mathcal{D}_{1},\mathcal{D}_{0})$ as $V(\mathcal{D}_{1},\mathcal{D}_{0})$ modulo local null vectors.
Then the multiplication is well defined for
$$A(\mathcal{D}_{1},\mathcal{D}_{0}) \times \tilde{V}_{\partial(\mathcal{D}_1\times D^{n-k}),-} \to \tilde{V}_{\partial(\mathcal{D}_0\times D^{n-k}),-}.$$
\end{definition}

\begin{definition}
Suppose $|\mathcal{S}|=S^{k-1}$, $|\mathcal{D}_i|=D^k$ and $\partial(\mathcal{D}_i)=\mathcal{S}$, for $i=0,1,2$.
For $T_1\in A(\mathcal{D}_{1},\mathcal{D}_{0})$ and $T_2 \in A(\mathcal{D}_{2},\mathcal{D}_{1})$, we define their multiplication $T_2T_1$ by composing them along the radius of $D^{n-k}$.
Then the multiplication is associate and
$A(\mathcal{D}_{\bullet},\mathcal{D}_{\bullet})$ forms an algebroid, which we call the $S^{n-k-1}$ annular algebroid with boundary $\mathcal{S}^{k-1}$.
\end{definition}

\begin{definition}
For indecomposible $k$-morphisms $\alpha_i$ of type $\mathcal{D}_{i}$, $i=0,1$ and common boundary $\mathcal{S}^{k-1}$, we call them annular equivalent, if there are $T_1\in A(\mathcal{D}_{1},\mathcal{D}_{0})$ and $T_2 \in A(\mathcal{D}_{0},\mathcal{D}_{1})$, such that
$$T_2\alpha_0=\alpha_1, \quad T_1\alpha_1=\alpha_0.$$
\end{definition}

\begin{proposition}\label{Prop: Morita and annular}
For $0\leq k \leq n-1$,
two indecomposible $k$-morphisms $\alpha_0$ and $\alpha_1$ are annular equivalent iff they are Morita equivalent.   
\end{proposition}

\begin{proof}
If $T(\alpha_1)=\alpha_0$ for an annular action $T$, then the normal microbundle of $T$ at $D^k\times [-1,-\frac{1}{2}]\times \varepsilon D^{n-k-1}$ is a $\alpha_0-\alpha_1$ bimodule. So they are Morita equivalent.

Conversely, suppose $\beta$ is a $\alpha_0-\alpha_1$ bimodule of type $\mathcal{D}^{k+1}$. 
Suppose $\phi$ is a homeomorphism from $\mathcal{D}^{k} \times D^{n-k}\setminus \frac{1}{2}D^{n-k}$ to the normal microbundle $U$ of $\mathcal{D}^{k} \times S^{n-k-1}$. We label $U$ by $\beta$ at the microbundle of $\mathcal{D}^{k}\times \varepsilon D^1 \times O_{n-k-1}$, and denoted the result by $T$. Take $T_{\beta}=\phi^{-1}(T)$,
then 
\begin{align}\label{Equ: annular action}
T_{\beta}(\alpha_1)&=\frac{Tr(\beta)}{Tr(\alpha_0)} \alpha_0.    
\end{align}
as illustrated in Fig.~\ref{fig: annular action}.
By Prop.~\ref{Prop: non-zero trace}, $Tr(\beta)\neq 0$, so $\phi^{-1}(T)$ is non-zero. So $\alpha_0$ and $\alpha_1$ are annular equivalent.

\begin{figure}
    \centering
    \begin{tikzpicture}
    \begin{scope}[scale=.75]
\draw (-2,-2) rectangle (2,2);
\draw[blue] (-1,-1) rectangle (1,1);
\node at (-1.8,0) {$\alpha_0$};
\node at (-.6,0) {$\alpha_1$};
\node at (-1.2,0) {$\beta$};
    \end{scope}
\end{tikzpicture}
    \caption{The annular action induced by the bimodule.}
    \label{fig: annular action}
\end{figure}
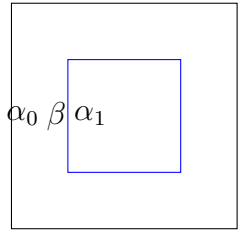

\end{proof}

\begin{corollary}
The annular equivalence is an equivalence relation.    
\end{corollary}
\begin{proof}
It follows from Prop.~\ref{Prop: Morita and annular} and \ref{Prop: Morita Equivalence}. 
\end{proof}

Suppose $\phi:\mathcal{D}_1 \to \mathcal{D}_0$ is a \PL homeomorphism fixing the boundary. Then they are isotopic by Alexander's trick, through $\phi_t$, $t\in [0,1]$, in Equ.~\ref{Equ: Alexander}.
We construct $T \in A(\mathcal{D}_1,\mathcal{D}_0)$, such that
$T|_{D^k \times t S^{n-k}}$, $t\in [\frac{1}{2},1]$ is given by $\phi_{2t-1}(\mathcal{D}_1) \times S^{n-k}$. Then
$T(xy)=T(x)T(y)$. So any $k$-morphism $\alpha$ is of type $\mathcal{D}_1$ is annular equivalent to the $k$-morphism $T(\alpha)$ of type $\phi(\mathcal{D}_1)$.
Thus the equivalence class it well defined up to isotopy $\phi:\mathcal{D}_1 \to \mathcal{D}_0$.

Recall from Def.~\ref{Def: condensation},
for a stratified manifold $\mathcal{S}$, $|\mathcal{S}|=S^{k-1}$, $\widetilde{D^k\times S^{n-k}}_{\mathcal{S}^{k-1} \times S^{n-k}}(L)$ is the vector space of $L$-labelled stratified manifold with support $D^k\times S^{n-k}$ and boundary $\mathcal{S}^{k-1} \times S^{n-k}$ modulo null vectors.

\begin{notation}
We define $\mathcal{D}^k(\beta)\times S^{n-k}$ to be the stratified manifold $\mathcal{D}^k \times S^{n-k}$ labelled by a $k$-morphism $\beta$ at $\mathcal{D}^k\times D^{n-k}\times \{-1\}$.
\end{notation}

\begin{proposition}\label{Prop: k-morphism class}
When $A(\mathcal{D}^k\times D^{n-k})$ is semisimple for any $\mathcal{D}^k$ with support $D^k$ and boundary $\mathcal{S}$. 
The vector space $\widetilde{D^k\times S^{n-k}}_{\mathcal{S}^{k-1} \times S^{n-k}}(L)$ is spanned by $\mathcal{D}^k(\beta)\times S^{n-k}$, for indecomposible $k$-morphisms $\beta$. Moreover, 
if $k$-morphism $\beta_0$ and $\beta_1$ are annular equivalent, then 
$\mathcal{D}^k(\beta_0)\times S^{n-k}=\lambda \mathcal{D}^k(\beta_1)\times S^{n-k}$, 
in $\widetilde{D^k\times S^{n-k}}_{\mathcal{S}^{k-1} \times S^{n-k}}(L)$, for some $\lambda\in \mathbb{K}$.
\end{proposition}

\begin{proof}
Suppose $\ell$ is an $L$-labelled stratified manifold in $\widetilde{D^k\times S^{n-k}}_{\mathcal{S}^{k-1} \times S^{n-k}}(L)$. By isotopy, we may assume that $\ell$ intersect with $D^k\times O_{n-k}\times \{-1\}$ transversely and the intersection is a stratified manifold $\mathcal{D}^k$.
By isotopy, we assume that $\ell|_{D^k\times D^{n-k}\times \{-1\}}$ is $\mathcal{D}^k\times O_{n-k}\times \{-1\}$. We decompose the identity $\mathcal{D}^k\times O_{n-k}$ of $A(\mathcal{D}^k\times D^{n-k})$ as a sum of minimal idempotents, i.e., $k$-morphisms $\beta$.
By the semisimplicity, a vector in $\widetilde{D^k\times S^{n-k}}_{\mathcal{S}^{k-1} \times S^{n-k}}(L)$ with a label $\beta$ at $D^k\times D^{n-k}\times \{-1\}$ is a multiple of $\mathcal{D}^k(\beta)\times S^{n-k}$.
So the vector space $\widetilde{D^k\times S^{n-k}}_{\mathcal{S}^{k-1} \times S^{n-k}}(L)$ is spanned by these $\mathcal{D}^k(\beta)\times S^{n-k}$.

If $k$-morphism $\beta_0$ and $\beta_1$ are annular equivalent, by isotopy, we can change $\mathcal{D}^k(\beta_0)\times S^{n-k}$ as a vector labelled by $\beta_1$ at $D^k\times D^{n-k}\times \{-1\}$. By the semisimplicity, the vector is a multiple of $\mathcal{D}^k(\beta_1)\times S^{n-k}$.
\end{proof}

\begin{definition}\label{Def: Z CF}
For the $S^n$ functional $Z$, we call $Z$ $n$-finite, if for any $\mathcal{S}$, $|\mathcal{S}|=S^{n-1}$, the vector space $V_{\mathcal{S},-}$ is finite dimensional.
For $0\leq k\leq n-1$, we call $Z$ $k$-finite, if for any $\mathcal{S}$, $|\mathcal{S}|=S^{k-1}$, there are finitely many annular equivalent classes of indecomposible $k$-morphisms with boundary $\mathcal{S}$.
We call $Z$ complete finite, if it is $k$-finite for all $0\leq k \leq n$.
\end{definition}

\textcolor{black}{For any $|\mathcal{S}|=S^{k-1}$, it is more conceptual to consider an equivalence classes of $k$-morphisms as a vector with support $D^k \times S^{n-k}$ and boundary $\mathcal{S}^{k-1} \times S^{n-k}$.}

\subsection{Simplicial morphisms}

In this section, we will label indecomposible $k$-morphisms on the $k$-simplices of an $n$-simplex, as a preparation for computing the partition function of $(n+1)$-manifolds as a state sum based on a triangulation. To simplify the state sum formula, we expect to choose the representatives of $k$-morphisms as simple as possible. 

The methods also work for polytopes, which could be used for CW-complex decomposition of manifolds. This will be clear after we construct the TQFT and derive the skein theory for general $k$-morphisms. 

Just like in algebraic topology, usually it is more convenient to prove theoretical results using simplicial decomposition, as there are less shapes to discuss.
It is more convenient to compute the invariant using CW-complex decomposition in practice, as it has less cells in the CW-complex decomposition than in the triangulation of a manifold.

Now let us show how to label the simplices by indecomposible $k$-morphisms.

For $k=0$, a 0-morphism is a minimal idempotent of the algebra in $A(D^0\times D^n)$. It has the boundary $\mathcal{S}=\partial (D^0\times D^n=S^{n-1})$.

For $1\leq k \leq n-1$, suppose $\mathcal{D}^k$ is a stratified manifold $D^k$, and $\alpha$ is a $k$-morphism of type $\mathcal{D}^k$.
Suppose $\Delta^{k}$ is a linear $k$-simplex and $\sigma: \Delta^k \to D^k$ is a \PL homeomorphism, such that the skeletons of $\sigma(\Delta^k)$ intersects with $ \mathcal{D}^k$ transversely.
We consider $\sigma$ as a $\Delta^k$ decomposition of $\partial \mathcal{D}^{k}$.
For every 0-face $F$ of $\Delta^k$, we consider $\mathcal{F}:=\mathcal{D}^k|_{\sigma(F)}$ as a 0-face of $\mathcal{D}^k$.
Take a \PL homeomorphism $\phi_F: F \to D^{k-1}$.
It extends to a homeomorphism $\phi'_F: U(F) \to D^{k-1}\times D^{n-k+1}$, where $U(F)$ is a normal microbundle of $F$ in $\mathcal{D}^{k}\times D^{n-k}$.

Suppose $\beta$ is a $(k-1)$-morphism of type $\phi_F(\mathcal{F})$.
We denote $P_{\beta}$ as the stratified manifold $\mathcal{D}^k\times D^{n-k}$ labelled by $(\phi'_F)^{-1}(\beta)$.
Then $P_{\beta}$ is an idempotent of $A(\mathcal{D}^k\times D^{n-k})$. 
And the identity decomposes as $\mathcal{D}^k\times D^{n-k}=\sum_{\beta} P(\beta)$ summing over all $k$-morphisms of type $\phi_F(\mathcal{F})$. 
By \PL isotopy, $P(\beta)$ commutes with $\alpha$. So there is a unique $k$-morphism $\beta$ containing the minimal idempotent $\alpha$.

\begin{definition}
For a $k$-morphism $\alpha$ with support $\mathcal{D}^k$, a $\Delta^k$ decomposition $\sigma$ and a 0-face $F$, we call the above unique $P_{\beta}$ containing $\alpha$ the boundary $(k-1)$-morphism of $\alpha$ at $\sigma(F)$, denoted by $\alpha_F$.
We say $\alpha_F$ is homeomorphic equivalent to $\beta$.
\end{definition}

\begin{definition}
For a $k$-morphism $\alpha$ with a simplicial decomposition $\sigma$, we define its simplicial boundary as
$$\partial_{\sigma}(\alpha)= \bigoplus_{F} \alpha_{\sigma(F)},$$
summing over all 0-faces $F$.
\end{definition}

\begin{definition}
Suppose $C_k$, $0 \leq k \leq n$, is a set $k$-morphisms. 
We call $C_0$ to be complete if any $0$-morphism is annular equivalent to an element in $C_0$.

For $1\leq k \leq n-1$, we call $C_{k}$ to be $C_{k-1}$-complete, if for any $k$-morphism $\alpha$, there is a homeomorphism $\phi$ on $D^k$, such that $(\phi\times I_{D^{n-k}})(\alpha)$ is annular equivalent to an element in $C_k$, whenever $\alpha$ has a simplicial decomposition with boundary $(k-1)$-morphisms homeomorphic equivalent to $(k-1)$-morphisms in $C_{k-1}$.

We call $C_{n}$ to be $C_{n-1}$-complete, if for any $n$-morphism $\alpha$, i.e. a vector in $\tilde{V}_{\mathcal{S},-}$, there is a homeomorphism $\phi$ on $D^k$, such that $\phi(\alpha)$ is a linear sum of vectors in $C_n$, whenever the boundary face of $|\mathcal{S}|=S^{n-1}$ are labelled by homeomorphisms of $n-1$ morphisms in $C_{n-1}$.
\end{definition}

\begin{definition}
Suppose $C_k$ is a set of $k$-morphisms, $0 \leq k \leq n$. If $C_{k}$ is $C_{k-1}$-complete, for all $k$, then we call $C_{\bullet}=\cup_{k=0}^n C_k$ a simplicial representative set.
If $C_{\bullet}$ has no simplicial representative subset, then we call it minimal.
\end{definition}

For a minimal simplicial representative set, $k$-morphisms in $C_k$, $0 \leq k \leq n-1$, are pairwise annular inequivalent and homeomorphic inequivalent, and vectors in $C_n$ are linearly independent.

\begin{definition}
We call the $S^n$ functional $Z$ simplicial finite, if it has a simplicial representative set $C_{\bullet}=\cup_{k=0}^n C_k$ and every $C_k$ is a finite set. 
\end{definition}

\begin{question}
Whether the simplicial finiteness of $Z$ implies the complete finiteness? 
\end{question}

\subsection{Spherical n-category}

As defined in Def.~\ref{Def: composition of morphisms},
we can compose two $k$-morphisms into one $k$-morphism by gluing the two $k$-discs into one $k$-disc along the boundary. We consider a $k+1$-morphism as a morphism or a bimodule between two $k$-morphisms.
In this way, all $k$-morphisms, $0\leq k \leq n$, form an $n$-category.
The higher pivotal structure of the $n$-category, such as the saddle surface of a 3-category in Fig.~\ref{fig:saddle}, is encoded by the regular stratified PL manifold.

\begin{figure}[H]
    \centering
    \begin{tikzpicture}
        \begin{scope}
            \fill[opacity=.15](0,0)--(1,-.4)--(1,-.4-.8)--(0,-.8);
        \end{scope}
        \begin{scope}[shift={(.8,.2)}]
            \fill[opacity=.15](0,0)--(1,-.4)--(1,-.4-.8)--(0,-.8);
        \end{scope}
        \fill[opacity=.15](0,-.8)--(.8,.2-.8)--(1.8,.2-.4-.8)--(1,-.8-.4);
        \draw[gray!50,dashed](0,-.8)--(.8,.2-.8)--(1.8,.2-.4-.8);
        \begin{scope}
            \fill[opacity=.15](0,.2-.8)--(.8,.2-.8);
        \end{scope}
        \draw[gray!50,line width=.3pt](1.8,.2-.4-.8)--(1,-.8-.4)--(0,-.8);

        \fill[opacity=.15](0,-.8)--(.8,.2-.8)--(.8,-1.5)--(0,-1.7)--(0,-.8);
        \begin{scope}[shift={(1,-.4)}]
            \fill[opacity=.15](0,-.8)--(.8,.2-.8)--(.8,-1.5)--(0,-1.7)--(0,-.8);
        \end{scope}
        
    \end{tikzpicture}
    \caption{Saddle surface}
    \label{fig:saddle}
\end{figure}
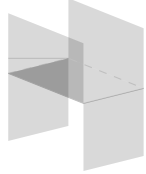

\begin{definition}
We call the $n$-category arisen from a non-degenerate, complete finite and strong semisimple $S^n$ functional $Z$ a spherical $n$-category. If $Z$ is Hermitian or reflection positive, then the spherical $n$-category is semisimple, Hermitian or unitary respectively.     
\end{definition}
We consider the spherical $n$-category as a generalization of the spherical multi-fusion category with non-zero global dimension for $n=2$. 
The complete finite and strong semisimple conditions are required to construct the TQFT in the next section.
One may release the two conditions to study more general spherical $n$-categories.

One can regard an $n$-category, as the label space $L$, together with the $S^n$ functional $Z$ as a spherical $n$-category. When $Z$ is reflection positive, it can be considered as a unitary $n$-category. The higher pivotal structures are captured by regular stratified manifolds.  

A monoidal $n$-category $\mathscr{C}$ has $k$-morphisms with compositions in the direction of the $k^{th}$ coordinate, and $k+1$-morphism between $k$-morphisms in the direction of the $(k+1)^{th}$ coordinate. Every $k$-morphism has the unit $(k+1)$-morphism.
We consider $\mathscr{C}$ as the label space $L$.

\section{Alterfold TQFT}\label{Sec: Alterfold TQFT}
In this section, we assume that the $S^n$ functional $Z$ is complete finite and reflection positivity. We will introduce the alterfold TQFT and construct a unitary alterfold $(n+1)$-TQFT from $Z$. For a general field $\mathbb{K}$, we replace reflection positivity by strong semisimplicity.

\begin{proposition}
If $Z$ is complete finite and reflection positive, then $Z$ is strong semisimple. 
\end{proposition}
\begin{proof}
When $Z$ is complete finite and reflection positive, the trace of the algebra $A(\mathcal{D}^k\times D^{n-k})$ is positive definite. So the algebra is a finite dimensional $C^*$-algebra, which is semisimple.
By Equ.~\ref{Equ: global dimension}, the global dimension of every $k$-morphism is positive. So $Z$ is strong semisimple.
\end{proof}

\subsection{Quantum Invariant}

\begin{definition}
Suppose $B$ is a compact, oriented $(n+1)$-manifold and $\mathcal{M}$ is an $L$-labelled stratified $n$-manifold with support $\partial B$, we call the pair $\mathcal{B}=(B,\mathcal{M})$ a bulk vector.
We call $B$ the support of $\mathcal{B}$, and $\mathcal{M}$ the space boundary of $\mathcal{B}$.
\end{definition}

We prove the following main theorem.

\begin{theorem}\label{Thm: RP invariant}
When the $S^n$ functional $Z$ is complete finite and strongly semisimple, $Z$ can be extended to a partition function $Z$ on $\mathcal{B}$, which is homeomorphic invariant.
If $Z$ is reflection positive, then its extension is reflection positive. 
\end{theorem}

\begin{definition}\label{Def: zeta}
Fix a non-zero $\zeta \in \mathbb{K}$.
Suppose $\mathcal{B}=(B,\mathcal{M})$ is a bulk vector and $\phi: \mathcal{B} \to D^{n+1}$ is a \PL homeomorphism,
we define    
$$Z(\mathcal{B}):=\zeta Z(\phi(\mathcal{M})).$$
\end{definition}

As $Z$ is homeomorphic invariant on $S^n$, so $Z(\phi(B))$ is independent of the choice of $\phi$.

\begin{definition}\label{Def: multiplicative}
When $\mathcal{B}$ is a union of $\mathcal{B}_i$, $|B_i|\sim D^{n+1}$,
we define
$$Z(\mathcal{B})=\prod_i Z(\mathcal{B}_i).$$
\end{definition}

\subsection{Surgery Moves}
For $D^{n+1}$, $0<\varepsilon'<\varepsilon<1$.
we remove $D^{k-1}\times \varepsilon D^{n-k+1}$ and
$D^{k-1} \times [\varepsilon,1] \times \varepsilon' D^{n-k}$,
then the result $D_+$ is \PL homeomorphic to $D^{n+1}$.
The shape of $D_+$ on the $k^{th}$ and $(k+1)^{th}$ coordinates is
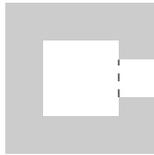
\begin{center}
\begin{tikzpicture}
\begin{scope}[scale=1]
\fill[fill opacity=.2] (-1,-1) rectangle (1,1);
\fill[white] (-.5,-.5) rectangle (.5,.5);
\fill[white] (.5,-.25) rectangle (1,.25);
\draw[dashed] (.5,-.25)--(.5,.25);
\end{scope}
\end{tikzpicture}
\end{center}
For $D^{n+1}$,
we remove $D^{k-1}\times \varepsilon D^{n-k+1}$ and
$D^{k-1} \times [-1,-\varepsilon] \times \varepsilon' D^{n-k}$,
then the result result $D_-$ is \PL homeomorphic to $D^{n+1}$.
The shape of $D_-$ on the $k^{th}$ and $(k+1)^{th}$ coordinates is
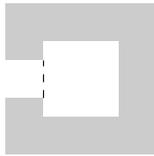
\begin{center}
\begin{tikzpicture}
\begin{scope}[scale=1]
\fill[fill opacity=.2] (-1,-1) rectangle (1,1);
\fill[white] (-.5,-.5) rectangle (.5,.5);
\fill[white] (-.5,-.25) rectangle (-1,.25);
\draw[dashed] (-.5,-.25)--(-.5,.25);
\end{scope}
\end{tikzpicture}
\end{center}

When we remove $D^{k-1}\times \varepsilon D^{n-k+1}$ from $D^{n+1}$, we change
the $S^{k-1} \times \varepsilon D^{n+1-k}$ part of the boundary $\partial D^{n+1}$ to
$D^k \times \varepsilon S{n-k}$.
For a linear sum $L$-labelled stratified manifold $\mathcal{S}$ on $\partial D^{n+1}$, such that the part $S^{k-1} \times \varepsilon D^{n+1-k}$ is $\mathcal{S} \times \varepsilon D^{n+1-k}$,
we will introduce a local $\mathcal{S}^{k-1}$-relation which changes $S^{k-1} \times \varepsilon D^{n+1-k}$ to a linear sum $L$-labelled stratified manifold with support on $D^k \times \varepsilon S^{n-k}$.
Topologically, we may consider the move as removing a $k$-handle with support  $D^k \times \varepsilon D^{n-k+1}$. Algebraically, we consider $S^{k-1} \times \varepsilon D^{n+1-k}$ as the identity of a $D^{n+1-k}$ algebra and the relation as a decomposition of the identity into minimal idempotents. The minimal idempotent is given by the annular equivalence class of $k$-morphisms with boundary $\mathcal{S}^{k-1}$.

After removing $D^k \times \varepsilon D^{n+1-k}$ from $D^{n+1}$, the orientation of its boundary is opposite to the orientation of the boundary $\mathcal{D}^k \times S^{n-k}$ of $D^k \times \varepsilon D^{n+1-k}$.

\begin{definition}
For a $k$-morphism $\alpha_k$ of type $\mathcal{D}^k$,
we define $\mathcal{D}^k(\alpha_k)\times \varepsilon S^{n-k}$ by
replacing the $\mathcal{D}^k  \times \varepsilon D^{n-k} $ part of $\mathcal{D}^k \times S^{n-k}$ by $\alpha_k$.
\end{definition}

\begin{definition}\label{Def: S-relation}
Suppose $Z$ is complete finite and $\mathcal{S}^{k-1}$ is a stratified manifold $\mathcal{S}^{k-1}$ with support $S^{k-1}$. We define an $\mathcal{S}^{k-1}$-relation for the bulk vector $(D^{n+1}, \mathcal{S}^{k-1}\times D^{n-k+1})$.

When $k=0$, 
\begin{align}\label{Equ: move-0}
\Phi(D^{n+1},\emptyset)&= \sum_{\alpha_0}\frac{Tr(\alpha_0)}{\mu(\alpha_0)} (D^{n+1}\setminus \varepsilon D^{n+1}, \mathcal{M}(\alpha_0)),
\end{align}
summing over representatives of $0$-morphisms $\alpha_0$, and
$\mathcal{M}(\alpha_0)$ is $\varepsilon S^{n}$ labelled by $\mathcal{D}^0(\alpha_0) \times \varepsilon S^{n}$.

When $1\leq  k \leq n-1$, 
\begin{align}\label{Equ: move-k}
\Phi(D^{n+1},\mathcal{S}^{k-1} \times D^{n+1-k})&= \sum_{\alpha_k}\frac{Tr(\alpha_k)}{\mu(\alpha_k)} (D^{n+1}\setminus D^{k}\times \varepsilon D^{n-k+1}, \mathcal{M}(\alpha_k)),
\end{align}
summing over representatives of $k$-morphisms $\alpha_k$ with link boundary $\mathcal{S}^{k-1}$, and 
$\mathcal{M}(\alpha_k)$ is $\mathcal{S}^{k-1} \times  (D^{n+1-k} \setminus \varepsilon D^{n+1-k}) \cup \mathcal{D}^k\times \varepsilon S^{n-k}$ labelled by $\mathcal{D}^k(\alpha_k) \times \varepsilon S^{n-k}$.
(Here $\alpha_k$ is the type of $\mathcal{D}^k$.)

When $k=n$, 
\begin{align}\label{Equ: move-n}
\Phi(D^{n+1},\mathcal{S}^{n-1} \times D^{1})&= \sum_{\alpha_n} (D^{n+1}\setminus D^{n}\times \varepsilon D^{1}, \mathcal{M}(\alpha_n)),
\end{align}
summing over $\alpha_n$ in a basis of $V_{\mathcal{S}^{n-1}}$, with the dual vector $\alpha_n'$ in the dual basis, and
$\mathcal{M}(\alpha_n)$ is $\mathcal{S}^{n-1} \times  (D^{1} \setminus \varepsilon D^{1}) \cup \Lambda \mathcal{S}^{n-1} \times \varepsilon S^{0}$ labelled by $\alpha_n \times \{\varepsilon\}$ and  $\alpha_n' \times \{-\varepsilon\}$, 

When $k=n+1$, suppose $\mathcal{S}^n$ is an $L$-labelled stratified manifold,
\begin{align}\label{Equ: move-n+1}
\Phi(D^{n+1},\mathcal{S}^n)&=\zeta Z(\mathcal{S}^n) (\emptyset,\emptyset).   
\end{align}

\end{definition}

We will prove that the $\mathcal{S}^{k-1}$-relation is independent of the choice of the representatives of the annular equivalence class in Theorem \ref{Thm:commtative k-morphism algebra}. 
The relations come from the {\it null principle}, which we will explain in Theorem \ref{Thm: unique extension} and Table~\ref{Table: null principle}.

\begin{definition}
Suppose $\mathcal{S}^n$ is an $L$-labelled stratified with a local part $\mathcal{S}^{k-1} \times \varepsilon D^{n+1-k}$ and a local part  $\mathcal{S}^{k} \times \varepsilon D^{n-k}$.
We define the $(k-1,k)$ moves as follows.
We first apply the $\mathcal{S}^{k-1}$-relation to the local part $\mathcal{S}^{k-1} \times \varepsilon D^{n+1-k}$. 
Then we obtain stratified manifolds $\mathcal{S}^{k}_{\pm}$ as
\begin{align*}
\mathcal{S}^{k}_+&=\mathcal{D}^k(\alpha_k) \times \{\varepsilon\} \cup \mathcal{S}^{k}|_{x_k \geq \varepsilon}     \\
\mathcal{S}^{k}_-&=\rho_{(k,k+1)}(\mathcal{D}^k(\alpha_k) \times \{\varepsilon\}) \cup \mathcal{S}^{k}|_{x_k \leq -\varepsilon}, \\
\end{align*}
where $\rho_{(k,k+1)}$ is the $180^{\circ}$ rotation of the $k^{th}$ and $(k+1)^{th}$ coordinates.
Then we apply the $\mathcal{S}^{k}_{\pm}$-relation. The support of the result $L$-labelled stratified manifolds are homeomorphic to $S^n$ as illustrated in Fig.~\ref{fig: Zpm}. We denoted their $Z$ value by $Z_{\pm}$ respectively.
\end{definition}

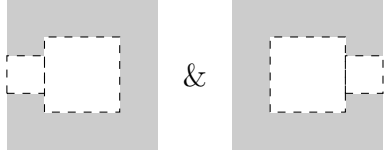
\begin{figure}[h]
    \centering
   \begin{tikzpicture}
\node at (1.5,0) {$\&$};
\begin{scope}[scale=1]
\fill[fill opacity=.2] (-1,-1) rectangle (1,1);
\fill[white] (-.5,-.5) rectangle (.5,.5);
\fill[white] (-.5,-.25) rectangle (-1,.25);
\draw[dashed] (-.5,-.5) rectangle (.5,.5);
\draw[dashed] (-1,-.25) rectangle (-.5,.25);
\end{scope}
\begin{scope}[scale=1,shift={(3,0)}]
\fill[fill opacity=.2] (-1,-1) rectangle (1,1);
\fill[white] (-.5,-.5) rectangle (.5,.5);
\fill[white] (.5,-.25) rectangle (1,.25);
\draw[dashed] (-.5,-.5) rectangle (.5,.5);
\draw[dashed] (.5,-.25) rectangle (1,.25);
\end{scope}
\end{tikzpicture}
    \caption{The $(k-1,k)$-move on the $k^{th}$ and $(k+1)^{th}$ coordinates.}
    \label{fig: Zpm}
\end{figure}

We prove the following key result about adjacent moves.

\begin{theorem}[Adjacent Move]\label{Thm: Adjacent Move}
When $Z$ is complete finite and strong semisimple, the $S^n$ functional $Z$ is invariant under the $(k-1,k)$-move, i.e., $$Z(\mathcal{S}^n)=Z_+=Z_-.$$
Consequently, $Z_{\pm}$ are independent of the choices of the representatives of in the $\mathcal{S}^{k-1}$-relation.
\end{theorem}

\begin{proof}
This is proved by induction on $k$.
When $k=n$, the $\mathcal{S}^{n-1}$-relation cut $D^{n+1}$ into two pieces and the $\mathcal{S}^{n}$-relation evaluate one and then multiply the value of the other.
This follows from non-degeneracy of the inner product and the fact that
$$\langle  \beta_-,\beta_+ \rangle= \sum_{\alpha_{n} } \langle \beta_1, \alpha_n \rangle \langle \alpha'_n ,\beta_2 \rangle, ~\forall~ \beta_\pm\in \tilde{V}_{D^n,\mathcal{S}^{n-1},\pm}.$$ 
summing over $\alpha_n$ in a basis of $\tilde{V}_{D^n,\mathcal{S}^{n-1},-}$ with the dual vector $\alpha_n'$ in the dual basis.

As $Z$ is invariant under \PL isotopy on $S^n$.
By Theorem~\ref{Thm: transversality}, we assume that the transversal intersection of $\mathcal{S}^n$ and $S^{k-1}\times O_{n+1-k}$ is $\mathcal{S}^{k-1}$; and the transversal intersection of $\mathcal{S}^n$ and $S^{k}\times O_{n-k}$ is $\mathcal{S}^{k}$. 

Applying the $\mathcal{S}^{k-1}$-relation, we
obtain a diagram denoted by $\mathcal{S}_1$. Its support is $$S^n\setminus (S^{k-1}\times \varepsilon D^{n-k+1}) \cup D^k \times \varepsilon S^{n-k},$$ which is \PL homeomorphic to $S^k\times S^{n-k}$.
Applying the $\mathcal{S}^{k}_{\pm}$-relation to $\mathcal{S}_1$, we obtain $\mathcal{S}_{\pm}$, whose support are homeomorphic to $S^n$, and their values are $Z_\pm$ respectively.
Applying the negative $(k,k+1)$-move to $\mathcal{S}_{\pm}$, we obtain the same diagram with value $Z'$, as illustrated in Fig.~\ref{fig: Z_+=Z'=Z_-}.
By induction, we have that $Z_+=Z'=Z_-$.

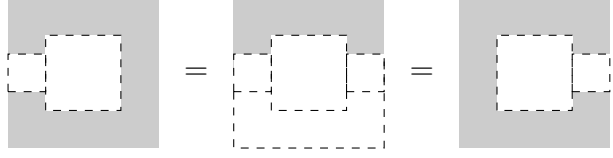
\begin{figure}
    \centering
    \begin{tikzpicture}
\node at (1.5,0) {$=$};
\node at (4.5,0) {$=$};
\begin{scope}[scale=1]
\fill[fill opacity=.2] (-1,-1) rectangle (1,1);
\fill[white] (-.5,-.5) rectangle (.5,.5);
\fill[white] (-.5,-.25) rectangle (-1,.25);
\draw[dashed] (-.5,-.5) rectangle (.5,.5);
\draw[dashed] (-1,-.25) rectangle (-.5,.25);
\end{scope}
\begin{scope}[scale=1,shift={(3,0)}]
\fill[fill opacity=.2] (-1,-1) rectangle (1,1);
\fill[white] (-1,-1) rectangle (1,0);
\fill[white] (-.5,-.5) rectangle (.5,.5);
\fill[white] (-.5,-.25) rectangle (-1,.25);
\fill[white] (.5,-.25) rectangle (1,.25);
\draw[dashed] (-.5,-.5) rectangle (.5,.5);
\draw[dashed] (-1,-.25) rectangle (-.5,.25);
\draw[dashed] (.5,-.25) rectangle (1,.25);
\draw[dashed] (-1,-.25)--(-1,-1)--(1,-1)--(1,.25);
\end{scope}
\begin{scope}[scale=1,shift={(6,0)}]
\fill[fill opacity=.2] (-1,-1) rectangle (1,1);
\fill[white] (-.5,-.5) rectangle (.5,.5);
\fill[white] (.5,-.25) rectangle (1,.25);
\draw[dashed] (-.5,-.5) rectangle (.5,.5);
\draw[dashed] (.5,-.25) rectangle (1,.25);
\end{scope}
\end{tikzpicture}
    \caption{The equality $Z_+=Z'=Z_-$.}
    \label{fig: Z_+=Z'=Z_-}
\end{figure}

If we apply \PL isotopy to the
$S^n\setminus S^{k-1}\times \varepsilon D^{n-k+1}$ part of $\mathcal{S}_1$,
and then apply the $\mathcal{S}^{k}_{\pm}$-relation, the result $Z_\pm$ do not change, as the local \PL isotopy will not affect at least one of the $\mathcal{S}^{k}_{\pm}$ relation.
So the value $Z_{\pm}$ is invariant under local isotopy.

Note that the $S^n\setminus S^{k-1}\times \varepsilon D^{n-k+1}$ part of $\mathcal{S}_1$ has boundary $\mathcal{S}^{k-1} \times \varepsilon S^{n-k}$.
Let $\phi: D^k\times S^{n-k} \to S^n\setminus S^{k-1}\times \varepsilon D^{n-k+1}$ be a homeomorphism.
By Prop.\ref{Prop: k-morphism class}, it is enough to check $Z=Z_\pm$, when the $S^n\setminus S^{k-1}\times \varepsilon D^{n-k+1}$ part is  $\phi (\mathcal{D}^k(\beta)\times S^{n-k})$.
Then in the $\mathcal{S}^{k-1}$-relation, only the term with $\beta$ contributes to a non-zero scalar.
The identity $Z=Z_\pm$ follows from the definition of $\mu(\alpha_k)$.

Both $Z_{\pm}$ are independent of the choices of the representatives in the $\mathcal{S}^{k-1}$-relation, as they equal to $Z(\mathcal{S}^n)$.
\end{proof}

The $\mathcal{S}^{k-1}$-relation eliminates a normal microbundle $\mathcal{D}^k\times \varepsilon D^{n-k-1}$.
For a bulk vector $\mathcal{B}=(B,\mathcal{M})$ and a triangulation of $B$, we may eliminate the microbundle of its $k$-simplices for $k=0,1,\cdots,n+1$, and eventually get a scalar. A main result of this section is proving that the scalar is invariant under homeomorphisms on $n+1$ manifolds and the $\mathcal{S}^{k-1}$-relation. So all the $\mathcal{S}^{k-1}$-relations are consistent.

\begin{definition}[Evaluation Algorithm]
Suppose $Z$ is complete finite. 
Suppose $\mathcal{B}=(B,\mathcal{M})$ is a bulk vector.
Take a combinatorial triangulation $\Delta$ of $B$. Take an ambient isotopy of $\mathcal{M}$, so that it is transversal to the triangulation. For every $k$-simplex $\Delta_{k,i}$, we fix a homeomorphism $\phi_{k,i}: \Delta_{k,i} \to D^k$.

For every $k$-simplex $\Delta_{k,i}$ not on the boundary of $B$, take $\mathcal{S}^{k-1}=\phi_{k,i}(\partial \Delta_{k,i} \cap \mathcal{M})$.
We change its normal microbundle by the $\mathcal{S}^{k-1}$-relation in Def.~\ref{Def: S-relation}, for $k=0,1,\cdots, n+1$.
Eventually we obtain a scalar denoted by $Z(\mathcal{B})$, called the partition function of $\mathcal{B}$, as a state sum. 
\end{definition}

As $Z$ is homeomorphic invariant on $S^n$, the partition function is independent of the choice of the homeomorphisms $\phi_{k,i}$.
Now let us prove that the partition function $Z(\mathcal{B})$ is independent of the choice of the triangulation, the ambient isotopy and the choice of the $k$-morphisms.
We first prove it is invariant under subdivisions.

\begin{lemma}\label{Lem: subdivision}
The $Z$ value of a labelled linear $n$-simplex is invariant under the subdivision, for any $1\leq k \leq n$, any point $p$ in the interior of its $k$-sub simplex and any choice of $k$-morphisms.
\end{lemma}

\begin{proof}
Suppose $p$ is an interior of $[e_0,e_1,\cdots,e_k]$ 
The $p$-subdivision of the linear $n$-simplex $[e_0,e_1,\cdots,e_n]$ is a sequence of adjacent moves.
More precisely, take $E=\{p\}+[e_{k+1},e_1,\cdots,e_n]$. 
We apply the adjacent move to $(E, \{e_k\}+E)$;
For $j=0,1,\cdots,k-1$ and every $j$ sub simplex $J$ of $[e_0,\cdots,e_{k-1}]$,
we apply the adjacent moves to $(J+E, J+\{e_k\}+E)$.
By Theorem~\ref{Thm: Adjacent Move} and the homeomorphic invariance of $Z$ on $S^n$, the $Z$ value of a labelled linear $n$-simplex is invariant under the subdivision and independent of the choice of $k$-morphisms.
\end{proof}

\begin{theorem}\label{Thm: Semisimple invariant}
If $Z$ is complete finite and strong semisimple, then the partition function $Z(\mathcal{B})$ is independent of the choice of the triangulation $\Delta$ and the choice of the $k$-morphisms. Therefore, it is invariant under homeomorphisms of $n+1$ manifolds.
\end{theorem}

\begin{proof}
We first fix the combinatorial triangulation on the boundary. By Theorem~\ref{Thm: transversality}, we ambient isotope the stratified manifold $\mathcal{M}$ on the boundary $\partial B$, so that it is transversal to the triangulation. 
Take a combinatorial triangulation $\Delta$ of the bulk with given boundary triangulation. Take a choice of $k$-morphisms in the $\mathcal{S}^{k-1}$-relation.

Firstly, we prove that the partition function $Z$ is independent of the choice of the $k$-morphisms. For any $k$-simplex $\Delta_{k,i}$, its link is a combinatorial sphere. By Lemma \ref{Lem: subdivision}, we can change choice of $k$-morphism label of $\Delta_{k,i}$ without changing the value $Z$.
This will only change morphisms labelled at the star of $\Delta_{k,i}$, i.e., the $(k+k')$-simplices containing $\Delta_{k,i}$. We can iterate it and change the choice of $k$-morphisms for $k=0,1,\cdots,n$ to any other choice without changing the value $Z$.

Secondly, we prove that $Z$ is independent of the triangulation in the bulk.
By Theorem \ref{Thm: Adjacent Move}, the value $Z$ is invariant under subdivision. 
By Theorem \ref{Thm: triangulation}, two bulk triangulations of $B$ have the same subdivision.
So $Z$ is independent of the bulk triangulation.

Thirdly, we prove that $Z$ is invariant under ambient isotopy near any point $p$ on $\mathcal{M}$.
The star $st(p)$ of $p$ in $B$ consists of $n+1$ simplices of $\Delta$ containing $p$.
As the triangulation is combinatorial, $|\partial st(p)|\sim S^n$.
By Theorem \ref{Thm: Adjacent Move}, the state sum over simplices inside $st(p)$ equals to the value of the stratified manifold with support $\partial st(p)$, which is invariant under ambient isotopy near $p$. So the partition function $Z$ is invariant under ambient isotopy of $p$, for any $p$.

Next, we prove that $Z$ is independent of the triangulation on the boundary. For two combinatorial triangulations on the boundary, they share a subdivision.
By Theorem~\ref{Thm: transversality}, we ambient isotope $\mathcal{M}$, so that it is transversal to the common subdivision. 
By Theorem~\ref{Thm: Adjacent Move}, the value $Z$ is invariant under subdivision. By Theorem~\ref{Thm: triangulation}, $Z$ is independent of the triangulation on the boundary.
So $Z$ is independent of the choice of the triangulation $\Delta$. 

For any homeomorphism $\phi$ of the $n+1$ manifolds, we have 
$Z(\mathcal{B},\Delta,C_{\bullet})=Z(\phi(\mathcal{B}),\phi(\Delta),C_{\bullet})$, because $Z(\Delta_{n+1,j})=Z(\phi(\Delta_{n+1,j}))$ for any $(n+1)$-simplex of $\Delta$.
Therefore, $Z$ is homeomorphic invariant.
\end{proof}

\begin{corollary}\label{Cor: null vector}
Suppose $\mathcal{B}=(B,\mathcal{M})$ is a bulk vector and $\mathcal{M}$ is labelled by a null vector, then 
$Z(\mathcal{B}, \Delta, C_\bullet)=0$.
\end{corollary}
\begin{proof}
By Theorem~\ref{Thm: Semisimple invariant},
we assume that the bulk vector is on the face of an $n+1$ simplex $\Delta^{n+1,j}$ of $\Delta$. Then the factor $Z(\Delta^{n+1,j})$ is always zero in the state sum of $Z(\mathcal{B}, \Delta, C_\bullet)$. So $Z((\mathcal{B},\mathcal{M}),C_\bullet)=0$.   
\end{proof}

\begin{theorem}
The partition function $Z$ is invariant under the $\mathcal{S}^{k-1}$-relations in Equ.~\ref{Equ: move-0},\ref{Equ: move-k},\ref{Equ: move-n} and \ref{Equ: move-n+1}. 
\end{theorem}

\begin{proof}
We prove the statement by induction for $k=n+1,n,\cdots,0$.
By the Def.~\ref{Def: multiplicative} and Theorem~\ref{Thm: Semisimple invariant}, $Z$ is invariant under the $\mathcal{S}^{n}$-relation. Assume that $Z$ is invariant under any $\mathcal{S}^{k}$-relation.

Suppose $\mathcal{B}=(B,\mathcal{M})$ is a bulk vector $\phi: U \times D^{n+1}$ is a homeomorphism in a regular chart, so that $\phi(\mathcal{M}|_U)=\mathcal{S}^{k-1}\times D^{n+1-k}$.
Applying the $\mathcal{S}^{k-1}$-relation, we obtain a bulk vector $\mathcal{B'}=(B',\mathcal{M}')$ which is identical to $B$ outside $U$, and $\mathcal{M}|_U$ is changed to 
$$\phi^{-1}(\sum_{\alpha_k}\frac{Tr(\alpha_k)}{\mu(\alpha_k)} (D^{n+1}\setminus D^{k}\times \varepsilon D^{n-k+1}, \mathcal{M}(\alpha_k))).$$

Now let us prove $Z(\mathcal{B})=Z(\mathcal{B}').$
We choose a combinatorial triangulation on $B\setminus \mathring{U}$ and extend it to a triangulation $\Delta$ on $B$ and a triangulation $\Delta'$ on $B'$.
We eliminate $B\setminus \mathring{U}$ according to the common triangulation on $B\setminus \mathring{U}$ by the relations. Then we obtain an $L$-labelled stratified manifold $\mathcal{N}$ with support $\phi^{-1}(D^{k} \times S^{n-k})$.
By Corollary \ref{Cor: null vector} and Prop.~\ref{Prop: k-morphism class}, we replace $\phi(\mathcal{N})$ as a linear sum of $\mathcal{D}^k(\beta)\times S^{n-k}$.
It is enough to check the identity 
\begin{align*}
&Z(D^{k}\times D^{n-k+1}, \partial\mathcal{D}^k \times D^{n-k+1} \cup \mathcal{D}^k(\beta)\times S^{n-k})\\
=&\sum_{\alpha_k}\frac{Tr(\alpha_k)}{\mu(\alpha_k)} 
Z(D^{k}\times (D^{n-k+1} \setminus \varepsilon D^{n+k-1}), \partial\mathcal{D}^k \times (D^{n-k+1} \setminus \varepsilon D^{n+k-1}) \cup \mathcal{M}(\alpha_k) \cup \mathcal{D}^k(\beta)\times S^{n-k}).
\end{align*}
The left hand side is $\zeta Tr(\beta)$.
Applying the $\mathcal{S}^{k}$-relation to the right side,
it is non-zero only when $\alpha_k$ and $\beta$ are have a bimodule. In this case, they are annular equivalent by Prop.~\ref{Prop: Morita and annular}. By Prop.~\ref{Prop: k-morphism class}, we assume that $\alpha_k=\beta$.
Then 
\begin{align}\label{Equ: inner product}
&Z(D^{k}\times (D^{n-k+1} \setminus \varepsilon D^{n+k-1}), \partial\mathcal{D}^k \times (D^{n-k+1} \setminus \varepsilon D^{n+k-1}) \cup \mathcal{M}(\beta) \cup \mathcal{D}^k(\beta)\times S^{n-k}) \nonumber \\
=&\sum_{\gamma} \frac{Tr(\gamma)}{\mu(\gamma)} \times Tr(\gamma)=\mu(\beta), 
\end{align}
summing over representatives of indecomposible $\beta-\beta$ bimodules $\gamma$.
So the right side becomes
$\frac{Tr(\beta)}{\mu(\beta)}\mu(\beta)$, which is equal to the left side $Tr(\beta)$. So $Z(\mathcal{B})=Z(\mathcal{B}').$
We complete the induction.
\end{proof}

Now we are free to use $\mathcal{S}^{k-1}$-relations to evaluated the partition function $Z$. These relations allow us to evaluate $Z$ not only by triangulations, but also by surgery theory. In practice, we would like to choose the $k$-morphisms in $\mathcal{S}^{k-1}$-relation as simple as possible, so that there are fewer terms in the state sum. Moreover, we may fix the choice of the $k$-morphisms as elements in a minimal simplicial representative set $C_{\bullet}$ up to annular equivalence and homeomorphic equivalence.

\begin{definition}[State Sum Evaluation Algorithm for Triangulation]
Suppose $Z$ is complete finite. We fix a minimal simplicial representative set $C_{\bullet}$.Suppose $\mathcal{B}=(B,\mathcal{M})$ is a bulk vector.Take a combinatorial triangulation $\Delta$ of $B$, which is transversal to $\mathcal{M}$ on the boundary $\partial B$. For every $k$-simplex $\Delta_{k,i}$, we fix a homeomorphism $\phi_{k,i}: \Delta{k,i} \to D^k$.

For every $k$-simplex $\Delta_{k,i}$ on the boundary of $B$, we consider the normal microbundle of $\Delta_{k,i} \cap \mathcal{M}$ as a label of the identity of type $\phi_{k,i}(\Delta_{k,i} \cap \mathcal{M})$, and decompose it as a sum of $k$-morphisms $\alpha_k$ and then change them to $k$-morphisms homeomorphic equivalent to elements in $C^k$ by annular actions.  

Then every $k$-simplex $\Delta_{k,i}$ not on the boundary of $B$, $k=0,1,\cdots, n+1$, take $\mathcal{S}^{k-1}=\phi_{k,i}(\partial \Delta_{k,i} \cap \mathcal{M})$.We change its normal microbundle by the $\mathcal{S}^{k-1}$-relation in Def.~\ref{Def: S-relation} and choose the $k$-morphisms homeomorphic equivalent to the ones in $C_k$ for $0\leq k \leq n$. Finally we obtain a scalar denoted by $Z(\mathcal{B},\Delta,C_\bullet)$, called the partition function of $\mathcal{B}$, as a state sum. 
\end{definition}

When $B$ is a closed $n+1$ manifold without boundary and $\Delta$ is a triangulation with $k$-simplices $\Delta_{k,i}$, $1\leq i \leq n_k$.
A $C_{\bullet}$-color triangulation $\alpha_{\bullet}$ is assignment of every $k$-simplex $\sigma_{k,i}$ of the triangulation a label $\phi_{k,i}^{-1}(\alpha_{k,i})$ for $\alpha_{k,i} \in C_k$ and a homeomorphism $\phi_{k,i} : \Delta_{k,i} \to D^k$, 
such that the link boundary of $\phi_{k,i}^{-1} (\alpha_{k,i})$ are given by $\phi_{k-1,i'}^{-1}\alpha_{k-1,i'}$ for all faces. (Otherwise the value is zero.)
Then for every $\Delta_{n+1,j}$, its faces are labelled by vectors in $C_k$ or their dual according to the orientation, which form an $L$ labelled stratified manifold with support $\partial \Delta_{n+1,j}$. We denote its value by $F(\Delta_{n+1,j})(\alpha_{\bullet})$.
Then the partition function $Z(\mathcal{B}=Z(\mathcal{B},\Delta,C_\bullet)$ is the state sum over $C_{\bullet}$-color triangulations according to the decomposition of the $\mathcal{S}^{k-1}$-relation in Def.~\ref{Def: S-relation}. It has the following form:

\begin{align}
Z(\mathcal{B},\Delta,C_\bullet)&=\sum_{\alpha_{\bullet}} \prod_{k=0}^{n}\prod_{i=1}^{n_k}  \frac{Tr(\alpha_{k,i})}{\mu(\alpha_{k,i})} \prod_{j=1}^{n+1_{k}}F(\Delta_{n+1,j})(\alpha_{\bullet}),
\end{align}
where $Tr(\Delta_{n,i})=\mu(\Delta_{n,i})=1$.

\subsection{TQFT with space-time boundary}
Now let us extend the partition function in Theorem~\ref{Thm: Semisimple invariant} to a TQFT with space-times boundary.
We consider the previous $(B,\mathcal{M})$ as a bulk vector without time boundary.
We study bulk vectors with time boundary, similar to Section~\ref{Sec: Hyper-Sphere Functions}.

\begin{definition}
Suppose $F$ is an oriented compact $n$-manifold and $\mathcal{S}$ is a stratified $(n-1)$-manifold with support $\partial F$, we call the pair $\mathcal{F}:=(F,\mathcal{S})$ a time boundary.
\end{definition}

\begin{definition}
A bulk vector with time boundary $\mathcal{F}=(F,\mathcal{S})$ is a pair $\mathcal{B}:=(B,\mathcal{M})$, such that 
\begin{enumerate}
    \item $B$ is an oriented compact $n+1$ manifold;
    \item $\mathcal{M}$ is an oriented compact $L$-labelled stratified $n$-manifold;
    \item $|\mathcal{M}| \cup F=\partial B$;
    \item $|\mathcal{M}| \cap F= \partial |\mathcal{M}|=-\partial F$;
    \item $\mathcal{M}$ and $F$ intersect transversely. 
\end{enumerate}
\end{definition}

\begin{definition}
Suppose $\mathcal{F}=(F,\mathcal{S})$ is a time boundary.
We define $V_{\mathcal{F},\pm}$ to be the vector space spanned by bulk vectors with time boundary $(-1)^{n+1}\pm\mathcal{F}$.
\end{definition}

Suppose bulk vectors $\mathcal{B}_{\pm}=(B_{\pm},\mathcal{M}_{\pm})$ have a common boundary $\mathcal{F}$ with opposite orientations.
Then we can glue the boundary and obtain a bulk vector $\mathcal{B}=\mathcal{B}_+ \cup \mathcal{B}_-$ without boundary.

\begin{definition}
The partition function $Z$ defines a bi-linear form on $V_{\mathcal{F},-}\times V_{\mathcal{F},+}$,
$$Z(\mathcal{B}_- \times \mathcal{B}_+):=Z(\mathcal{B}_-\cup \mathcal{B}_+), ~\forall~ \mathcal{B}_{\pm} \in V_{\mathcal{F},\pm}.$$
\end{definition}

\begin{definition}
A vector $v \in  V_{\mathcal{F},\pm}$ is called a null vector, if $$Z(v \cup w)=0, \forall w \in V_{\mathcal{F},\mp}.$$
The subspace of all null vectors are denoted by $K_{\mathcal{F},\pm}$.
Their quotient spaces are denoted by $\tilde{V}_{\mathcal{F},\pm}:=V_{\mathcal{F},\pm}/K_{\mathcal{F},\pm}$.
\end{definition}

Suppose $\mathcal{B}=(B,\mathcal{M})$ is a bulk vector with time boundary $\mathcal{F}$.
By Corollary~\ref{Cor: null vector}, if $\mathcal{M}$ is labelled by a null vector then $\mathcal{B}$ is a null vector.
So for vectors in $\tilde{V}_{\mathcal{F},\pm}$ we can apply all relations in the kernel $K_Z$ to the local $n$-disc of $\mathcal{M}$.

By Theorem~\ref{Thm: Semisimple invariant}, any homeomorphism of $(n+1)$-manifolds fixing the boundary $\mathcal{F}$ induces the identity map on $\tilde{V}_{\mathcal{F},\pm}$.

\begin{theorem}\label{Thm: finite dimension}
When the $S^n$ functional $Z$ is complete finite and strong semisimple,
the vector space $\tilde{V}_{\mathcal{F},\pm}$ is finite dimensional.
If $Z$ is reflection positive, then for any $\zeta>0$, $Z$ induces a positive definite inner product on $\tilde{V}_{\mathcal{F},\pm}$, so the vector space is a Hilbert space. 
\end{theorem}

\begin{proof}
Suppose the time boundary is $\mathcal{F}=(F,\mathcal{S})$.
Take a triangulation $\Delta$ of $F$ transversal to $\mathcal{S}$.

When $Z$ is reflection positive, both $Tr(\alpha)$ and $\mu(\alpha)$ are positive for any $k$-morphism $\alpha$.
We replace the normal microbundle of 0-simplices of $\Delta|_\mathcal{S}$ as a linear sum of 0-morphisms in $C_0$. For $k=1,\cdots,n$, for every $k$-simplex $\sigma$ of $\Delta|\mathcal{S}$, we decompose normal microbundle its as a linear sum of $k$-morphisms $\alpha$ whose link boundary are labelled by $(k-1)$-morphisms in $C_{k-1}$.
By the completeness of $C_k$, $C_k$ has a $k$-morphism $\alpha'$ in $C_k$ annular equivalent to $\alpha$. Take a $\alpha-\alpha'$ bimodule $\beta$.
By Equ.~\ref{Equ: annular action}, we obtain annular action $T_{\beta}(\alpha)=\frac{Tr(\beta)}{Tr(\alpha')}\alpha'$.
Define $T_{\pm}$ to be the restriction of $T_{\beta}(\alpha)$ on $D^{n-1}\times D^1_{\pm}$. Then $T_-=(T_+)^*$ and $(T_-)^*T_-=\frac{Tr(\beta)}{Tr(\alpha')}\alpha'$. So $\sqrt{\frac{Tr(\alpha')}{Tr(\beta)}} T_-$ is an embedding map. Applying this embedding map, we can substitute the label $\alpha$ of the normal microbundle of $k$-simplex $\sigma$ to the label $\alpha'$. Iterate the substitution for all $k$-simplices for $k=1,2,\cdots, n-1$. It is enough to prove the finiteness and positivity when all $k$-simplex of $\sigma|_{\mathcal{S}}$ are labelled by $k$-morphisms in $C_k$.

\begin{figure}
    \centering
    \begin{tikzpicture}
\draw (-2,-2) rectangle (2,2);
\draw[blue] (-1,-1) rectangle (1,1);
\node at (-1.8,.2) {$\alpha$};
\node at (-.6,.2) {$\alpha'$};
\node at (-1.2,.2) {$\beta$};
\draw[dashed] (-2.2,0) -- (2.2,0);
\node at (-.2,-.2) {$\sigma$};
\end{tikzpicture}
    \caption{Switch $k$-morphism near the $k$-simplex of $\Delta|\mathcal{S}$.}
    \label{fig:switch k-morphisms on the boundary}
\end{figure}
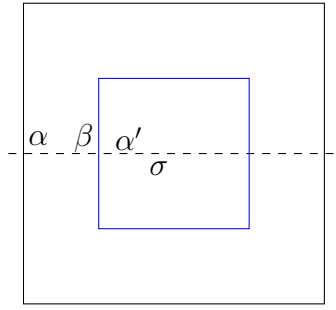

Now we deal with the simplices not on the boundary of $F$. The idea is applying the ``square root'' of the $\mathcal{S}^{k-1}$-relation to every $k$-simplex for $k=0,1,\cdots n$.

By Equ.~\ref{Equ: move-0},
\begin{align}
(D^{n+1},\emptyset)&= \sum_{\alpha_0}\frac{Tr(\alpha_0)}{\mu(\alpha_0)} (D^{n+1}\setminus \varepsilon D^{n+1}, \mathcal{M}(\alpha_0)).
\end{align}
We define the bulk vector $\iota(\alpha_0,\pm)$ as the restriction of 
$(D^{n+1}\setminus \varepsilon D^{n+1}, \mathcal{M}(\alpha_0))$ on $D^{n+1}_{\pm}$.
Take 
$$\iota_\pm= \sum_{\alpha_0}\sqrt{\frac{Tr(\alpha_0)}{\mu(\alpha_0)}} \iota(\alpha_0,\pm).$$
Then $\iota_+=\iota_-^*$ and the multiplication $\iota_+\iota_-=(D^{n+1},\emptyset)$ is the identity on the boundary $(D^{n},\emptyset)$.

Suppose $\sigma_0$ is a 0-simplex of $\Delta$ in $F$, which is not on the boundary of $F$. Suppose $(U,\phi,\emptyset)$ is a regular chart of $\sigma_0$. 
Applying the embedding map $\iota_-$ to the normal microbundle $U$, we can substitute the normal microbundle $U$ by $\phi^{-1}(D^n\setminus\varepsilon D^n)$, with a label $\alpha_0$ on $\partial \varepsilon D^n$, summing over all $\alpha_0 \in C_0$. 
The vector space $\tilde{V}_{\mathcal{F},-}$ is embedded in a direct sum of vector spaces with times boundary $(F', \mathcal{S}')$, where $F'$ is removing the normal microbundles $\phi^{-1}(\varepsilon D^n)$ of these 0-simplices from $F$; and $\mathcal{S}'$ is adding $\phi^{-1}(\partial \varepsilon D^n)$ of these 0-simplices to $\mathcal{S}$. Moreover, every $\phi^{-1}(\partial \varepsilon D^n)$ is labelled by a $0$-morphism in $C_0$ up to homeomorphic equivalence. 

For $k=1,2,\cdots, n$, by the $\mathcal{S}^{k-1}$-relations in Equ.~\ref{Equ: move-k} and \ref{Equ: move-n},
\begin{align}
(D^{n+1},\mathcal{S}^{k-1} \times D^{n+1-k})&= \sum_{\alpha_k}\frac{Tr(\alpha_k)}{\mu(\alpha_k)} (D^{n+1}\setminus (D^{k-1}\times \varepsilon D^{n-k+1}, \mathcal{M}(\alpha_k)).   
\end{align}
We define the bulk vector $\iota(\alpha_k,\pm)$ as the restriction of 
$(D^{n+1}\setminus D^{k}\times \varepsilon D^{n-k+1}, \mathcal{M}(\alpha_k))$ on $D^{n+1}_{\pm}$.
Take 
$$\iota_{\mathcal{S}^{k-1},\pm}= \sum_{\alpha_k}\sqrt{\frac{Tr(\alpha_k)}{\mu(\alpha_k)}} \iota(\alpha_k,\pm).$$
Then $\iota_{\mathcal{S}^{k-1},+}=\iota_{\mathcal{S}^{k-1},-}^*$ and the multiplication $\iota_{\mathcal{S}^{k-1},+}\iota_{\mathcal{S}^{k-1},-}=(D^{n+1},\mathcal{S}^{k-1} \times D^{n+1-k})$ is the identity on the boundary $(D^{n},\mathcal{S}^{k-1} \times D^{n-k})$.

For $k=1,2,\cdots, n$,
suppose $\sigma_k$ is a $k$-simplex of $\Delta$ in $F$, which is not on the boundary of $F$, the $k-1$-simplicies of $\partial \sigma_k$ are labelled by $(k-1)$-morphisms in $C_{k-1}$ up to homeomorphic equivalence.
Suppose $\phi_k$ is a homeomorphism from the normal microbundle $U_k$ of the $k$-simplex in $F$ to $D^k\times D^{n-k}$, such that $\phi_k(U_k|_{\partial \sigma_k})=\mathcal{S}^{k-1}\times O_{n-k}$.
Applying the embedding map $\iota_{\mathcal{S}^{k-1},-}$ to the normal microbundle $U_k$, we can substitute $U_k$ by $\phi_k^{-1}(D^n \setminus D^{k-1} \times \varepsilon D^{n-k})$, and substitute $\phi_k^{-1}(\mathcal{S}^{k-1} \times \varepsilon D^{n-k})$ by $\phi_k^{-1}(\mathcal{D}^{k} \times \varepsilon S^{n-k-1})$, labelled by $\mathcal{D}^{k}(\alpha_k) \times \varepsilon S^{n-k-1}$.

Note that each embedding map $\iota_{\mathcal{S}^{k-1},-}$ eliminates a normal microbundle of the $k$-simplex in $F$. After applying the embedding maps to all simplices of $\Delta$, $F$ becomes the empty set.
Therefore, the vector space $\tilde{V}_{\mathcal{F},-}$ is embedded in a direct sum of vector spaces of the tensor product of the 1-dimensional vector spaces spanned by vectors in $C_n$. Every $C_k$ is a finite set and the triangulation has finitely many simplices, so the vector space $\tilde{V}_{\mathcal{F},-}$ is finite dimensional. 
Every 1-dimensional vector space is a Hilbert space, so a direct sum of their tensor products is still a Hilbert space. Moreover, $\tilde{V}_{\mathcal{F},-}$ is embedded as a sub Hilbert space, so it is a Hilbert space.

For a general field $\mathbb{K}$, we cannot decompose the scalar as the product of its square roots, as the square root may not be in the field. Instead, we decompose the scalar as the product of 1 and itself. The rest part of the proof of the finite dimensional condition is similar.
\end{proof}

Atiyah's TQFT \cite{Ati88} is a symmetric monoidal functor from the cobordism category to the category of vector spaces.
Every closed $n$-manifold $F$ is assigned a finite dimensional vector space $V_F$ on the field $\mathbb{K}$, and $V_{\emptyset}\cong \mathbb{K}$.
In particular, the map from closed $n+1$-manifolds to  $\mathbb{K}$ is called the partition function, which is homeomorphic invariant.
Every $n+1$-cobordism is assigned to a linear transformation on the vector spaces.  
The disjoint union and the gluing map correspond to the tensor and contraction respectively. A general operation is a composition of the two elements operations.
The TQFT is called unitary, if the partition function is reflection positive, namely, it induces a positive definite inner product on the vector space. 

Now we consider the cobordism as an $n+1$-manifold with time boundary. 
Let us introduce the TQFT with space-time boundary, so that the cobordism is generalized to an $(n+1)$-manifold with space-time boundary, and the space boundary is a stratified $n$-manifolds.  
We keep in mind that we are working on the category of PL manifolds and regular stratified manifolds of a given local shape.

\begin{definition}
A space-time cobordism with time boundary $\mathcal{F}=(F,\mathcal{S})$ is a pair $\mathcal{B}:=(B,\mathcal{M})$, such that 
\begin{enumerate}
    \item $B$ is an oriented compact $n+1$ manifold;
    \item $\mathcal{M}$ is an oriented compact stratified $n$-manifold;
    \item $|\mathcal{M}| \cup F=\partial B$;
    \item $|\mathcal{M}| \cap F= \partial |\mathcal{M}|=-\partial F$;
    \item $\mathcal{M}$ and $F$ intersect transversely. 
\end{enumerate}
\end{definition}

\begin{definition}
Given a local shape set $LS_{\bullet}$, an $(n+1)$-TQFT with space-time boundary is a symmetric monoidal functor from the category of space-time cobordisms to the category of vector spaces over a field $\mathbb{K}$. 
More precisely, every time boundary $\mathcal{F}=(F,\mathcal{S})$ is assigned a finite dimensional vector space $V_\mathcal{F}$, and $V_\emptyset \cong \mathbb{K}$.
Every space-time $n+1$-cobordism is assigned to a linear transformation on the vector spaces.  The disjoint union and the gluing map correspond to the tensor and contraction respectively.
The map $Z:V_\emptyset \to \mathbb{K}$ is called the partition function. 
The TQFT is called unitary, if $Z$ is reflection positive.
\end{definition}

When the space boundary is the empty set, we return to Atiyah's TQFT \cite{Ati88}.

\begin{theorem}\label{Thm: space-time TQFT}
When the $S^n$ functional $Z$ is complete finite and strong semisimple, we obtain an $n+1$ TQFT with space-time boundary. If $Z$ is reflection positive and $\zeta>0$, then the TQFT is unitary.
\end{theorem}

\begin{proof}
By Theorem~\ref{Thm: Semisimple invariant} and \ref{Thm: finite dimension},  
for every time boundary $\mathcal{F}=(F,\mathcal{S})$, the vector space $\tilde{V}_{\mathcal{F},-}$ is finite dimensional.
A space-time cobordism acts on a vector by gluing the boundary. It sends a null vector to a null vector. So it is well-defined on $\tilde{V}_{\mathcal{F},-}$.
Thus we obtain an $n+1$ TQFT with space-time boundary. If $Z$ is reflection positive, then the TQFT is unitary.
If $Z$ is reflection positive, then by Theorem~\ref{Thm: finite dimension}, the TQFT is unitary.
\end{proof}

\begin{definition}\label{Def: boundary vector}
For a time boundary $\mathcal{F}=(F,\mathcal{S})$, we call $\mathcal{B}=(B,\mathcal{M}) \in \tilde{V}_{\mathcal{F},-}$ a boundary vector on $\mathcal{F}$, if $B=F\times[-\varepsilon,0]$ and $|\mathcal{M}|=F\times \{-\varepsilon\} \cup \partial F \times [-\varepsilon,0].$ 
Similarly, we define boundary vectors in $\tilde{V}_{\mathcal{F},+}$.
\end{definition}
We identify $F$ with $F\times \{0\}$. Up to isotopy, we assume that $\mathcal{M}|_{\partial F \times [-\varepsilon,0]}=\mathcal{S}\times [-\varepsilon,0]$.

\begin{notation}
Suppose $\ell$ is an $L$-labelled stratified manifold in the condensation space $F_{\mathcal{S}}(L)$ in \ref{Def: condensation}). We construct a boundary vector $\ell_-$ in $\tilde{V}_{\mathcal{F},-}$ with a label $\ell$ on the space boundary $F\times \{-\varepsilon\}$.    
\end{notation}

\begin{remark}
To compare with the notions in topological orders, one may regard the vector space $\tilde{V}_{\mathcal{F},-}$ as the space of ground states of a Hamiltonian. The boundary vector $\ell_-$ is the projection of $\ell$ onto the space of ground states.  
\end{remark}

\begin{theorem}\label{Thm: boundary vector}
For any time boundary $\mathcal{F}=(F,\mathcal{S})$, the vector space $\tilde{V}_{\mathcal{F},\pm}$ is spanned by boundary vectors.
\end{theorem}

\begin{proof}
As the vector is invariant under homeomorphisms fixing the boundary $\mathcal{F}$, we assume that $B=F\times [-\varepsilon,0] \cup B_0$, $B_0$ is a closed sub $(n+1)$-manifold of $B$ and $B_0\cap F\times [-\varepsilon,0]=F\times \{-\varepsilon\}$.
Take a triangulation $\Delta$ of the $B_0$.
Apply the $\mathcal{S}^{k-1}$-relation to every $k$-simplex of $\Delta$ summing over $k$-morphisms in $C_k$, for $0\leq k\leq n+1$. Then the bulk vector $\mathcal{B}=(B,\mathcal{M})$ becomes a boundary vector $\mathcal{B}'=(F\times [-\varepsilon,0],\mathcal{M}')$, such that $|\mathcal{M}'|=F\times \{-\varepsilon\} \cup \partial F \times [-\varepsilon,0].$
By Theorem~\ref{Thm: Semisimple invariant}, $\mathcal{B}=\mathcal{B}'$ in $\tilde{V}_{\mathcal{F},-}$.
The proof for $\tilde{V}_{\mathcal{F},+}$ is similar.
\end{proof}

By Prop.~\ref{Prop: k-morphism class},
the vector space $\widetilde{D^k\times S^{n-k}}_{\mathcal{S}^{k-1} \times S^{n-k}}(L)$ is spanned by the boundary vectors $\mathcal{D}^k(\beta)\times S^{n-k}$, for $k$-morphisms $\beta$. Moreover, if $k$-morphism $\beta_0$ and $\beta_1$ are annular equivalent, then $\mathcal{D}^k(\beta_0)\times S^{n-k}=\lambda \mathcal{D}^k(\beta_1)\times S^{n-k}$, in $\widetilde{D^k\times S^{n-k}}_{\mathcal{S}^{k-1} \times S^{n-k}}(L)$, for some $\lambda\in \mathbb{K}$.
Now let us upgrade $\mathcal{D}^k(\beta)\times S^{n-k}$ to a bulk vector. 

\begin{notation}
Take $F=D^k\times S^{n-k}$ and $\mathcal{S}=\mathcal{S}^{k-1}\times S^{n-k}.$ 
We denote $\mathcal{D}^k(\beta)_-\times S^{n-k}$ to be the bulk vector in $\tilde{V}_{\mathcal{F},-}$ with time boundary 
$\mathcal{F}=(F,\mathcal{S})$, which has support $D^k\times [-1,0] \times S^{n-k} $ and $D^k\times \{-1\} \times S^{n-k}$ is the labelled stratified manifold $\mathcal{D}^k(\beta)\times S^{n-k}$.
We denote $\mathcal{D}^k(\beta)_+\times S^{n-k}$ to be the bulk vector in $\tilde{V}_{\mathcal{F},+}$ with time boundary 
$\mathcal{F}=(F,\mathcal{S})$, which has support $D^k\times [0,1] \times S^{n-k} $ and $D^k\times \{1\} \times S^{n-k}$ is the labelled stratified manifold $\rho_{k,k+1}(\mathcal{D}^k(\beta)\times S^{n-k})$.
\end{notation}

\begin{theorem}\label{Thm: k-morphism basis}
Take the time boundary $\mathcal{F}=(F,\mathcal{S})$, $F=D^k\times S^{n-k}$, $\mathcal{S}=\mathcal{S}^{k-1}\times S^{n-k}.$ Take a set of representatives $\{\beta_i: i\in I\}$ of $k$-morphisms with link boundary $\mathcal{S}$ of annular equivalence classes.
Then the boundary vectors $\{\mathcal{D}^k_\pm(\beta_i)\times S^{n-k}:i\in I\}$ form a basis of $\tilde{V}_{\mathcal{F},\pm}$ respectively. 
Moreover, their inner product is
\begin{align}
Z(\mathcal{D}^k_-(\beta_i)\times S^{n-k}) \cup \mathcal{D}^k_+(\beta_j)\times S^{n-k}) )&=\delta_{i,j} \mu(\beta_i).   
\end{align}
\end{theorem}

\begin{proof}
By Theorem~\ref{Thm: boundary vector} and Prop.~\ref{Prop: k-morphism class}, the bulk vectors $\{\mathcal{D}^k_\pm(\beta_i)\times S^{n-k}:i\in I\}$ form a basis of $\tilde{V}_{\mathcal{F},\pm}$ respectively. 
Moreover, their inner product is $\delta_{i,j} \mu(\beta_i)$ as computed in Equ.~\ref{Equ: inner product}. 
\end{proof}

\begin{theorem}\label{Thm:commtative k-morphism algebra}
Take the time boundary $\mathcal{F}=(F,\mathcal{S})$, $F=D^k\times S^{n-k}$, $\mathcal{S}=\mathcal{S}^{k-1}\times S^{n-k}.$ The vector space $\tilde{V}_{\mathcal{F},-}$ forms a commutative semisimple algebra with the identity $(D^{k}\times D^{n-k+1}, \mathcal{S}^{k-1},D^{n-k+1})$. 
The $\mathcal{S}^{k-1}$-relation in Equ.~\ref{Equ: move-0}, \ref{Equ: move-k}, \ref{Equ: move-n} is a resolution of the identity, for $0\leq k\leq n$. (The multiplication direction is the last coordinate.)
\end{theorem}

\begin{proof}
The multiplication of $\tilde{V}_{\mathcal{F},-}$ is commutative as shown in Fig.~\ref{fig: commutativity}.
By isotopy, $(D^{n+1}, \mathcal{S}^{k-1}\times D^{n-k+1})$ is the identity. 
It has a basis $\{\mathcal{D}^k_\pm(\beta_i)\times S^{n-k}:i\in I\}$.
The multiplication of two vectors for two different representatives $\beta_i$ and $\beta_j$ is zero, because there is no bimodules between them.
So $\mathcal{S}^{k-1}$-relation is a decomposition of the identity as a sum of minimal idempotents $0\leq k\leq n$.
\end{proof}

\begin{theorem}\label{Thm: invariance of the intrinsic dimension}
The intrinsic dimension $\displaystyle \frac{Tr(\alpha)^2}{\mu(\alpha)}$ is invariant under annular equivalence.   
\end{theorem}

\begin{proof}
By Theorem \ref{Thm:commtative k-morphism algebra},
$\frac{Tr(\alpha_k)}{\mu(\alpha_k)} (D^{n+1}\setminus D^{k}\times \varepsilon D^{n-k+1}, \mathcal{M}(\alpha_k))$ is a minimal idempotent in the commutative semisimple algebra $\tilde{V}_{\mathcal{F},-}$. It is independent of the choice of the representative $\alpha_k$ in the annular equivalence class. Its trace is intrinsic dimension $\frac{Tr(\alpha)^2}{\mu(\alpha)}$. 
\end{proof}

\begin{theorem}\label{Thm: unique extension}
Suppose the $S^n$ functional $Z$ is complete finite and strong semisimple. Then its has a unique extension to the $(n+1)$-TQFT with space-time boundary, such that $Z(\mathcal{D}^{n+1},\mathcal{S})=\zeta Z(\mathcal{S})$ and $Z$ is multiplicative.
\end{theorem}

\begin{proof}
As $Z$ is multiplicative, the partition function of bulk vectors with support $D^{n+1}\times S^0$ and empty time boundary is the product of partition function of the two components.

By induction on $k=n,n-1,\cdots,0,-1$, suppose the partition function on vectors with support $D^{k+1}\times S^{n-k}$ and empty time boundary are determined. 
Then for any time boundary $\mathcal{F}=(F,\mathcal{S})$, $F=D^k\times S^{n-k}$, $\mathcal{S}=\mathcal{S}^{k-1}\times S^{n-k},$ the inner product between boundary vectors in $\tilde{V}_{\mathcal{F},-}$ and $\tilde{V}_{\mathcal{F},+}$ are determined. 
The union of $\mathcal{S}^{k-1}\times D^{n-k+1}$ in $V_{\mathcal{F},-}$ and a boundary vector in $V_{\mathcal{F},+}$ has support homeomorphic to $D^{n+1}$ and space boundary homeomorphic to $S^n$.
So their inner product is determined. 
So the expression of  $\mathcal{S}^{k-1}\times D^{n-k+1}$ as a linear sum of boundary vectors in  $V_{\mathcal{F},-}$ is unique, which is the $\mathcal{S}^{k-1}$-relation in Def.\ref{Def: S-relation}.
By Theorem \ref{Thm: Semisimple invariant}, the partition function is determined by $\mathcal{S}^{k-1}$-relations. So it is determined by the $S^n$ functional $Z$. 
So the TQFT is also determined by $Z$.
We illustrate the induction process in Table.~\ref{Table: null principle}.
\end{proof}

We summarize the $\mathcal{S}^{k-1}$-relations in Table.~\ref{Table: null principle}, which shows how to derive all relations from the $S^n$ functional $Z$.
In general, for any time boundary $\mathcal{F}=(F,\mathcal{S})$, $F=D^k\times S^{n-k}$, $\mathcal{S}=\mathcal{S}^{k-1}\times S^{n-k},$ the $\mathcal{S}^{k-1}$-relation could also be considered as a bistellar move of the space boundary, which changes the space boundary from $D^{k} \times \partial D^{n-k+1}$ to $\partial D^k \times D^{n-k+1}$ and change the bulk from $D^{n+1}$ to $\emptyset$.

For any vector $v$ with support $D^{n+1}$ time boundary $\mathcal{F}$, such as the identity $\mathcal{S}^{k-1}\times D^{n-k+1}$ in $\tilde{V}_{\mathcal{F},-}$,
its union with a boundary vector in $\tilde{V}_{\mathcal{F},+}$ has support homeomorphic to $D^{n+1}$ and space boundary homeomorphic to $S^n$. So we can regard the $v$ as a linear functional on $\tilde{V}_{\mathcal{F},+}$ according to the value of $Z$ on $S^n$.

When the inner product between $\tilde{V}_{\mathcal{F},-}$ and $\tilde{V}_{\mathcal{F},+}$ is non-degenerate, we can regard $\tilde{V}_{\mathcal{F},-}$ as the dual space of $\tilde{V}_{\mathcal{F},+}$. 
Therefore we can express $v$ as a boundary vector $b$ in  $\tilde{V}_{\mathcal{F},-}$ induced from the inner product with boundary vectors $b_+$ in the dual space $\tilde{V}_{\mathcal{F},+}$. We call this method of producing the skein relations the null principle.
\begin{align*}
\langle \Phi(v),b_+ \rangle = \langle b, b_+\rangle ~\forall~ b_+ \in \tilde{V}_{\mathcal{F},+}.
\end{align*}

\begin{table}[H]
\centering
\setlength{\tabcolsep}{1pt}
\renewcommand{\arraystretch}{0.83} 
\begin{tabular}{c| c | c | c}
 \hline
 bulk $B$ & \makecell[c]{Space boundary\\
 $\varepsilon \to 0$}
 & Time boundary & Partition function \& bistellar moves\\ [0.5ex]
 \hline\hline
 $D^{n+1}$ & $\mathcal{S}^n$ & $\emptyset$ & $Z(\mathcal{D}^{n+1}):=\zeta Z(\mathcal{S}^n)$ \\
 \hline\hline
 $D^{n+1}\times S^0$ & $\mathcal{S}^{n}\times \mathcal{S}^0$ & $\emptyset$ & $Z(\mathcal{D}^{n+1}\times \mathcal{S}^0):= Z(\mathcal{D}^{n+1})Z(\mathcal{D}^{n+1})$ \\
 \hline
 $D^{n}_-\times S^0$ & $\mathcal{D}^{n}\times \mathcal{S}^0$ & $D^{n}\times \mathcal{S}^0$ & non-degenerate inner product\\
 \hline
 $D^{n+1}$ & $\mathcal{S}^{n-1}\times D^1 $ & $D^{n}\times \mathcal{S}^0$ & \makecell[c]{ $\Phi(D^{n+1},\mathcal{S}^{n-1}\times D^1)$\\ $:=\sum (D^{n}_-\times S^0, \mathcal{D}^{n}\times \mathcal{S}^0)$} \\
 \hline
 \hline
 $D^{k+1}\times S^{n-k}$ & $\mathcal{S}^{k}\times \mathcal{S}^{n-k}$ & $\emptyset$ & \makecell[c]{$Z(\mathcal{D}^{k+1}\times \mathcal{S}^{n-k}):=$\\ $\langle \mathcal{D}^{k+1}\times D^{n-k},\mathcal{D}^{k+1}\times \mathcal{D}^{n-k}\rangle $} \\
 \hline
 $D^{k}_-\times S^{n-k}$ & $\mathcal{D}^{k}\times \mathcal{S}^{n-k}$ & $D^{k}\times \mathcal{S}^{n-k}$ & non-degenerate inner product\\
 \hline
 $D^{n+1}$ & $\mathcal{S}^{k-1}\times D^{n-k+1} $ & $D^{k}\times \mathcal{S}^{n-k}$ &\makecell[c]{ $\Phi(D^{n+1},\mathcal{S}^{k-1}\times D^{n-k+1})$\\$:=\sum (D^{k}_-\times S^{n-k},\mathcal{D}^{k}\times \mathcal{S}^{n-k}) $} \\
 \hline
\end{tabular}
\label{Table: null principle}
\caption{null principle and Bistellar Moves}
\end{table}

We have the freedom to choose $\zeta\neq 0$ in the definition of the partition function. To ensure reflection positivity, we need $\zeta>0$.
\begin{proposition}
The choice of $\zeta$ changes the partition function $Z(B,\mathcal{M})$ by a global factor $\zeta^{(-1)^{n+1}E(B)}$, where $E(B)$ is the Euler number of $B$.    
\end{proposition}
\begin{proof}
By the null principle in Table~\ref{Table: null principle}, the $\mathcal{S}^{k-1}$ has a global factor $\zeta^{(-1)^{n+1-k}}$.
Thus for the triangulation $\Delta$ of the $n+1$ manifold $B$, 
the global factor is $\prod_{\Delta_{k,i}}\zeta^{(-1)^{n+1-k}}=\zeta^{(-1)^{n+1}E(B)}.$
\end{proof}

We assumed the complete finite and strong semisimple conditions of $Z$ in our main theorems.
The complete finite condition is to ensure that the state sum of the partition function is a finite sum. 
The strong semisimple is necessary to construct the $(n+1)$-TQFT, when $Z$ is semisimple.
Otherwise if a $k$-morphism $\beta$ with link boundary $\mathcal{S}^{k-1}$ has zero global dimension, then as shown in Theorem \ref{Thm: k-morphism basis}, $\mathcal{D}^k_+(\beta_j)\times S^{n-k}$ is a null vector in $\tilde{V}_{\mathcal{F},+}$. However, its inner product with $\mathcal{S}^{k-1}\times D^{n-k+1}$ is the quantum dimension $Tr(\beta)$, which is non-zero. It is a contradiction.

\subsection{Alterfold TQFT}
In this section, we introduce the $n+1$ alterfold TQFT which extend the notions in $n+1$ TQFT with space-time boundary.

\begin{definition}
Suppose $M^{n+1}$ is an oriented, compact $(n+1)$-manifold and $M^n$ is an oriented, closed sub $n$-manifold, which separates $M^{n+1}$ as two closed sub $(n+1)$-manifolds $A^{n+1}$ and $B^{n+1}$, so that $A^{n+1}\cap B^{n+1}=M^n$, $M^n$ and $\partial B^{n+1}$ have the same orientation. 
We call the triple $(M^{n+1},B^{n+1},M^{n})$ an $(n+1)$-alterfold. 
Its boundary is an $n$-alterfold $(\BB,\TB,\STB)=(\partial M^{n+1}, \partial B^{n+1} \cap \partial M^{n+1}, \partial M^{n} \cap \partial M^{n+1})$. 
\end{definition}

\begin{definition}
We consider $A^{n+1}$ and $B^{n+1}$ as $A/B$-colored $(n+1)$ manifolds respectively. The $A$-color is called the trivial color.
The $B$-color is called the bulk color.
\end{definition}

\begin{definition}
For $(n+1)$-alterfold $(M^{n+1},B^{n+1},M^{n})$, if $\mathcal{M}$ is an $L$-labelled stratified manifold with support $M^n$, we call $(M^{n+1},B^{n+1},\mathcal{M}^{n})$ an $L$-labelled alterfold.
\end{definition}
We can update the label space $L$ by the vector spaces in the $S^n$-algebra $(\tilde{V},Z)$ and define the $\tilde{V}$-labelled alterfolds.

\begin{definition}\label{Def: Alterfold partition function}
We extend the partition function $Z$ on the bulk vectors $(B^{n+1},\mathcal{M}^{n})$ without time boundary to the partition function on $L$-labelled alterfolds without boundary as
$$Z(M^{n+1},B^{n+1},\mathcal{M}^{n}):=Z(B^{n+1},\mathcal{M}^{n}).$$ 
\end{definition}
Note that $A^{n+1}=\overline{M^{n+1}\setminus B^{n+1}}$. The extension of $Z$ is irrelevant to $A^{n+1}$. In particular the value of a $A$-colored manifold $M^{n+1}$ without boundary is constant 1.
$$Z(M^{n+1},\emptyset,\emptyset)=Z(\emptyset,\emptyset)=1.$$

We can compute the partition function $Z$ using surgery theory in topology based on the free change of $A$-color manifolds and the $\mathcal{S}^{k-1}$-relations in Def.~\ref{Def: S-relation}. We refer the readers to \cite{LMWW23a,LMWW23b} for the discussions on surgery theory in 2+1 alterfold TQFT.

\begin{definition}
Suppose $(E,F,S)$ is an $n$-alterfold $\partial E=\emptyset$. Suppose $\mathcal{S}$ is a stratified $(n-1)$-manifold with support $S$. We call the triple $\mathcal{E}:=(E,F,\mathcal{S})$ an alterfold boundary.
\end{definition}

\begin{definition}
An alterfold cobordism with an alterfold boundary $\mathcal{E}=(E, F,\mathcal{S})$ is a triple $\mathcal{M}^{n+1}:=(M^{n+1},B^{n+1},\mathcal{M}^n)$, $\mathcal{M}^{n+1}:=(M^{n+1},B^{n+1},|\mathcal{M}^n|)$ is an alterfold with boundary  $(E, F,|\mathcal{S}|)$ and $(B^{n+1},\mathcal{M}^n)$ is a space-time cobordism with time boundary $(F,\mathcal{S})$.
\end{definition}

\begin{definition}\label{Def: alterfold TQFT}
Given a local shape set $LS_{\bullet}$, an alterfold $(n+1)$-TQFT is a symmetric monoidal functor from the category of alterfold cobordisms to the category of vector spaces over a field $\mathbb{K}$.
More precisely, every alterfold boundary $\mathcal{E}=(E,F,\mathcal{S})$ is assigned a finite dimensional vector space $V_\mathcal{E}$, and $V_{(E,\emptyset,\emptyset)} \cong \mathbb{K}$.
Every alterfold $n+1$-cobordism is assigned to a linear transformation on the vector spaces.  The disjoint union and the gluing map correspond to the tensor and contraction respectively.
The map $Z:V_\emptyset \to \mathbb{K}$ is called the partition function. 
The alterfold TQFT is called unitary, if $Z$ is reflection positive.
\end{definition}

When the $E=F$, equivalently $\mathcal{S}=\emptyset$, we return to Atiyah's TQFT \cite{Ati88}.

\begin{theorem}\label{Thm: Alterfold TQFT}
Suppose $Z$ is a linear functional on labelled stratified manifolds with support $S^n$ over the field $\mathbb{C}$, satisfying the three conditions
\begin{enumerate}
    \item (RP) reflection positivity; (Def.~\ref{Def: Z RP})
    \item (HI) homeomorphic invariance; (Def.~\ref{Def: Z HI})
    \item (CF) complete finiteness. (Def.~\ref{Def: Z CF})
\end{enumerate}
Then we obtain an $n+1$ unitary alterfold TQFT for any $\zeta>0$ in Def.~\ref{Def: zeta}.

For a general field $\mathbb{K}$, we replace RP by strong semisimplicity, and then we obtain an $n+1$ alterfold TQFT.
\end{theorem}

\begin{proof}
Recall that the homeomorphic invariance of $Z(\mathcal{B})=Z(B^{n+1},\mathcal{M}^{n})$ is proved in Theorem \ref{Thm: Semisimple invariant}.
By Def.~\ref{Def: Alterfold partition function}, 
$$Z(M^{n+1},B^{n+1},\mathcal{M}^{n}):=Z(B^{n+1},\mathcal{M}^{n}),$$ 
the partition function $Z$ of the alterfold TQFT is irrelevant to the $A$-color part.
So it is also homeomorphic invariant.

Moreover, for any alterfold boundary $\mathcal{E}=(E,F,\mathcal{S})$, $\mathcal{F}=(F,\mathcal{S})$ is a time boundary of the space-time TQFT.
As the inner product induced by $Z$ is irrelevant to the $A$-color part, so the vector space $V_{\mathcal{E}}$ is isomorphic to $V_{\mathcal{F}}$.
By Theorem \ref{Thm: finite dimension}, the vector space is finite dimensional.
So we obtain an $n+1$ alterfold TQFT.
Furthermore if $Z$ is reflection positive and $\zeta>0$, by Theorem \ref{Thm: finite dimension}, the vector space is a finite dimensional Hilbert space.
So we obtain a unitary $n+1$ alterfold TQFT.
This extends the construction of the space-time TQFT in Theorem~\ref{Thm: space-time TQFT}.    
\end{proof}

A time boundary $(F,\mathcal{S})$ in the space-time TQFT could be extended to different alterfold boundary $(E,F,\mathcal{S})$ from different $A$-color manifolds $\overline{E\setminus F}$. 
They may correspond to different interpretations in different situations. 

\begin{remark}
One can study the idempotent completion of the local $D^{n+1}$-algebra of the $n+1$ alterfold TQFT and obtain an $n+1$-category with two 0-morphisms corresponding to the two colors $A$ and $B$. The $0$-morphisms of the $n$-category from $S^n$-algebra $(\tilde{V},Z)$ becomes $1$-morphisms from $A$ to $B$. The $k$-morphisms of the $n$-category from $(\tilde{V},Z)$ becomes $(k+1)$-morphisms.   
\end{remark}

\section{Examples}\label{Sec: Examples}
TQFT is an extremely fruitful theory which has been studied from various perspectives in mathematics and physics. We did not attempt to review the extensive literature. 
In this section, let us discuss some concepts in the lower dimensional alterfold TQFT, so that one can have a better understanding of the higher analogue.
Recall that we obtain an $n+1$ alterfold TQFT from an $S^n$ functional $Z$ is complete finite and strong semisimple or reflection positive,

When $n=0$, a $D^0$ algebra is a finite dimensional vector space $V$.
An $S^0$ functional is a (non-degenerate) inner product. If $Z$ is reflection positive, then $V$ is a Hilbert space. 
The invariant of an oriented compact 1-manifold is $dim(V)^{\#S^1}$.
The vector space of the $0+1$ TQFT with the time boundary consisting of $m$ points is the $m^{th}$ tensor power of $V$. The union for all $m$ is the Fock space of $V$. 

When $n=1$, a semisimple $D^1$ algebra is an associate algebra $\bigoplus_i M_{n_i}(\mathbb{K}).$ 
An $S^1$ functional is a trace. If $Z$ is reflection positive, then $V$ is a $C^*$-algebra. Let $d_j$ be the trace of the minimal idempotent $p_j$ of $M_{n_i}(\mathbb{K})$. Then we have the following $\mathcal{S}^{k-1}$-relations, $0\leq k\leq 2$, in the 1+1 alterfold TQFT:
\begin{center}
\begin{tikzpicture}
\draw [dashed] (-1,-1) rectangle (1,1);
\node at (0,0) {$B$};
\draw (-.5,-.5) rectangle (.5,.5);
\node at (.75,0) {$A$};
\node at (-.75,0) {$p_j$}; 
\node at (2,0) {$=d_j$}; 
\begin{scope}[shift={(4,0)}]
\node at (0,0) {$A$};
\draw [dashed] (-1,-1) rectangle (1,1);
\end{scope}
\end{tikzpicture}
\end{center}

\begin{center}
\begin{tikzpicture}
\node at (-.75,.8) {$A$};
\node at (0,.8) {$B$};
\node at (.75,.8) {$A$};
\draw (-.5,-1)--(-.5,1);
\draw (.5,-1)--(.5,1);
\node at  (-.75,0) {$p_j$};
\node at  (.25,0) {$\rho(p_j)$};
\draw [dashed] (-1,-1) rectangle (1,1);
\node at (2,0) {$=d_j^{-1}$};
\begin{scope}[shift={(4,0)}]
\node at (.75,.8) {$A$};
\node at (0,.8) {$B$};
\draw (-.5,-1)--++(0,.5)--++(1,0)--++(0,-.5);
\draw (-.5,1)--++(0,-.5)--++(1,0)--++(0,.5);
\node at  (-.75,.8) {$p_j$};
\node at  (-.75,-.8) {$p_j$};
\draw [dashed] (-1,-1) rectangle (1,1);
\end{scope}
\end{tikzpicture}
\end{center}

\begin{center}
\begin{tikzpicture}
\node at (0,.8) {$B$};
\draw [dashed] (-1,-1) rectangle (1,1);
\node at (2,0) {$=\displaystyle\sum_j d_j$};
\begin{scope}[shift={(4,0)}]
\node at (0,.8) {$B$};
\draw [dashed] (-1,-1) rectangle (1,1);
\draw (-.5,-.5) rectangle (.5,.5);
\node at (.25,0) {$p_j$}; 
\end{scope}
\end{tikzpicture}
\end{center}
Applying the $\mathcal{S}^{k-1}$-relations to the normal microbundle of a $k$-simplex of a triangulation of an oriented surface $S$, we obtain the invariant 
$$Z(S)=\sum_j d_j^{n_2-n_1+n_0}=\sum_j d_j^E(S).$$
where $n_k$ is the number of $k$-simplices and $E(S)$ is the Euler number of $S$.
The complete finiteness condition ensure the invariant to be a finite sum.
The strong semisimple condition ensure $d_j\neq 0$.
We can release the complete finiteness condition if the sum convergences.

In particular, if we take the semisimple associate algebra to be the group algebra $\mathcal{L}G$ for a finite group $G$ with the trace from the regular representation $L^2(G)$. Then the trace of the group element $g\in G$ is $\delta_{g,1}|G|$, diagrammatically,
\begin{center}
\begin{tikzpicture}
\draw [dashed] (-1,-1) rectangle (1,1);
\node at (0,0) {$B$};
\draw (-.5,-.5) rectangle (.5,.5);
\node at (.75,0) {$A$};
\node at (-.75,0) {$g$}; 
\node at (2,0) {$=\delta_{g,1}|G|$}; 
\begin{scope}[shift={(4,0)}]
\node at (0,0) {$A$};
\draw [dashed] (-1,-1) rectangle (1,1);
\end{scope}
\end{tikzpicture}
\end{center}
Moreover, the group elements forms an orthonormal basis, we have that
\begin{center}
\begin{tikzpicture}
\node at (-.75,.8) {$A$};
\node at (0,.8) {$B$};
\node at (.75,.8) {$A$};
\draw (-.5,-1)--(-.5,1);
\draw (.5,-1)--(.5,1);
\draw [dashed] (-1,-1) rectangle (1,1);
\node at (2,0) {$=\sum_{g\in G}$};
\begin{scope}[shift={(4,0)}]
\node at (.75,.8) {$A$};
\node at (0,.8) {$B$};
\draw (-.5,-1)--++(0,.5)--++(1,0)--++(0,-.5);
\draw (-.5,1)--++(0,-.5)--++(1,0)--++(0,.5);
\node at  (-.75,.8) {$g$};
\node at  (-.75,-.8) {$g^{-1}$};
\draw [dashed] (-1,-1) rectangle (1,1);
\end{scope}
\end{tikzpicture}
\end{center}
The minimal projection $p_j$ of $\mathcal{L}G$ corresponds to an irreducible representation $V_j$ of dimensional $d_j$. By Peter-Weyl Theorem, the resolution of the identity of $\mathcal{L}G$ has $d_j$ minimal projections equivalent to $p_j$. So

\begin{center}
\begin{tikzpicture}
\node at (0,.8) {$B$};
\draw [dashed] (-1,-1) rectangle (1,1);
\node at (2,0) {$=\displaystyle\sum_j d_j$};
\begin{scope}[shift={(4,0)}]
\node at (0,.8) {$B$};
\draw [dashed] (-1,-1) rectangle (1,1);
\draw (-.5,-.5) rectangle (.5,.5);
\node at (.25,0) {$p_j$}; 
\end{scope}
\node at (6,0) {$=$};
\begin{scope}[shift={(8,0)}]
\node at (0,.8) {$B$};
\draw [dashed] (-1,-1) rectangle (1,1);
\draw (-.5,-.5) rectangle (.5,.5);
\end{scope}
\end{tikzpicture}
\end{center}
We can evaluate the partition function $Z(S)$ by these three relations in terms of group elements, and then derive that
$$Z(S)=\#\hom(\pi_1(S), G) |G|^{E(S)-1},$$
where $\pi_1(S)$ is the fundamental group of $S$.
Then we obtain the Mednykh's formula \cite{Med78}:
$$\#\hom(\pi_1(S), G)) |G|^{E(S)-1} = \sum_j \dim V_j^{E(S)}.$$
summing over irreducible representations $V_j$ of $G$.
We refer the readers to \cite{FreFra93,MulYu05,TurTur06,LauPfe09,Sny17} for further discussions for the Mednykh's formula and related results on 1+1 TQFT.

When $n=2$, there have been extensive studies on the 2+1 TQFT, since the fundamental work of the Witten-Reshetikhin-Turaev TQFT and Turaev-Viro TQFT in \cite{Wit88,ResTur91,TurVir92}.
The $D^2$ algebras were studied as planar algebras by Jones \cite{Jon21}. The $S^2$ algebras were studied as spherical planar algebras. If $Z$ is reflection positive and complete finite, then one obtains subfactor planar algebras of finite depth. 
(A subfactor planar algebra with infinite depth has an $S^2$ function which is 2-finite, but not 1-finite. Its idempotent completion has infinitely many 1-morphisms in the 2-category.)
Its idempotent completion is a spherical 2-category $\mathscr{C}$ or a multi-fusion category.
The invariant of a $B$-colored 3-manifold equals to its Turaev-Viro invariant \cite{TurVir92}.

The 2+1 alterfold TQFT has been studied in \cite{LMWW23a,LMWW23b}.
In this 2+1 alterfold TQFT, the braided tensor category of 2-morphisms with link boundary $S^1$ (without stratification) is the Drinfeld center \cite{Dri86} of $\mathscr{C}$. It is considered as the category of point excitations in physics.
Moreover, both the Turaev-Viro TQFT of $\mathscr{C}$ and the  Reshetikhin-Turaev TQFT \cite{ResTur91} of the Drinfeld center of $\mathscr{C}$ can be embedded into the 2+1 alterfold TQFT by the blow up procedures. 
The skeleton of the triangulation in the Turaev-Viro TQFT blows up to $B$-color normal microbundles. (The 2D boundary surface has contractible $B$-color 2-discs.)
The link in the Reshetikhin-Turaev TQFT of blows up to $A$-color normal microbundles. (The 2D boundary surface have contractible $A$-color 2-discs.) 
We can reverse the blow up procedures by shrinking the $A/B$ color bulk in the 2+1 alterfold TQFT.
Their tensor product has the shape of two closed strings which can merge into one closed string by an isometry.

Kitaev's toric code is a square lattice model on the torus \cite{Kit03}. 
It can be generalized to lattices of a general shape on a surface. 
Every edge has a qubit, a vector in the two-dimensional Hilbert space. The configuration space of a lattice is the tensor product of the qubit spaces for all edges. (We omit the label at the vertex or consider the label as the scalar 1 of the one dimensional Hilbert space $\mathbb{C}$.)
Every vertex $v$ has a vertex operator $A_v$ given by the tensor product of Pauli $Z$ on the nearest qubits. Every plaquette $p$ has a plaquette operator $B_p$ given by the tensor product of Pauli $X$ on nearest the qubits. The local vertex operators and plaquette operators commute. 
The Hamiltonian $H$ is
$$H=-\sum_{v}A_v-\sum_{p}B_p,$$
summing over all vertices and plaquettes.
The ground state of $H$ is unique for any lattice on $S^2$. 
The corresponding vector $\Omega$ is the sum of all loops on the lattice, where a loop has labels $\ket{0}$ or $\ket{1}$ on the edges, such that the union of all edges labelled by $\ket{1}$ have no boundary. 
As the configuration space is a Hilbert space, we consider the vector $\Omega$ as a linear functional $Z$ on the configuration space.

The linear functional $Z$ takes value $1$ on a loop and value $0$ for other labels.
It satisfies the three conditions (RP), (HI) and (CF). We obtain a spherical fusion category $Rep(\mathbb{Z}_2)$ and the 2+1 alterfold TQFT from $Z$ by Theorem \ref{Thm: Alterfold TQFT}, where $Rep(\mathbb{Z}_2)$ is the representation category of the group $\mathbb{Z}_2$ with two invertible objects $\mathbbm{1}$ and $g$.

We can refine the lattice by blowing up a vertex to a face and changing the position point of every edge to a vertex, see an example in Fig.~\ref{Fig: Refined Lattices} for the refinement of a square lattice on the torus. (The refined lattice is considered as a quantized graph to construct quantum error correcting codes in Page 5 in \cite{Liu19qecc} based on the Quon language \cite{JLW17}.)
Then the vertex operator $A_v$ becomes a plaquette operator on the refined lattice.
So the group of homeomorphisms of a lattice is enlarged and an $A_v$ operator can be identified with an $B_p$ operator under a homeomorphism.

We consider the normal microbundle of a vertex in the refined lattice as a parameterized disc with four points on the boundary circle.
We consider the basis of the qubit label space as the two diagrams of non-intersecting strings with four boundary points.
Then a fully labelled refined lattice is the underline surface with a stratification of non-intersecting closed strings. 
We define its $Z$ value to be $\sqrt{2}^m$, where $m$ is the number of closed strings. It coincides with the linear functional $Z$ on the original lattice model.
By Theorem \ref{Thm: Alterfold TQFT}, we obtain the Ising category from $Z$, beyond the $Rep(\mathbb{Z}_2)$. (We will apply this idea to construct the $S^3$ functional in Def.~\ref{Def: Ising Z} and derive a unitary spherical 3-category of Ising type. )

Moreover, we can recover the configuration space of a lattice, i.e., a stratified 2-manifold $\mathcal{M}^2$, from the vector space of the $2+1$ alterfold TQFT with alterfold boundary $\mathcal{E}=(E,F,\mathcal{S})$, where $E=M^2$; $F$ is a $\varepsilon$-neighbourhood of $M^1$; and $\mathcal{S}$ is the boundary of $F$ without stratification. Each connected components of the $A$-color region of $\mathcal{E}$ is a plaquette $P$ and we recover the local plaquette operator as $P\times [0,1]$ with labelled by $g$-colored $\partial P\times \{1/2\}$ acting on the boundary plaquette $P$.
The ground state space of a lattice $\mathcal{M}^2$ is isomorphic to the vector space of $B$-color boundary $M^2$ without stratification in the $2+1$ alterfold TQFT, which is the vector space with boundary $M^2$ in Atiyah's TQFT. The ground state space is independent of the choice of the lattice up to isomorphism.

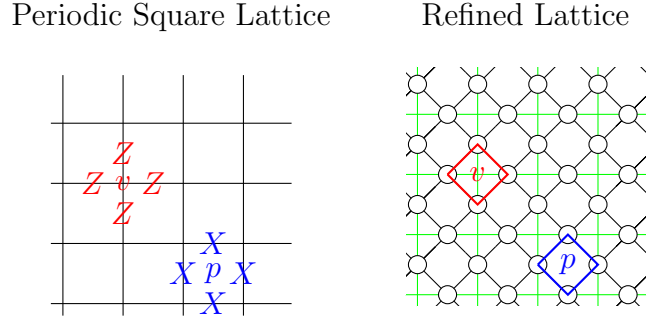
\begin{figure}[h] \begin{center}
\begin{tabular}{cc}
Periodic Square Lattice & Refined Lattice\\
\begin{tikzpicture}
\begin{scope}[scale=.8]
\foreach \x in {0,1,2,3}{
\draw (\x,-.2)--++(0,4);
}
\foreach \y in {0,1,2,3}{
\draw (-.2,\y)--++(4,0);
}
\node at (1,2) {$\textcolor{red}{v}$};
\node at (1-.5,2) {$\textcolor{red}{Z}$};
\node at (1+.5,2) {$\textcolor{red}{Z}$};
\node at (1,2-.5) {$\textcolor{red}{Z}$};
\node at (1,2+.5) {$\textcolor{red}{Z}$};
\node at (2.5,.5) {$\textcolor{blue}{p}$};
\node at (2.5-.5,.5) {$\textcolor{blue}{X}$};
\node at (2.5+.5,.5) {$\textcolor{blue}{X}$};
\node at (2.5,.5-.5) {$\textcolor{blue}{X}$};
\node at (2.5,.5+.5) {$\textcolor{blue}{X}$};
\end{scope}
\end{tikzpicture}
&
\begin{tikzpicture}
\begin{scope}[scale=.8]
\foreach \x in {0,1,2,3}{
\draw[green] (\x,-.2)--++(0,4);
}
\foreach \y in {0,1,2,3}{
\draw[green] (-.2,\y)--++(4,0);
}
\foreach \x in {0,1,2}{
\foreach \y in {0,1,2,3}{
\draw (\x+.3,\y-.2)--++(1,1);
\draw (\x-.2,\y+.7)--++(1,-1);
\draw (3+.3,\y-.2)--++(.5,.5);
\draw (0-.2,\y+.3)--++(.5,.5);
}}
\foreach \y in {1,2,3}{
\draw (3-.2,\y+.7)--++(1,-1);
}
\draw (3-.2,0+.7)--++(.9,-.9);
\foreach \x in {0,1,2,3}{
\draw (\x+.7,3.8)--++(.1,-.1);
}
\foreach \x in {0,1,2,3}{
\foreach \y in {0,1,2,3}{
\fill[white] (\x+.5,\y) circle (.15);
\draw (\x+.5,\y) circle (.15);
\fill[white] (\x,\y+.5) circle (.15);
\draw (\x,\y+.5) circle (.15);
}}
\fill[white] (0,-.5) rectangle (3.8,-.2);
\fill[white] (0,3.8) rectangle (3.8,4.3);
\fill[white] (-1,0) rectangle (-.2,3.8);
\fill[white] (3.8,0) rectangle (4.6,3.8);
\node at (1,2) {$\textcolor{red}{v}$};
\draw[red,thick](1-.5,2)--++(.5,-.5) --++(.5,.5)--++(-.5,.5)--++(-.5,-.5);
\node at (2.5,.5) {$\textcolor{blue}{p}$};
\draw[blue,thick](2.5-.5,.5)--++(.5,-.5) --++(.5,.5)--++(-.5,.5)--++(-.5,-.5);
\end{scope}
\end{tikzpicture}
\end{tabular}
\caption{Refined Lattices}\label{Fig: Refined Lattices}
\end{center} \end{figure}

If we define the $S^2$ functional $Z$ on $m$ non-intersecting circles in $S^2$ as $\delta^{m}$, then the Jones index Theorem \cite{Jon83,Jon21} implies that 
$Z$ is reflection positive, iff $\delta=2\cos \frac{\pi}{2+\ell}$, $\ell\in \mathbb{N}_+$ or $\delta \geq 2$. In addition, $Z$ is complete finite, iff $\delta=2\cos \frac{\pi}{2+\ell}$. We obtain the unitary quotient of the representation category of quantum $SU(2)$ at level $\ell$. The Ising category is for $\ell=2$.

In general, for the Levin-Wen model \cite{LevWen05} from a spherical multi-fusion category $\mathscr{C}$, we can define the $S^2$ functional $Z$ of a string-net on $S^n$ as its evaluation in the spherical category $\mathscr{C}$. From the functional $Z$ of string nets, we can recover the spherical category $\mathscr{C}$ and construct a 2+1 alterfold TQFT by Theorem \ref{Thm: Alterfold TQFT}. Moreover, the plaquette operator, the Hamiltonian and the ground states etc of the Levin-Wen model can be recovered from the 2+1 alterfold TQFT, see Remark 4.25 in \cite{LMWW23b}.
In particular, the reflection positivity condition of the Hamiltonian of the Levin-Wen model is proved in Theorem 3.2 \cite{JafLiu20}, generalizing the geometric proof of RP in Section 7 in \cite{JafLiu17}. We will study the Hamiltonian of the lattice model, the ground states and relevant properties, such as reflection positivity, from the functional integral point of view in any dimension in the coming paper \cite{Liu24-preparation}.

When $n=3$, Witten constructed a 3+1 TQFT \cite{Wit88} which captures Donaldson's invariant of smooth 4-manifolds \cite{Don83,Don83b,Don90}. The Turaev-Viro state sum construction has been generalized to construct 3+1 TQFT using a braided fusion category in \cite{CraYet93,CKY97,Cui19}. Douglas and Reutter proposed a commonly accepted definition of fusion 2-categories and constructed a 3+1 TQFT as a state sum in \cite{DogReu18}, see further discussions and references therein.
Similar to the 2+1 theory, for a spherical fusion 2-category $\mathscr{C}$, one can define an $S^3$ functional by the evaluating the string-net diagrams on $S^3$ in $\mathscr{C}$. Then one recovers the category $\mathscr{C}$ and the 3+1 TQFT from Theorem \ref{Thm: Alterfold TQFT}.

The Dijkgraaf-Witten TQFT \cite{DijWit90} is an $n$ dimensional TQFT coming from a cocycle in the group cohomology $H^n(G,\mathbb{K}^{\times})$ of a finite group $G$. One can define the $S^{n-1}$ function of an $n$-simplex, whose 1-simplices are labelled by group elements, as the value of the cocycle. Then one can recover the TQFT from Theorem \ref{Thm: Alterfold TQFT}.

\section{Higher Braid Statistics}\label{Sec: Higher Braid Statistics}
In dimension 2+1, the braid statistics of the point excitations of a the two-dimensional lattice model, such as the Levin-Wen model \cite{LevWen05}, is captured by the Drinfeld center of the spherical 2-category. 
In dimension 3+1 or higher, the category of point excitations is usually studied as a symmetric fusion category. The symmetric fusion category is a representation of a group or a super group proved by Deligne in \cite{Del02}. The particles are classified as bosons or fermions according to the braid statistics.
In that sense, the lower dimension topology is more interesting due to emergence of anyons with general braid statistics.
The braid statistics seems less interesting in higher dimensions.
However, the topology in dimension 4 are much more complicated. It is a paradox.

We give a conceptual explanation to resolve the paradox. 
The point excitation of the $n$-dimensional lattice model is an $n$-morphism with the $B$-color link boundary $S^{n-1}$ in the $n+1$ alterfold TQFT. The type of the $n$-morphism is an alterfold $(D^n,B^n,M^{n-1})$, rather than a point. We use the $A$-color part $A^n$ to denoted its type for short. 
When $n=3$ and $A^3$ is a solid torus, there are infinitely many ways to fuse the two point excitations of such type, because the two solid torus can be braided in the $D^3$, such as the solid Hopf link. The worldsheet is even more complicated, as one solid torus can move along the other. 
This leads to extremely rich higher braid statistics of membranes.
Based on this observation, we propose the following definition of the local center to study the higher braid statistics in the future.

\begin{definition}
For a spherical $n$-category $\mathscr{C}$ (from a complete finite and strong semisimple $S^n$ functional $Z$), we define its local center as the category of point excitations of the $n$-dimensional lattice model.
Its objects are $n$-morphisms with $B$-color link boundary $S^{n-1}$ and its morphisms are $(n+1)$-morphisms in the $n+1$ alterfold TQFT. For two objects of type $A^n_1$ and $A^n_2$, we obtain a fusion for any non-intersecting embedding of $A^n_1$ and $A^n_2$ in $D^n$.       
\end{definition}

\section{A 3+1 Alterfold TQFT of Ising type}\label{Sec: 3+1 Ising}

In this section, we illustrate our theory on the functional integral construction of TQFT in a concrete example. We construct an $S^3$ functional $Z$, which is reflection positive and complete finite. We obtain a unitary spherical 3-category of Ising type and a non-invertible 3+1 unitary alterfold TQFT from Theorem \ref{Thm: Alterfold TQFT}. 
Moreover, we compute the 20j-symbols and verify the $(3,3)$, $(2,4)$ and $(1,5)$ Pachner moves \cite{Pan91}. It is a one-dimensional higher analogue of the $6j$ symbols and pentagon equations. The number 20 is given by the number of 1-simplices and 2-simplices in the 4-simplex $\Delta^4$.
The 20j-symbols could be used to compute a scalar invariant of 2-knots in PL 4-manifolds explicitly. 

Recall that the Ising 2-category can be derived from an $S^2$ functional $Z$ which takes value $\sqrt{2}^m$ for $m$ non-intersecting closed circles in $S^2$.
It seems natural to generalize it to an $S^3$ functional $Z$ on non-intersecting surfaces in $S^3$, which is $\sqrt{2}$ for when the surface is a sphere.
After checking certain compatible conditions for the fusion rule and quantum dimensions, we obtain the value of a surface with Euler number $e$ to be $2^{1-\frac{e_i}{4}}.$ This $S^3$ functional $Z$ verifies reflection positivity, but it has infinitely many 1-morphisms up to equivalence. 
To achieve the complete finiteness condition, we allow the surfaces to intersect. 
Finally we obtain the following formula for the $S^3$ functional $Z$.

\begin{definition}\label{Def: Ising Z}
We define the $S^3$ functional $Z$ on stratified 3-manifold $(S^3, \cup_i S_i)$ where $\{S_i\}_{i\in I}$ are PL surfaces in $S^3$ which may intersect transversely and be unoriented. It has the stratification:
\begin{enumerate}
    \item $M^3$ is $S^3$; 
    \item $M^2$ is the union of PL surfaces;
    \item $M^1$ is the union of intersection curves of every two surfaces;
    \item $M^0$ is the union of intersection points of every three surfaces.
\end{enumerate}
The intersection of every four surfaces is empty.
Let $e_i$ be the Euler number of $S_i$. We define 
\begin{equation}
\label{partfun}
Z(S^3, \cup_i {S_i})=\prod_{i=1}^n 2^{1-\frac{e_i}{4}}.
\end{equation}
\end{definition}
The Euler number is homeomorphic invariant, so $Z$ is homeomorphic invariant. 
Moreover, the Euler number does not change under reflection, so $Z$ is Hermitian.  
Now let us prove the reflection positivity and complete finiteness; construct its simplical $k$-morphisms; and compute its 20j-symbols.

For any point $p\in M^{k}\setminus M^{k-1}$, it has a regular chart as shown in Fig.~\ref{fig: 3D Ising regular chart}, for $k=3,2,1,0$. Each of the stratified manifold on the boundary $S^2$ corresponds to the only local shape in $LS_k$ as in Def.~\ref{Def:LS}-\ref{Def:LL}:

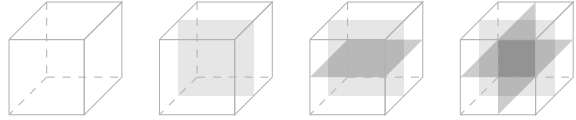
\begin{figure}[H]
    \centering
    \begin{tikzpicture}
     \begin{scope}[shift={(0,0)}]
        
            \draw[gray!60] (1,0)--(1.5,.5)--(1.5,-.5)--(1,-1)--(1,0);
            \draw[gray!60] (0,0)--(1,0)--(1.5,.5)--(.5,.5)--(0,0);
            \draw[gray!60] (0,0)--(1,0)--(1,-1)--(0,-1)--(0,0);
            \draw[gray!60,dashed] (.5,.5)--(.5,-.5);
            \draw[gray!60,dashed] (.5,-.5)--(1.5,-.5);
            \draw[gray!60,dashed] (.5,-.5)--(0,-1);
        
        \end{scope}

         \begin{scope}[shift={(2,0)}]

            \begin{scope}[shift={(.25,.25)}]
                
            \fill[gray!20] (0,0)--(1,0)--(1,-1)--(0,-1);
            \end{scope}
            \draw[gray!60] (1,0)--(1.5,.5)--(1.5,-.5)--(1,-1)--(1,0);
            \draw[gray!60] (0,0)--(1,0)--(1.5,.5)--(.5,.5)--(0,0);
            \draw[gray!60] (0,0)--(1,0)--(1,-1)--(0,-1)--(0,0);
            \draw[gray!60,dashed] (.5,.5)--(.5,-.5);
            \draw[gray!60,dashed] (.5,-.5)--(1.5,-.5);
            \draw[gray!60,dashed] (.5,-.5)--(0,-1);
        
        \end{scope}

        \begin{scope}[shift={(4,0)}]

            \begin{scope}[shift={(.25,.25)}]
                
            \fill[gray!20] (0,0)--(1,0)--(1,-1)--(0,-1);
            \end{scope}
            \begin{scope}[shift={(0,-.5)}]
            \fill[opacity=.20] (0,0)--(1,0)--(1.5,.5)--(.5,.5);
            \end{scope}
            \draw[gray!60] (1,0)--(1.5,.5)--(1.5,-.5)--(1,-1)--(1,0);
            \draw[gray!60] (0,0)--(1,0)--(1.5,.5)--(.5,.5)--(0,0);
            \draw[gray!60] (0,0)--(1,0)--(1,-1)--(0,-1)--(0,0);
            \draw[gray!60,dashed] (.5,.5)--(.5,-.5);
            \draw[gray!60,dashed] (.5,-.5)--(1.5,-.5);
            \draw[gray!60,dashed] (.5,-.5)--(0,-1);
        
        \end{scope}

        \begin{scope}[shift={(6,0)}]

            \begin{scope}[shift={(.25,.25)}]
                
            \fill[gray!20] (0,0)--(1,0)--(1,-1)--(0,-1);
            \end{scope}
            \begin{scope}[shift={(0,-.5)}]
            \fill[opacity=.20] (0,0)--(1,0)--(1.5,.5)--(.5,.5);
            \end{scope}

            \begin{scope}[shift={(-.5,0)}]
                \fill[opacity=.20] (1,0)--(1.5,.5)--(1.5,-.5)--(1,-1);
            \end{scope}
            \draw[gray!60] (1,0)--(1.5,.5)--(1.5,-.5)--(1,-1)--(1,0);
            \draw[gray!60] (0,0)--(1,0)--(1.5,.5)--(.5,.5)--(0,0);
            \draw[gray!60] (0,0)--(1,0)--(1,-1)--(0,-1)--(0,0);
            \draw[gray!60,dashed] (.5,.5)--(.5,-.5);
            \draw[gray!60,dashed] (.5,-.5)--(1.5,-.5);
            \draw[gray!60,dashed] (.5,-.5)--(0,-1);
        
        \end{scope}
\end{tikzpicture}
    
    \caption{The regular chart of a point in $S^3$.}
    \label{fig: 3D Ising regular chart}
\end{figure}


When we cut $S^3$ into two parts $D^{2}_{\pm}$ by the equator $S^2$, the restriction of $M^3$ on the boundary $S^2$ contains PL curves $C=\{C_j\ : 1\leq j \leq m\}$, $|C_j|\sim S^1$.
The vector space $\tilde{V}(S^2,C)$ is spanned by stratified manifolds $\mathcal{M}$ with support $D^2$ and boundary $C$ and $M^2$ consists of surfaces with boundary $C$.

We have the following relations in $K_Z$. Any ambient isotopy of the surfaces will not change $Z$ and we can split surfaces away from each other. In particular, we can move one plane across another plane, or across a line or a point which is intersection of planes, as shown in Fig.~\ref{fig:Intersection}. The last equation is a one-dimensional higher analogue of the Yang-Baxter equation.
Moreover, for each connected surface, we can change its shape by multiplying the scalar $2^{-\frac{e}{4}}$ according to the change of the Euler number $e$.

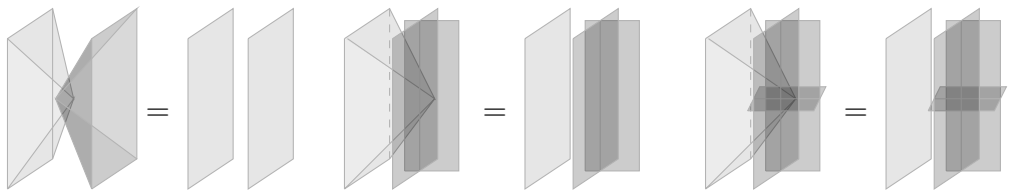
\begin{figure}[H]
    \centering
    \begin{tikzpicture}
    \begin{scope}[scale=.8]
    
    \begin{scope}

         \fill[gray!15](0,0)--(0,-2.5)--(1.1,-1);
         \fill[gray!15](.75,.5)--(.75,.5-2.5)--(1.1,-1);
         \fill[gray!20](0,-2.5)--(.75,.5-2.5)--(1.1,-1);
         \fill[gray!20](0,0)--(.75,.5)--(1.1,-1);

         \draw[gray!60](0,0)--(.75,.5)--(1.1,-1)--(0,0);
         \draw[gray!60](0,-2.5)--(.75,.5-2.5)--(1.1,-1)--(0,-2.5);
         \draw[gray!60](0,0)--(0,-2.5);
         \draw[gray!60](.75,.5)--(.75,.5-2.5);

    \end{scope}

    \begin{scope}[shift={(1.4,0)}]

         \fill[opacity=.15](0,0)--(0,-2.5)--(-.6,-1);
         \fill[opacity=.15](.75,.5)--(.75,.5-2.5)--(-.6,-1);
         \fill[opacity=.20](0,-2.5)--(.75,.5-2.5)--(-.6,-1);
         \fill[opacity=.20](0,0)--(.75,.5)--(-.6,-1);

         \draw[gray!60](0,0)--(.75,.5)--(-.6,-1)--(0,0);
         \draw[gray!60](0,-2.5)--(.75,.5-2.5)--(-.6,-1)--(0,-2.5);
         \draw[gray!60](0,0)--(0,-2.5);
         \draw[gray!60](.75,.5)--(.75,.5-2.5);

    \end{scope}

    \begin{scope}[shift={(3,0)}]

         \node at (-.5,-1.25){$=$};
         
         \fill[gray!20](0,0)--(0,-2.5)--(.75,.5-2.5)--(.75,.5)--(0,0);
         \draw[gray!60](0,0)--(0,-2.5)--(.75,.5-2.5)--(.75,.5)--(0,0);

         \begin{scope}[shift={(1,0)}]
             \fill[gray!20](0,0)--(0,-2.5)--(.75,.5-2.5)--(.75,.5)--(0,0);
         \draw[gray!60](0,0)--(0,-2.5)--(.75,.5-2.5)--(.75,.5)--(0,0);
         \end{scope}

    \end{scope}

    \end{scope}

    \begin{scope}[scale=.8, shift={(6,0)}]
    
    \begin{scope}
           
         \fill[gray!15](0-.4,0)--(0-.4,-2.5)--(1.1,-1);
         \fill[gray!15](.75-.4,.5)--(.75-.4,.5-2.5)--(1.1,-1);
         \fill[gray!20](0-.4,-2.5)--(.75-.4,.5-2.5)--(1.1,-1);
         \fill[gray!20](0-.4,0)--(.75-.4,.5)--(1.1,-1);

         \draw[gray!60](0-.4,0)--(.75-.4,.5)--(1.1,-1)--(0-.35,0);
         \draw[gray!60](0-.4,-2.5)--(.75-.4,.5-2.5)--(1.1,-1)--(0-.4,-2.5);
         \draw[gray!60](0-.4,0)--(0-.4,-2.5);
         \draw[gray!60,dashed](.75-.4,.5)--(.75-.4,.5-2.5);

         \fill[opacity=.20](0.6,-2.2)--(1.5,-2.2)--(1.5,.3)--(0.6,.3);

         \draw[gray!60](0.6,-2.2)--(1.5,-2.2)--(1.5,.3)--(0.6,.3);

         \draw[gray!60](0.6,.3)--(0.6,-2.2);

         \draw[gray!60](0.85,-2.2)--(0.85,.3);
          \begin{scope}[shift={(.4,0)}]
             \fill[opacity=.20](0,0)--(0,-2.5)--(.75,.5-2.5)--(.75,.5)--(0,0);
         \draw[gray!60](0,0)--(0,-2.5)--(.75,.5-2.5)--(.75,.5)--(0,0);
         \end{scope}
         
    \end{scope}

    \begin{scope}[shift={(3,0)}]

         \node at (-.9,-1.25){$=$};

         \begin{scope}[shift={(-.4,0)}]
        
         \fill[gray!20](0,0)--(0,-2.5)--(.75,.5-2.5)--(.75,.5)--(0,0);
         \draw[gray!60](0,0)--(0,-2.5)--(.75,.5-2.5)--(.75,.5)--(0,0);
     
         \end{scope}

            \fill[opacity=.20](0.6,-2.2)--(1.5,-2.2)--(1.5,.3)--(0.6,.3);

         \draw[gray!60](0.6,-2.2)--(1.5,-2.2)--(1.5,.3)--(0.6,.3)--(0.6,-2.2);
          \draw[gray!60](0.85,-2.2)--(0.85,.3);
         
          \begin{scope}[shift={(.4,0)}]
             \fill[opacity=.20](0,0)--(0,-2.5)--(.75,.5-2.5)--(.75,.5)--(0,0);
         \draw[gray!60](0,0)--(0,-2.5)--(.75,.5-2.5)--(.75,.5)--(0,0);
         
         \end{scope}
         
    \end{scope}

    \end{scope}

     \begin{scope}[scale=.8, shift={(12,0)}]
    
    \begin{scope}
           
          \fill[gray!15](0-.4,0)--(0-.4,-2.5)--(1.1,-1);
         \fill[gray!15](.75-.4,.5)--(.75-.4,.5-2.5)--(1.1,-1);
         \fill[gray!20](0-.4,-2.5)--(.75-.4,.5-2.5)--(1.1,-1);
         \fill[gray!20](0-.4,0)--(.75-.4,.5)--(1.1,-1);

         \draw[gray!60](0-.4,0)--(.75-.4,.5)--(1.1,-1)--(0-.35,0);
         \draw[gray!60](0-.4,-2.5)--(.75-.4,.5-2.5)--(1.1,-1)--(0-.4,-2.5);
         \draw[gray!60](0-.4,0)--(0-.4,-2.5);
         \draw[gray!60,dashed](.75-.4,.5)--(.75-.4,.5-2.5);

         \fill[opacity=.20](0.6,-2.2)--(1.5,-2.2)--(1.5,.3)--(0.6,.3);
         
        \fill[opacity=.20](0.3,-1.2)--(1.4,-1.2)--(1.6,-0.8)--(0.5,-.8);
        \draw[gray!60](0.3,-1.2)--(1.4,-1.2)--(1.6,-0.8)--(0.5,-.8)--(0.3,-1.2);
        \draw[gray!60](.4,-1)--(1.5,-1);
         
         \draw[gray!60](0.6,-2.2)--(1.5,-2.2)--(1.5,.3)--(0.6,.3);

         \draw[gray!60](0.6,.3)--(0.6,-2.2);

         \draw[gray!60](0.85,-2.2)--(0.85,.3);
         
          \begin{scope}[shift={(.4,0)}]
             \fill[opacity=.20](0,0)--(0,-2.5)--(.75,.5-2.5)--(.75,.5)--(0,0);
         \draw[gray!60](0,0)--(0,-2.5)--(.75,.5-2.5)--(.75,.5)--(0,0);

         \end{scope}

    \end{scope}

    \begin{scope}[shift={(3,0)}]

         \node at (-.9,-1.25){$=$};
         
         \begin{scope}[shift={(-.4,0)}]
        
         \fill[gray!20](0,0)--(0,-2.5)--(.75,.5-2.5)--(.75,.5)--(0,0);
         \draw[gray!60](0,0)--(0,-2.5)--(.75,.5-2.5)--(.75,.5)--(0,0);
     
         \end{scope}

            \fill[opacity=.20](0.6,-2.2)--(1.5,-2.2)--(1.5,.3)--(0.6,.3);

         \draw[gray!60](0.6,-2.2)--(1.5,-2.2)--(1.5,.3)--(0.6,.3)--(0.6,-2.2);

         \fill[opacity=.20](0.3,-1.2)--(1.4,-1.2)--(1.6,-0.8)--(0.5,-.8);
        \draw[gray!60](0.3,-1.2)--(1.4,-1.2)--(1.6,-0.8)--(0.5,-.8)--(0.3,-1.2);
         \draw[gray!60](.4,-1)--(1.5,-1);

        \draw[gray!60](0.85,-2.2)--(0.85,.3);
         
          \begin{scope}[shift={(.4,0)}]
             \fill[opacity=.20](0,0)--(0,-2.5)--(.75,.5-2.5)--(.75,.5)--(0,0);
         \draw[gray!60](0,0)--(0,-2.5)--(.75,.5-2.5)--(.75,.5)--(0,0);
         
         \end{scope}

    \end{scope}

    \end{scope}
    
    \end{tikzpicture}
    \caption{Intersection of surfaces}
    \label{fig:Intersection}
\end{figure}

\begin{notation}
For surfaces $\{S_i\}_{i \in I}$ in $D^3$ with boundary $C=\{C_j\}_{j\in J}$. Let $|C|$ be the number closed curves. 
We define its connected type as $T=\{T_i: i\in I\}$, where $T_i=\{j \in J : C_j \subseteq \partial S_i\}$, which is a partition of the set $J$. 
\end{notation}

\begin{proposition}
\label{prop:vectorfinite}
The $S^3$ functional $Z$ is 3-finite, namely the vector space $\tilde{V}(S^2,C)$ is finite dimensional.
\end{proposition}
\begin{proof}
For each connected surface in $D^3$, we can change its shape by multiplying the scalar $2^{-\frac{e}{4}}$ according to the change of the Euler number $e$.
Thus surfaces with the same connected shape, they are scalar multiples of each other. 
So the dimension vector space $\tilde{V}(S^2,C)$ is bounded by the number of connected types, which is the number of partitions of $J$. So it is finite dimensional.
\end{proof}

When $|C|=0,1,2,3$, the number of partitions is $0,1,2,5$. An isotopy of the boundary $C$ induces an isometry of the vector spaces. We only need to consider the boundary $C$ with non-intersecting curves.

\begin{proposition}\label{Prop: 0-curve}
When $C$ has no curve, $\tilde{V}(S^2,C)$ is one-dimensional. Consequently, there is only one 0-morphism.
\end{proposition}

\begin{proof}
Any closed surface with Euler number $e$ reduces to a scalar $2^{1-\frac{e}{4}}$.
\end{proof}

\begin{proposition}\label{Prop: 1-curve}
When $|C|=1$, $\tilde{V}(S^2,C)$ is one-dimensional.
\end{proposition}

\begin{proof}
Suppose the curve is the boundary of a surface $S$. We can change $S$ to the disc by multiplying the scalar $2^{-\frac{e}{4}}$ according to the change of the Euler number $e$. 
\end{proof}

\begin{proposition}\label{Prop: 2-curve}
When $|C|=2$, $\tilde{V}(S^2,C)$ is two-dimensional.
\end{proposition}

\begin{proof}
There are two connected types depending on whether the two curves are connected by a surface or not. They form a basis. 
\end{proof}

The following result is the key to prove the complete finiteness of $Z$.

\begin{proposition}\label{Prop: key relation}
    When $|C|=3$, $\tilde{V}(S^2,C)$ is four-dimensional.
    It has a relation in Fig \ref{fig:SkeinRelation}.
\begin{figure}[H]
    \centering
    \begin{tikzpicture}
        
        \begin{scope}[shift={(0,-3)}]

        \begin{scope}[shift={(3,0)}]
        \node at (-.5,-1){$+$};

         \fill[gray!20](0,0)--(.75,.75)--(.75,.75-2.5)--(0,-2.5);
         \draw[gray!60](0,0)--(.75,.75)--(.75,.75-2.5)--(0,-2.5)--(0,0);
        
        \begin{scope}[shift={(1,0.5)}]

        \fill[gray!20](0,0)--(.25,.25)--(.25,.25-2)--(0,-2);
        \draw[gray!60](0,0)--(.25,.25)--(.25,.25-2);


        \fill[gray!15](0.25,0.25)--(.75,.25)--(.75,-1.25)--(0.25,-1.25);
        \draw[gray!60](0.25,0.25)--(.75,.25)--(.75,-1.25);

         \fill[gray!25](0,0)--(.5,0)--(.5,-2.5)--(0,-2.5);

         \draw[gray!60](0,0)--(.5,0)--(.5,-2.5)--(0,-2.5)--(0,0);
          \fill[gray!25](0.5,0)--(.75,.25)--(.75,.25-2.5)--(0.5,-2.5);
          \draw[gray!60](0.5,0)--(.75,.25)--(.75,.25-2.5)--(0.5,-2.5)--(0.5,0);
           \end{scope}

        \end{scope}

        \begin{scope}[shift={(0,0)}]
        
        \node at (-.7,-1){$-\sqrt{2}$};

         \fill[gray!20](0,0)--(.75,.75)--(.75,.75-2.5)--(0,-2.5);
         \draw[gray!60](0,0)--(.75,.75)--(.75,.75-2.5)--(0,-2.5)--(0,0);
        \begin{scope}[shift={(1,.5)}]

        \fill[gray!20](0,0)--(.25,.25)--(.25,.25-1)--(0,-1);
        \draw[gray!60](0,0)--(.25,.25)--(.25,.25-1);


        \fill[gray!15](0.25,0.25)--(.75,.25)--(.75,-.75)--(0.25,-.75);
        \draw[gray!60](0.25,0.25)--(.75,.25)--(.75,-.75);

        \fill[gray!25](0.25,0.25)--(.5+.25,0.25)--(.75,-.75)--(-.25,-1+.25)--(-.25,-.25)--(0.25,-.25)--(0.25,0.25);
        \draw[gray!60](0.25,0.25)--(.5+.25,0.25)--(.75,-.75)--(-.25,-1+.25)--(-.25,-.25)--(0.25,-.25)--(0.25,0.25);

        \fill[gray!20](-.5,-.5)--(0,-.5)--(0.25,-.25)--(-.25,-.25);
         \draw[gray!60](-.5,-.5)--(0,-.5)--(0.25,-.25)--(-.25,-.25)--(-.5,-.5);

         \fill[gray!25](0,0)--(.5,0)--(.5,-1)--(-.5,-1)--(-.5,-.5)--(0,-.5)--(0,0);
         \draw[gray!60](0,0)--(.5,0)--(.5,-1)--(-.5,-1)--(-.5,-.5)--(0,-.5)--(0,0);

        \fill[gray!25](0.5,0)--(.75,.25)--(.75,.25-1)--(0.5,-1);
          \draw[gray!60](0.5,0)--(.75,.25)--(.75,.25-1)--(0.5,-1)--(0.5,0);

        \end{scope}

        \begin{scope}[shift={(1,-.5)}]

        \fill[gray!20](-.5,-.5)--(0,-.5)--(0.25,-.25)--(-.25,-.25);
         \draw[gray!60](-.5,-.5)--(0,-.5)--(0.25,-.25)--(-.25,-.25)--(-.5,-.5);

        \fill[gray!25](0.25,-.25)--(.75,-.25)--(.75,-1.25)--(0.25,-1.25)--(0.25,-.75)--(-.25,-.75)--(-.25,-.25)--(0.25,-.25);
        \draw[gray!60](0.25,-.25)--(.75,-.25)--(.75,-1.25)--(0.25,-1.25)--(0.25,-.75)--(-.25,-.75)--(-.25,-.25)--(0.25,-.25);

         \fill[gray!25](0,-.5)--(.5,-.5)--(.5,-1.5)--(0,-1.5)--(0,-1)--(-.5,-1)--(-.5,-.5)--(0,-.5);
         \draw[gray!60](0,-.5)--(.5,-.5)--(.5,-1.5)--(0,-1.5)--(0,-1)--(-.5,-1)--(-.5,-.5)--(0,-.5);

         \fill[gray!25](-.5,-.5)--(.5,-.5)--(.75,-.25)--(-.25,-.25);
         \draw[gray!60](-.5,-.5)--(.5,-.5)--(.75,-.25)--(-.25,-.25)--(-.5,-.5);

         \fill[gray!25](.5,-.5)--(.75,-.25)--(.75,-1.25)--(.5,-1.5);
         \draw[gray!60](.5,-.5)--(.75,-.25)--(.75,-1.25)--(.5,-1.5)--(.5,-.5);

           \end{scope}

        \end{scope}

        \begin{scope}[shift={(6,0)}]
        
        \node at (-.5,-1){$+$};

         \fill[gray!20](0,0)--(.75,.75)--(.75,.75-2.5)--(0,-2.5);
         \draw[gray!60](0,0)--(.75,.75)--(.75,.75-2.5)--(0,-2.5)--(0,0);
        
        \begin{scope}[shift={(1,.5)}]

        \fill[gray!20](0,0)--(.25,.25)--(.25,.25-1)--(0,-1);
        \draw[gray!60](0,0)--(.25,.25)--(.25,.25-1);


        \fill[gray!15](0.25,0.25)--(.75,.25)--(.75,-.75)--(0.25,-.75);
        \draw[gray!60](0.25,0.25)--(.75,.25)--(.75,-.75);

        \fill[gray!25](0.25,0.25)--(.5+.25,0.25)--(.75,-.75)--(-.25,-1+.25)--(-.25,-.25)--(0.25,-.25)--(0.25,0.25);
        \draw[gray!60](0.25,0.25)--(.5+.25,0.25)--(.75,-.75)--(-.25,-1+.25)--(-.25,-.25)--(0.25,-.25)--(0.25,0.25);

        \fill[gray!20](-.5,-.5)--(0,-.5)--(0.25,-.25)--(-.25,-.25);
         \draw[gray!60](-.5,-.5)--(0,-.5)--(0.25,-.25)--(-.25,-.25)--(-.5,-.5);

         \fill[gray!25](0,0)--(.5,0)--(.5,-1)--(-.5,-1)--(-.5,-.5)--(0,-.5)--(0,0);
         \draw[gray!60](0,0)--(.5,0)--(.5,-1)--(-.5,-1)--(-.5,-.5)--(0,-.5)--(0,0);

          \fill[gray!25](0.5,0)--(.75,.25)--(.75,.25-1)--(0.5,-1);
          \draw[gray!60](0.5,0)--(.75,.25)--(.75,.25-1)--(0.5,-1)--(0.5,0);

        \end{scope}

        \begin{scope}[shift={(1,-.5)}]


        \fill[gray!25](0.25,-.25)--(.75,-.25)--(.75,-1.25)--(0.25,-1.25)--(0.25,-.75)--(0.25,-.25);
        \draw[gray!60](0.25,-.25)--(.75,-.25)--(.75,-1.25)--(0.25,-1.25)--(0.25,-.75)--(0.25,-.25);

         \fill[gray!25](0,-.5)--(.5,-.5)--(.5,-1.5)--(0,-1.5)--(0,-1)--(0,-.5);
         \draw[gray!60](0,-.5)--(.5,-.5)--(.5,-1.5)--(0,-1.5)--(0,-1)--(0,-.5);

         \fill[gray!25](0,-.5)--(.5,-.5)--(.75,-.25)--(.25,-.25);
         \draw[gray!60](0,-.5)--(.5,-.5)--(.75,-.25)--(.25,-.25)--(0,-.5);

         \fill[gray!25](.5,-.5)--(.75,-.25)--(.75,-1.25)--(.5,-1.5);
         \draw[gray!60](.5,-.5)--(.75,-.25)--(.75,-1.25)--(.5,-1.5)--(.5,-.5);

        \end{scope}

        \end{scope}

         \begin{scope}[shift={(9,0)}]
        
        \node at (-.5,-1){$+$};

         \fill[gray!20](0,0)--(.75,.75)--(.75,.75-2.5)--(0,-2.5);
         \draw[gray!60](0,0)--(.75,.75)--(.75,.75-2.5)--(0,-2.5)--(0,0);
        
        \begin{scope}[shift={(1,.5)}]

        \fill[gray!20](0,0)--(.25,.25)--(.25,.25-1)--(0,-1);
        \draw[gray!60](0,0)--(.25,.25)--(.25,.25-1);

        \fill[gray!15](0.25,0.25)--(.75,.25)--(.75,-.75)--(0.25,-.75);
        \draw[gray!60](0.25,0.25)--(.75,.25)--(.75,-.75);

        \fill[gray!25](0.25,0.25)--(.5+.25,0.25)--(.75,-.75)--(.25,-1+.25)--(0.25,0.25);
        \draw[gray!60](0.25,0.25)--(.5+.25,0.25)--(.75,-.75)--(.25,-1+.25)--(0.25,0.25);

         \fill[gray!25](0,0)--(.5,0)--(.5,-1)--(0,-1)--(0,0);
         \draw[gray!60](0,0)--(.5,0)--(.5,-1)--(0,-1)--(0,0);

          \fill[gray!25](0.5,0)--(.75,.25)--(.75,.25-1)--(0.5,-1);
          \draw[gray!60](0.5,0)--(.75,.25)--(.75,.25-1)--(0.5,-1)--(0.5,0);

           \end{scope}

        \begin{scope}[shift={(1,-.5)}]
         
        \fill[gray!20](-.5,-.5)--(0,-.5)--(0.25,-.25)--(-.25,-.25);
         \draw[gray!60](-.5,-.5)--(0,-.5)--(0.25,-.25)--(-.25,-.25)--(-.5,-.5);

        \fill[gray!25](0.25,-.25)--(.75,-.25)--(.75,-1.25)--(0.25,-1.25)--(0.25,-.75)--(-.25,-.75)--(-.25,-.25)--(0.25,-.25);
        \draw[gray!60](0.25,-.25)--(.75,-.25)--(.75,-1.25)--(0.25,-1.25)--(0.25,-.75)--(-.25,-.75)--(-.25,-.25)--(0.25,-.25);

         \fill[gray!25](0,-.5)--(.5,-.5)--(.5,-1.5)--(0,-1.5)--(0,-1)--(-.5,-1)--(-.5,-.5)--(0,-.5);
         \draw[gray!60](0,-.5)--(.5,-.5)--(.5,-1.5)--(0,-1.5)--(0,-1)--(-.5,-1)--(-.5,-.5)--(0,-.5);

         \fill[gray!25](-.5,-.5)--(.5,-.5)--(.75,-.25)--(-.25,-.25);
         \draw[gray!60](-.5,-.5)--(.5,-.5)--(.75,-.25)--(-.25,-.25)--(-.5,-.5);

         \fill[gray!25](.5,-.5)--(.75,-.25)--(.75,-1.25)--(.5,-1.5);
         \draw[gray!60](.5,-.5)--(.75,-.25)--(.75,-1.25)--(.5,-1.5)--(.5,-.5);

           \end{scope}

        \end{scope}

        \begin{scope}[shift={(12,0)}]
        
        \node at (-.5,-1){$-$};

         \fill[gray!20](0,0)--(.75,.75)--(.75,.75-2.5)--(0,-2.5);
         \draw[gray!60](0,0)--(.75,.75)--(.75,.75-2.5)--(0,-2.5)--(0,0);
        
        \begin{scope}[shift={(1,.5)}]

        \fill[gray!20](0,0)--(.25,.25)--(.25,.25-1)--(0,-1);
        \draw[gray!60](0,0)--(.25,.25)--(.25,.25-1);

        \fill[gray!15](0.25,0.25)--(.75,.25)--(.75,-.75)--(0.25,-.75);
        \draw[gray!60](0.25,0.25)--(.75,.25)--(.75,-.75);

        \fill[gray!25](0.25,0.25)--(.5+.25,0.25)--(.75,-.75)--(.25,-1+.25)--(0.25,0.25);
        \draw[gray!60](0.25,0.25)--(.5+.25,0.25)--(.75,-.75)--(.25,-1+.25)--(0.25,0.25);

         \fill[gray!25](0,0)--(.5,0)--(.5,-1)--(0,-1)--(0,0);
         \draw[gray!60](0,0)--(.5,0)--(.5,-1)--(0,-1)--(0,0);

          \fill[gray!25](0.5,0)--(.75,.25)--(.75,.25-1)--(0.5,-1);
          \draw[gray!60](0.5,0)--(.75,.25)--(.75,.25-1)--(0.5,-1)--(0.5,0);

           \end{scope}

        \begin{scope}[shift={(1,-.5)}]

        \fill[gray!25](0.25,-.25)--(.75,-.25)--(.75,-1.25)--(0.25,-1.25)--(0.25,-.75)--(0.25,-.25);
        \draw[gray!60](0.25,-.25)--(.75,-.25)--(.75,-1.25)--(0.25,-1.25)--(0.25,-.75)--(0.25,-.25);

         \fill[gray!25](0,-.5)--(.5,-.5)--(.5,-1.5)--(0,-1.5)--(0,-1)--(0,-.5);
         \draw[gray!60](0,-.5)--(.5,-.5)--(.5,-1.5)--(0,-1.5)--(0,-1)--(0,-.5);

         \fill[gray!25](0,-.5)--(.5,-.5)--(.75,-.25)--(.25,-.25);
         \draw[gray!60](0,-.5)--(.5,-.5)--(.75,-.25)--(.25,-.25)--(0,-.5);

         \fill[gray!25](.5,-.5)--(.75,-.25)--(.75,-1.25)--(.5,-1.5);
         \draw[gray!60](.5,-.5)--(.75,-.25)--(.75,-1.25)--(.5,-1.5)--(.5,-.5);

        \end{scope}

        \end{scope}

        \node at (15,-1){$=0$};

        \end{scope}
    \end{tikzpicture}
    \caption{The key relation with three $S^1$ on the boundary.}
    \label{fig:SkeinRelation}
\end{figure}
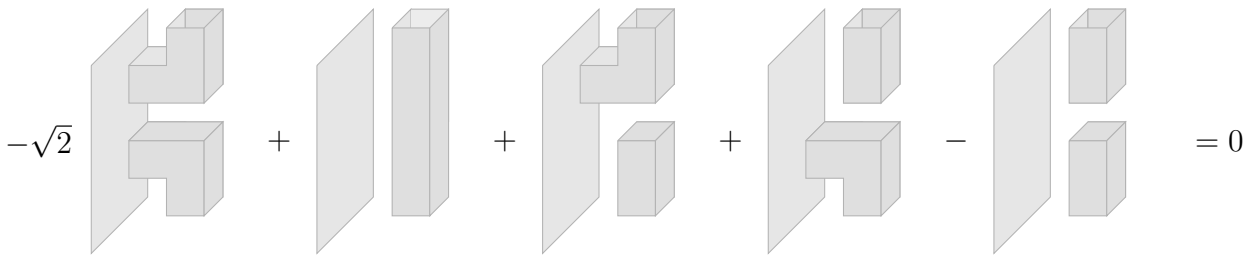
\end{proposition}

\begin{proof}
We have five vectors according to connected types:
$$(123),(12,3),(13,2),(23,1),(1,2,3)$$ corresponding to the five terms above. By direct computation, their inner product matrix has rank 4, and we obtain a relation from $K_Z$, as shown in \ref{fig:SkeinRelation}: $$-\sqrt{2}(123)+(12,3)+(13,2)+(23,1)-(1,2,3)=0$$ 
\end{proof}

Let us construct the indecomposible $2$-morphisms of type $\mathcal{D}^2$ with link boundary $S^1$ without stratification.
Up to isotopy, we only consider the case that $\mathcal{D}^2$ contains non-intersecting contractible closed curves which are the boundary of small 2-discs.
When there is no small 2-disc, we obtain an indecomposible $2$-morphisms $t_0$ as $D^2\times D^1$ without stratification.
When there is one small 2-disc $D$, we obtain the identity illustrated as the second figure in Fig.\ref{fig: red tube}, which has a stratification $\partial D\times D^1$; and an indecomposible $2$-morphism $t_0'$ equivalent to $t_0$ illustrated as the third figure in Fig.\ref{fig: red tube}. Their difference $t_r$ is an indecomposible $2$-morphism inequivalent to $t_0$, illustrated as the first figure, a red tube, in Fig.\ref{fig: red tube}.

\begin{figure}[H]
    \centering
    \begin{tikzpicture}
        \begin{scope}[shift={(0,.5)}]

        \node at (1,-.25)[red]{$r$};

        \fill[red!20](0,0)--(.25,.25)--(.25,.25-2.5)--(0,-2.5);
        \draw[red!60](0,0)--(.25,.25)--(.25,.25-2.5);


        \fill[red!15](0.25,0.25)--(.75,.25)--(.75,-1.25)--(0.25,-1.25);
        \draw[red!60](0.25,0.25)--(.75,.25)--(.75,-1.25);

         \fill[red!25](0,0)--(.5,0)--(.5,-2.5)--(0,-2.5);

         \draw[red!60](0,0)--(.5,0)--(.5,-2.5)--(0,-2.5)--(0,0);
          \fill[red!25](0.5,0)--(.75,.25)--(.75,.25-2.5)--(0.5,-2.5);
          \draw[red!60](0.5,0)--(.75,.25)--(.75,.25-2.5)--(0.5,-2.5)--(0.5,0);
           \end{scope}

        \begin{scope}[shift={(2,.5)}]

         \node at (-.5,-1){$=$};

        \fill[gray!20](0,0)--(.25,.25)--(.25,.25-2.5)--(0,-2.5);
        \draw[gray!60](0,0)--(.25,.25)--(.25,.25-2.5);


        \fill[gray!15](0.25,0.25)--(.75,.25)--(.75,-1.25)--(0.25,-1.25);
        \draw[gray!60](0.25,0.25)--(.75,.25)--(.75,-1.25);

         \fill[gray!25](0,0)--(.5,0)--(.5,-2.5)--(0,-2.5);

         \draw[gray!60](0,0)--(.5,0)--(.5,-2.5)--(0,-2.5)--(0,0);
          \fill[gray!25](0.5,0)--(.75,.25)--(.75,.25-2.5)--(0.5,-2.5);
          \draw[gray!60](0.5,0)--(.75,.25)--(.75,.25-2.5)--(0.5,-2.5)--(0.5,0);
           \end{scope}

        \begin{scope}[shift={(3,0)}]
        
        \node at (0.5,-.5){$-\frac{1}{\sqrt{2}}$};
        
        \begin{scope}[shift={(1,.5)}]

        \fill[gray!20](0,0)--(.25,.25)--(.25,.25-.5)--(0,-.5);
        \draw[gray!60](0,0)--(.25,.25)--(.25,.25-.5);


        \fill[gray!15](0.25,0.25)--(.75,.25)--(.75,-.25)--(0.25,-.25);
        \draw[gray!60](0.25,0.25)--(.75,.25)--(.75,-.25);

        \fill[gray!25](0.25,0.25)--(.5+.25,0.25)--(.75,-.25)--(.25,-.25)--(0.25,0.25);
        \draw[gray!60](0.25,0.25)--(.5+.25,0.25)--(.75,-.25)--(.25,-.25)--(0.25,0.25);

         \fill[gray!25](0,0)--(.5,0)--(.5,-.5)--(0,-.5)--(0,0);
         \draw[gray!60](0,0)--(.5,0)--(.5,-.5)--(0,-.5)--(0,0);

          \fill[gray!25](0.5,0)--(.75,.25)--(.75,.25-.5)--(0.5,-.5);
          \draw[gray!60](0.5,0)--(.75,.25)--(.75,.25-.5)--(0.5,-.5)--(0.5,0);

           \end{scope}

        \begin{scope}[shift={(1,-.5)}]


         \fill[gray!25](0,-1)--(.5,-1)--(.5,-1.5)--(0,-1.5)--(0,-1);
         \draw[gray!60](0,-1)--(.5,-1)--(.5,-1.5)--(0,-1.5)--(0,-1);

         \fill[gray!25](0,-1)--(.5,-1)--(.75,-.75)--(.25,-.75);
         \draw[gray!60](0,-1)--(.5,-1)--(.75,-.75)--(.25,-.75)--(0,-1);

         \fill[gray!25](.5,-1)--(.75,-.75)--(.75,-1.25)--(.5,-1.5);
         \draw[gray!60](.5,-1)--(.75,-.75)--(.75,-1.25)--(.5,-1.5)--(.5,-1);

           \end{scope}

        \end{scope}

    \end{tikzpicture}
    \caption{Definition of a red tube}
    \label{fig: red tube}
\end{figure}
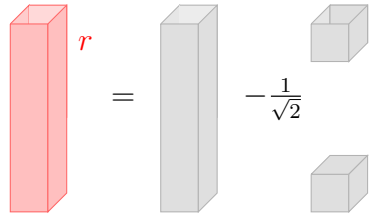

\begin{proposition}\label{Prop: merge red tube}
We have the following relation to fuse the red tube to with a plane:
\begin{figure}[H]
    \centering
    \begin{tikzpicture}

        \begin{scope}
            
         \fill[gray!20](0,0)--(.75,.75)--(.75,.75-2.5)--(0,-2.5);
         \draw[gray!60](0,0)--(.75,.75)--(.75,.75-2.5)--(0,-2.5)--(0,0);
        
        \begin{scope}[shift={(1,0)}]
         
        \node at (1,-.25)[red]{$r$};

        \fill[red!20](0,0)--(.25,.25)--(.25,.25-2)--(0,-2);
        \draw[red!60](0,0)--(.25,.25)--(.25,.25-2);


        \fill[red!15](0.25,0.25)--(.75,.25)--(.75,-1.25)--(0.25,-1.25);
        \draw[red!60](0.25,0.25)--(.75,.25)--(.75,-1.25);

         \fill[red!25](0,0)--(.5,0)--(.5,-2)--(0,-2);

         \draw[red!60](0,0)--(.5,0)--(.5,-2)--(0,-2)--(0,0);
          \fill[red!25](0.5,0)--(.75,.25)--(.75,.25-2)--(0.5,-2);
          \draw[red!60](0.5,0)--(.75,.25)--(.75,.25-2)--(0.5,-2)--(0.5,0);
           \end{scope}

        \end{scope}

        \begin{scope}[shift={(4,0)}]
        
        \node at (-1,-1){$=\sqrt{2}$};

         \fill[gray!20](0,0)--(.75,.75)--(.75,.75-2.5)--(0,-2.5);
         \draw[gray!60](0,0)--(.75,.75)--(.75,.75-2.5)--(0,-2.5)--(0,0);
        \begin{scope}[shift={(1,.5)}]
        
        \node at (1,-.25)[red]{$r$};

        \fill[red!20](0,0)--(.25,.25)--(.25,.25-1)--(0,-1);
        \draw[red!60](0,0)--(.25,.25)--(.25,.25-1);


        \fill[red!15](0.25,0.25)--(.75,.25)--(.75,-.75)--(0.25,-.75);
        \draw[red!60](0.25,0.25)--(.75,.25)--(.75,-.75);

        \fill[red!25](0.25,0.25)--(.5+.25,0.25)--(.75,-.75)--(-.25,-1+.25)--(-.25,-.25)--(0.25,-.25)--(0.25,0.25);
        \draw[red!60](0.25,0.25)--(.5+.25,0.25)--(.75,-.75)--(-.25,-1+.25)--(-.25,-.25)--(0.25,-.25)--(0.25,0.25);

        \fill[red!20](-.5,-.5)--(0,-.5)--(0.25,-.25)--(-.25,-.25);
         \draw[red!60](-.5,-.5)--(0,-.5)--(0.25,-.25)--(-.25,-.25)--(-.5,-.5);

         \fill[red!25](0,0)--(.5,0)--(.5,-1)--(-.5,-1)--(-.5,-.5)--(0,-.5)--(0,0);
         \draw[red!60](0,0)--(.5,0)--(.5,-1)--(-.5,-1)--(-.5,-.5)--(0,-.5)--(0,0);

          \fill[red!25](0.5,0)--(.75,.25)--(.75,.25-1)--(0.5,-1);
          \draw[red!60](0.5,0)--(.75,.25)--(.75,.25-1)--(0.5,-1)--(0.5,0);

           \end{scope}

        \begin{scope}[shift={(1,-.5)}]

        \node at (1,-.5)[red]{$r$};

        \fill[red!20](-.5,-.5)--(0,-.5)--(0.25,-.25)--(-.25,-.25);
         \draw[red!60](-.5,-.5)--(0,-.5)--(0.25,-.25)--(-.25,-.25)--(-.5,-.5);

        \fill[red!25](0.25,-.25)--(.75,-.25)--(.75,-1.25)--(0.25,-1.25)--(0.25,-.75)--(-.25,-.75)--(-.25,-.25)--(0.25,-.25);
        \draw[red!60](0.25,-.25)--(.75,-.25)--(.75,-1.25)--(0.25,-1.25)--(0.25,-.75)--(-.25,-.75)--(-.25,-.25)--(0.25,-.25);

         \fill[red!25](0,-.5)--(.5,-.5)--(.5,-1.5)--(0,-1.5)--(0,-1)--(-.5,-1)--(-.5,-.5)--(0,-.5);
         \draw[red!60](0,-.5)--(.5,-.5)--(.5,-1.5)--(0,-1.5)--(0,-1)--(-.5,-1)--(-.5,-.5)--(0,-.5);

         \fill[red!25](-.5,-.5)--(.5,-.5)--(.75,-.25)--(-.25,-.25);
         \draw[red!60](-.5,-.5)--(.5,-.5)--(.75,-.25)--(-.25,-.25)--(-.5,-.5);

         \fill[red!25](.5,-.5)--(.75,-.25)--(.75,-1.25)--(.5,-1.5);
         \draw[red!60](.5,-.5)--(.75,-.25)--(.75,-1.25)--(.5,-1.5)--(.5,-.5);

           \end{scope}

        \end{scope}

    \end{tikzpicture}
    \caption{The skein relation to fuse a red tube and a plane.}
\end{figure}
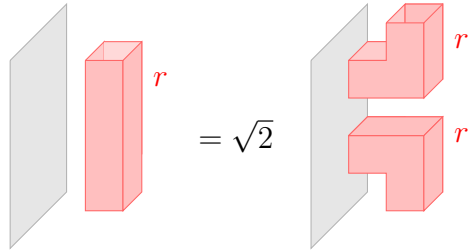
\end{proposition}

\begin{proof}
It follows from the definition of the red tube in Fig.~\ref{fig: red tube} and Prop.~\ref{Prop: key relation}.
\end{proof}

\begin{theorem}\label{Thm: RP}
The $S^3$ functional $Z$ in Def.~\ref{Def: Ising Z} is reflection positive.
The dimension of the vector space $\tilde{V}(S^2,C)$ is $2^{|C|-1}$.
\end{theorem}

\begin{proof}
An isotopy on the boundary induces an isometry, so we assume that the curves in $C$ do not intersect. 
If $C$ has zero or one curve, then $\tilde{V}(S^2,C)$ is one dimensional by Prop.~\ref{Prop: 0-curve},\ref{Prop: 1-curve}. Reflection positivity holds on these vectors spaces. 

Now we prove the statement by induction on the number of circles in $C$.
Take a disc $D$ of $S^2\setminus C$, s.t. $\partial D$ is a circle in $C$, we decompose $\bar{D} \times [0,1]$ as  a sum of orthogonal projections $t_0'$ and $t_r$. 
Therefore $\tilde{V}(S^2,C)$ decomposes a direct sum of two vector spaces. 
The vector space labelled by $t_0'$ is isomorphic to the vector space with $\tilde{V}(S^2, C\setminus \{\partial D\}$, because $t_0'$ is annular equivalent to $t_0$.
The vector space labelled by $t_r$ is isomorphic to the vector space with $\tilde{V}(S^2, C\setminus \{\partial D\}$ by Prop.\ref{Prop: merge red tube}.
Therefore $Z$ is reflection positive and the dimension of the vector space $\tilde{V}(S^2,C)$ is $2^{|C|-1}$.
\end{proof}

The type of a $k$-morphism is the transversal intersection of a $k$-simplex with  
stratified $3$-manifold $\mathcal{M}$.
When $k=1$, it is $D^1$ stratified by finite points.
When $k=2$, it is $D^2$ stratified by curves, which may intersect.
To prove complete finite, we need to reduce to finite types by the skein relations of their types induced from the annular equivalence.

\begin{theorem}\label{Thm: finite 2-morphism}
The $S^3$ functional $Z$ in Def.~\ref{Def: Ising Z} is 2-finite.   
\end{theorem}

\begin{proof}
The type of 2-morphisms is $D^2$ stratified by curves. By isotopy, we may assume that the curves do not intersect. 
A contractible curve decomposes as a sum of $t_0'$ and $t_g$. Moreover, the curve labelled by $t_0'$ is equivalent to empty set $t_0$. If the curve labelled by $t_g$ is adjacent to another curve, then we can merge $t_g$ into the other curve by Prop.\ref{Prop: merge red tube}.
So all closed curves can be eliminated unless the type has only one closed curve labelled by $t_g$.
For a fixed link boundary, the non-closed curves pair the boundary points, so there are only finitely many types of 2-morphisms. So $Z$ is 2-finite.  
\end{proof}

We will explicitly construct the simplical $k$-morphisms in the following subsection and compute the 20j symbols. In particular, we prove that $Z$ is 1-finite in Theorem \ref{Thm: finite 1-morphism}.

\begin{theorem}\label{Thm: RP and CF}
The $S^3$ functional $Z$ in Def.~\ref{Def: Ising Z} is reflection positive and complete finite. So we obtain an $3+1$ alterfold TQFT with reflection positivity.
\end{theorem}

\begin{proof}
It follows from Theorems~\ref{Thm: RP}, \ref{Thm: finite 2-morphism} and \ref{Thm: finite 1-morphism}.    
\end{proof}

\begin{remark}
If we only consider non-intersecting surfaces, then the Theorem \ref{Thm: finite 1-morphism} fails. Instead, we will obtain infinitely many indecomposible 1-morphisms indexed by natural numbers $\mathbb{N}$. The reflection positivity, 3-finite, 2-finite conditions still hold.
This example has its own interests. In the paper, we will focus on the examples which are complete finite, so that we can construct the $n+1$ alterfold TQFT.      
\end{remark}

\subsection{Simplicial 1-morphisms}

In this subsection, we construct the indecomposible 1-morphisms up to annular equivalence.
\begin{theorem}\label{Thm: finite 1-morphism}
There are three indecomposible 1-morphisms up to annular equivalence, denoted as $E_1=\{\mathbbm{1},\tau, g\}$. 
\end{theorem}

Now let us study 1-morphisms of type $\mathcal{D}^1_m$, which has 0-dim stratification of $m$ points.
When $m=0$, the algebra $A(\mathcal{D}^1\times D^2)$ is one-dimensional by Prop.~\ref{Prop: 0-curve}. 
The identity $D^1\times D^2$ is an indecomposible 1-morphism, denoted by $\mathbbm{1}$.

When $m=1$, the algebra $A(\mathcal{D}^1_1\times D^2)$ is one-dimensional by Prop.~\ref{Prop: 1-curve}. 
The identity $D^1_1\times D^2$ is an indecomposible 1-morphism, denoted by $\tau$. Its quantum dimension is $\sqrt{2}$ as shown in Fig.~\ref{fig: tr(tau)}.

\begin{figure}[H]
    \centering
    \begin{tikzpicture}
    \begin{scope}[shift={(.2,0)}]
    \node at (.5,-.25){$Tr($};
            \fill[gray!20] (1,0)--(1.5,.5)--(1.5,-.5)--(1,-1);
            \draw[gray!60] (1,0)--(1.5,.5)--(1.5,-.5)--(1,-1)--(1,0);
        \node at (2,-.25){$):=$};
        \end{scope}
        \begin{scope}[shift={(3,0)}]
        
            \fill[gray!20] (0,0)--(1,0)--(1,-1)--(0,-1);
            \fill[gray!20] (0,0)--(1,0)--(1.5,.5)--(.5,.5);
            \fill[gray!20] (1,0)--(1.5,.5)--(1.5,-.5)--(1,-1);
            \draw[gray!60] (1,0)--(1.5,.5)--(1.5,-.5)--(1,-1)--(1,0);
            \draw[gray!60] (0,0)--(1,0)--(1.5,.5)--(.5,.5)--(0,0);
            \draw[gray!60] (0,0)--(1,0)--(1,-1)--(0,-1)--(0,0);
            \draw[gray!60,dashed] (.5,.5)--(.5,-.5);
            \draw[gray!60,dashed] (.5,-.5)--(1.5,-.5);
            \draw[gray!60,dashed] (.5,-.5)--(0,-1);

        \node at (2.5,-.25){$=\sqrt{2}$};
        \end{scope}
    \end{tikzpicture}
    \caption{The quantum dimension of the 1-morphism $\tau$.}
    \label{fig: tr(tau)}
\end{figure}

When $m=2$, the algebra $A(\mathcal{D}^1_2\times D^2)$ is two-dimensional by Prop.~\ref{Prop: 2-curve}. The identity $\mathcal{D}^1_2\times D^2$ has decomposes as a sum of two minimal idempotent, namely two indecomposible 1-morphisms.
There are two basis vector according to the connected type $(12)$,$(1,2)$. 
As shown in Fig.\ref{fig: tau0}, $(1,2)$ is a multiple of a minimal idempotent, corresponding to a 1-morphism $\mathbbm{1}'$, which is annular equivalent to $\mathbbm{1}$, because it has no intersection with the 1-simplex in the middle of the tunnel.
The orthogonal complement of $\mathbbm{1}'$ is a 1-morphism, denoted by $g$, illustrated in Fig.\ref{fig: tau g} as a red line between double planes.

\begin{figure}[H]
    \centering
    \begin{tikzpicture}
    \begin{scope}

         \begin{scope}[shift={(.2,0)}]

         \fill[gray!20](0,0)--(.75,.75)--(.75,.75-2.5)--(0,-2.5);
         \draw[gray!60](0,0)--(.75,.75)--(.75,.75-2.5)--(0,-2.5)--(0,0);
        
 \end{scope}

        \draw[red, line width=1pt](.6,-1)--(1.6,-1);

            \begin{scope}[shift={(1.2,0)}]
             \fill[opacity=.2](0,0)--(.75,.75)--(.75,.75-2.5)--(0,-2.5);
         \draw[gray!60](0,0)--(.75,.75)--(.75,.75-2.5)--(0,-2.5)--(0,0);
        \end{scope}

        \end{scope}

    \begin{scope}[shift={(4,0)}]
    
     \node at (-.4,-1){$:=$};

         \begin{scope}[shift={(.2,0)}]

         \fill[gray!20](0,0)--(.75,.75)--(.75,.75-2.5)--(0,-2.5);
         \draw[gray!60](0,0)--(.75,.75)--(.75,.75-2.5)--(0,-2.5)--(0,0);
        
 \end{scope}

            \begin{scope}[shift={(1.2,0)}]
             \fill[opacity=.2](0,0)--(.75,.75)--(.75,.75-2.5)--(0,-2.5);
         \draw[gray!60](0,0)--(.75,.75)--(.75,.75-2.5)--(0,-2.5)--(0,0);
        \end{scope}

        \end{scope}

     \begin{scope}[shift={(7,0)}]
     \node at (-.4,-1){$-\frac{1}{\sqrt{2}}$};
        
         \begin{scope}[shift={(.2,0)}]

         \fill[gray!20](0,0)--(.75,.75)--(.75,.75-2.5)--(0,-2.5);
         \draw[gray!60](0,0)--(.75,.75)--(.75,.75-2.5)--(0,-2.5)--(0,0);
        
 \end{scope}

        \begin{scope}[shift={(1,-.25)}]


        \fill[gray!20](-.5,-.5)--(0,-.5)--(0.25,-.25)--(-.25,-.25);
         \draw[gray!60](-.5,-.5)--(0,-.5)--(0.25,-.25)--(-.25,-.25)--(-.5,-.5);

        \fill[gray!25](0.25,-.25)--(.75,-.25)--(.75,-.75)--(0.25,-.75)--(-.25,-.75)--(-.25,-.25)--(0.25,-.25);
        \draw[gray!60](0.25,-.25)--(.75,-.25)--(.75,-.75)--(0.25,-.75)--(-.25,-.75)--(-.25,-.25)--(0.25,-.25);

         \fill[gray!25](0,-.5)--(.5,-.5)--(.5,-1)--(0,-1)--(-.5,-1)--(-.5,-.5)--(0,-.5);
         \draw[gray!60](0,-.5)--(.5,-.5)--(.5,-1)--(0,-1)--(0,-1)--(-.5,-1)--(-.5,-.5)--(0,-.5);

         \fill[gray!25](-.5,-.5)--(.5,-.5)--(.75,-.25)--(-.25,-.25);
         \draw[gray!60](-.5,-.5)--(.5,-.5)--(.75,-.25)--(-.25,-.25)--(-.5,-.5);

           \end{scope}

            \begin{scope}[shift={(1.2,0)}]
             \fill[opacity=.2](0,0)--(.75,.75)--(.75,.75-2.5)--(0,-2.5);
         \draw[gray!60](0,0)--(.75,.75)--(.75,.75-2.5)--(0,-2.5)--(0,0);
        \end{scope}

        \begin{scope}[shift={(1,-.25)}]
        \fill[gray!10](.5,-.5)--(.75,-.25)--(.75,-.75)--(.5,-1);
         \draw[gray!60](.5,-.5)--(.75,-.25)--(.75,-.75)--(.5,-1)--(.5,-.5);

        \end{scope}

        \end{scope}
    \end{tikzpicture}
    \caption{The red line in the double-surfaces}
    \label{fig: tau g}
\end{figure}

\begin{lemma}\label{Lem: g 1 inequivalence}
If a 2-disc is connected by a red line from $g$, then it is zero. 
Consequently the 1-morphism $g$ is inequivalent to $\mathbbm{1}$.  
\end{lemma}
\begin{proof}
Gluing the vertical half tube to $g$ in Fig.\ref{fig: tau g} is zero on the right side. The left side is a 2-disc is connected by a red line from $g$.
The only possible bimodule between $g$ and $\mathbbm{1}$ is the vertical half tube, but gluing it with $g$ is zero. 
\end{proof}

When $m=3$, the algebra $A(\mathcal{D}^1_2\times D^2)$ is four-dimensional by Prop.~\ref{Prop: key relation}.
It has five idempotents according to the connected types
$P(1,2,3)$, $P(12,3)$, $P(13,2)$, $P(23,1)$, $P(123)$ with quantum dimension $2\sqrt{2}, \sqrt{2}, \sqrt{2}, \sqrt{2}, \frac{\sqrt{2}}{2}$.
We obtain four minimal idempotents $P(12,3)-P(123)$, $P(13,2)-P(123)$, $P(23,1)-P(123)$, $P(123)$ with the same quantum dimension $\frac{\sqrt{2}}{2}$. They are all annular equivalent to $\tau$.
For example, a 1-simplex can go through the tunnel of the first two planes of $P(12,3)-P(123)$ and intersect with the third plane as illustrated in Fig.~\ref{fig: tau equivalence}, so it is annular equivalent to $\tau$.

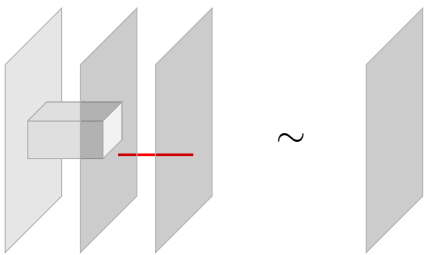
\begin{figure}[H]
    \centering
    \begin{tikzpicture}
        \begin{scope}
        
        \begin{scope}[shift={(.2,0)}] 
        
         \fill[gray!20](0,0)--(.75,.75)--(.75,.75-2.5)--(0,-2.5);
         \draw[gray!60](0,0)--(.75,.75)--(.75,.75-2.5)--(0,-2.5)--(0,0);
        
        \end{scope}

        \begin{scope}[shift={(1,-.25)}]

        \fill[gray!20](-.5,-.5)--(0,-.5)--(0.25,-.25)--(-.25,-.25);
         \draw[gray!60](-.5,-.5)--(0,-.5)--(0.25,-.25)--(-.25,-.25)--(-.5,-.5);

        \fill[gray!25](0.25,-.25)--(.75,-.25)--(.75,-.75)--(0.25,-.75)--(-.25,-.75)--(-.25,-.25)--(0.25,-.25);
        \draw[gray!60](0.25,-.25)--(.75,-.25)--(.75,-.75)--(0.25,-.75)--(-.25,-.75)--(-.25,-.25)--(0.25,-.25);

         \fill[gray!25](0,-.5)--(.5,-.5)--(.5,-1)--(0,-1)--(-.5,-1)--(-.5,-.5)--(0,-.5);
         \draw[gray!60](0,-.5)--(.5,-.5)--(.5,-1)--(0,-1)--(0,-1)--(-.5,-1)--(-.5,-.5)--(0,-.5);

         \fill[gray!25](-.5,-.5)--(.5,-.5)--(.75,-.25)--(-.25,-.25);
         \draw[gray!60](-.5,-.5)--(.5,-.5)--(.75,-.25)--(-.25,-.25)--(-.5,-.5);

           \end{scope}

        \begin{scope}[shift={(1.2,0)}]
        
        \draw[red, line width=1pt](.5,-1.2)--(1.5,-1.2);
        \fill[opacity=.2](0,0)--(.75,.75)--(.75,.75-2.5)--(0,-2.5);
         \draw[gray!60](0,0)--(.75,.75)--(.75,.75-2.5)--(0,-2.5)--(0,0);
        \end{scope}

        \begin{scope}[shift={(1,-.25)}]
        \fill[gray!10](.5,-.5)--(.75,-.25)--(.75,-.75)--(.5,-1);
         \draw[gray!60](.5,-.5)--(.75,-.25)--(.75,-.75)--(.5,-1)--(.5,-.5);

        \end{scope}

        \begin{scope}[shift={(2.2,0)}]
        \fill[opacity=.2](0,0)--(.75,.75)--(.75,.75-2.5)--(0,-2.5);
         \draw[gray!60](0,0)--(.75,.75)--(.75,.75-2.5)--(0,-2.5)--(0,0);
        \end{scope}

        \end{scope}

        \begin{scope}[shift={(5,0)}]
        \node at (-1,-1){\large{$\sim$}};
        \fill[opacity=.2](0,0)--(.75,.75)--(.75,.75-2.5)--(0,-2.5);
         \draw[gray!60](0,0)--(.75,.75)--(.75,.75-2.5)--(0,-2.5)--(0,0);
        \end{scope}

    \end{tikzpicture}
    \caption{Annular equivalence between $P(12,3)-P(123)$ and $\tau$.}
    \label{fig: tau equivalence}
\end{figure}

It's worth mentioning that $P(1,2,3)- P(12,3)-P(13,2)-P(23,1)+2P(1,2,3)$, illustrated in Fig.\ref{fig:trace0}, is a projection with quantum dimension zero. It becomes zero in the quotient by the kernel $K_Z$. It is the key to achieve the 1-finite condition of $Z$.

\begin{proof}[Proof of Theorem \ref{Thm: finite 1-morphism}]
As discussed above if $m\geq 3$, the indecomposible $1$-morphisms of type $\mathcal{D}^1_m$ are annular equivalent indecomposible $1$-morphisms of type  $\mathcal{D}^1_{m-2}$. So there are only three indecomposible $1$-morphisms up to equivalence.
\end{proof}

\begin{proposition}\label{Prop: Ising fusion rule}
We obtain the following fusion rule of Ising type for composition of indecomposible 1-morphisms:
\begin{align*}
\tau\otimes \tau &\sim \mathbbm{1} \oplus g;\\
\tau\otimes g &\sim g\otimes \tau \sim  \tau;\\
g \otimes g &\sim \mathbbm{1}.
\end{align*}    
\end{proposition}

\begin{proof}
The first two have been shown before. The vertical half $g$-color tube is a bimodule between $g \otimes g$ and $\mathbbm{1}$, which induces a sub idempotent of $g \otimes g$ with quantum dimension one. So it equals to $g \otimes g$. So $g \otimes g \sim \mathbbm{1}.$
\end{proof}

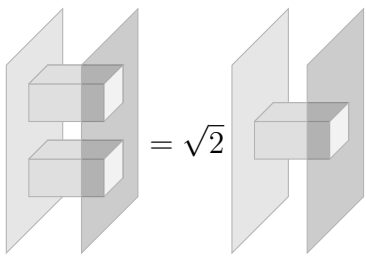
\begin{figure}[H]
    \centering
    \begin{tikzpicture}

    \begin{scope}
         \begin{scope}[shift={(.2,0)}]

         \fill[gray!20](0,0)--(.75,.75)--(.75,.75-2.5)--(0,-2.5);
         \draw[gray!60](0,0)--(.75,.75)--(.75,.75-2.5)--(0,-2.5)--(0,0);
        
        \end{scope}

        \begin{scope}[shift={(1,.25)}]


        \fill[gray!20](-.5,-.5)--(0,-.5)--(0.25,-.25)--(-.25,-.25);
         \draw[gray!60](-.5,-.5)--(0,-.5)--(0.25,-.25)--(-.25,-.25)--(-.5,-.5);

        \fill[gray!25](0.25,-.25)--(.75,-.25)--(.75,-.75)--(0.25,-.75)--(-.25,-.75)--(-.25,-.25)--(0.25,-.25);
        \draw[gray!60](0.25,-.25)--(.75,-.25)--(.75,-.75)--(0.25,-.75)--(-.25,-.75)--(-.25,-.25)--(0.25,-.25);

         \fill[gray!25](0,-.5)--(.5,-.5)--(.5,-1)--(0,-1)--(-.5,-1)--(-.5,-.5)--(0,-.5);
         \draw[gray!60](0,-.5)--(.5,-.5)--(.5,-1)--(0,-1)--(0,-1)--(-.5,-1)--(-.5,-.5)--(0,-.5);

         \fill[gray!25](-.5,-.5)--(.5,-.5)--(.75,-.25)--(-.25,-.25);
         \draw[gray!60](-.5,-.5)--(.5,-.5)--(.75,-.25)--(-.25,-.25)--(-.5,-.5);

           \end{scope}

        \begin{scope}[shift={(1,-.75)}]

        \fill[gray!20](-.5,-.5)--(0,-.5)--(0.25,-.25)--(-.25,-.25);
         \draw[gray!60](-.5,-.5)--(0,-.5)--(0.25,-.25)--(-.25,-.25)--(-.5,-.5);

        \fill[gray!25](0.25,-.25)--(.75,-.25)--(.75,-.75)--(0.25,-.75)--(-.25,-.75)--(-.25,-.25)--(0.25,-.25);
        \draw[gray!60](0.25,-.25)--(.75,-.25)--(.75,-.75)--(0.25,-.75)--(-.25,-.75)--(-.25,-.25)--(0.25,-.25);

         \fill[gray!25](0,-.5)--(.5,-.5)--(.5,-1)--(0,-1)--(-.5,-1)--(-.5,-.5)--(0,-.5);
         \draw[gray!60](0,-.5)--(.5,-.5)--(.5,-1)--(0,-1)--(0,-1)--(-.5,-1)--(-.5,-.5)--(0,-.5);

         \fill[gray!25](-.5,-.5)--(.5,-.5)--(.75,-.25)--(-.25,-.25);
         \draw[gray!60](-.5,-.5)--(.5,-.5)--(.75,-.25)--(-.25,-.25)--(-.5,-.5);

           \end{scope}

        \begin{scope}[shift={(1.2,0)}]
        \fill[opacity=.2](0,0)--(.75,.75)--(.75,.75-2.5)--(0,-2.5);
         \draw[gray!60](0,0)--(.75,.75)--(.75,.75-2.5)--(0,-2.5)--(0,0);
        \end{scope}

        \begin{scope}[shift={(1,.25)}]
                     \fill[gray!10](.5,-.5)--(.75,-.25)--(.75,-.75)--(.5,-1);
         \draw[gray!60](.5,-.5)--(.75,-.25)--(.75,-.75)--(.5,-1)--(.5,-.5);

        \end{scope}

        \begin{scope}[shift={(1,-.75)}]
                     \fill[gray!10](.5,-.5)--(.75,-.25)--(.75,-.75)--(.5,-1);
         \draw[gray!60](.5,-.5)--(.75,-.25)--(.75,-.75)--(.5,-1)--(.5,-.5);

        \end{scope}

        \end{scope}

    \begin{scope}[shift={(3,0)}]
     \node at (-.4,-1){$=\sqrt{2}$};
        
        \begin{scope}[shift={(.2,0)}] 
        
         \fill[gray!20](0,0)--(.75,.75)--(.75,.75-2.5)--(0,-2.5);
         \draw[gray!60](0,0)--(.75,.75)--(.75,.75-2.5)--(0,-2.5)--(0,0);
        
 \end{scope}

        \begin{scope}[shift={(1,-.25)}]

        \fill[gray!20](-.5,-.5)--(0,-.5)--(0.25,-.25)--(-.25,-.25);
         \draw[gray!60](-.5,-.5)--(0,-.5)--(0.25,-.25)--(-.25,-.25)--(-.5,-.5);

        \fill[gray!25](0.25,-.25)--(.75,-.25)--(.75,-.75)--(0.25,-.75)--(-.25,-.75)--(-.25,-.25)--(0.25,-.25);
        \draw[gray!60](0.25,-.25)--(.75,-.25)--(.75,-.75)--(0.25,-.75)--(-.25,-.75)--(-.25,-.25)--(0.25,-.25);

         \fill[gray!25](0,-.5)--(.5,-.5)--(.5,-1)--(0,-1)--(-.5,-1)--(-.5,-.5)--(0,-.5);
         \draw[gray!60](0,-.5)--(.5,-.5)--(.5,-1)--(0,-1)--(0,-1)--(-.5,-1)--(-.5,-.5)--(0,-.5);

         \fill[gray!25](-.5,-.5)--(.5,-.5)--(.75,-.25)--(-.25,-.25);
         \draw[gray!60](-.5,-.5)--(.5,-.5)--(.75,-.25)--(-.25,-.25)--(-.5,-.5);

           \end{scope}

        \begin{scope}[shift={(1.2,0)}]
        \fill[opacity=.2](0,0)--(.75,.75)--(.75,.75-2.5)--(0,-2.5);
         \draw[gray!60](0,0)--(.75,.75)--(.75,.75-2.5)--(0,-2.5)--(0,0);
        \end{scope}

        \begin{scope}[shift={(1,-.25)}]
        \fill[gray!10](.5,-.5)--(.75,-.25)--(.75,-.75)--(.5,-1);
         \draw[gray!60](.5,-.5)--(.75,-.25)--(.75,-.75)--(.5,-1)--(.5,-.5);

        \end{scope}

        \end{scope}
    \end{tikzpicture}
    \caption{ 2-morphism }
    \label{fig: tau0}
\end{figure}

\begin{figure}[H]
    \centering
     \begin{tikzpicture}
        \begin{scope}[shift={(0,0)}]

         \begin{scope}[shift={(.2,0)}]

         \fill[gray!20](0,0)--(.75,.75)--(.75,.75-2.5)--(0,-2.5);
         \draw[gray!60](0,0)--(.75,.75)--(.75,.75-2.5)--(0,-2.5)--(0,0);
        
 \end{scope}

        \draw[red, line width=1pt](.6,-1)--(1.6,-1);

            \begin{scope}[shift={(2.2,0)}]
             \fill[opacity=.2](0,0)--(.75,.75)--(.75,.75-2.5)--(0,-2.5);
         \draw[gray!60](0,0)--(.75,.75)--(.75,.75-2.5)--(0,-2.5)--(0,0);
        \end{scope}

        \begin{scope}[shift={(1.2,0)}]

         \draw[red, line width=1pt](.5,-1)--(1.5,-1);
             \fill[opacity=.2](0,0)--(.75,.75)--(.75,.75-2.5)--(0,-2.5);
         \draw[gray!60](0,0)--(.75,.75)--(.75,.75-2.5)--(0,-2.5)--(0,0);
        \end{scope}

         \draw[red, line width=1pt](.6,-.6)--(2.7,-.6);

        \end{scope}
    \end{tikzpicture}
    \caption{The 1-morphism with quantum dimension 0.}
    \label{fig:trace0}
\end{figure}
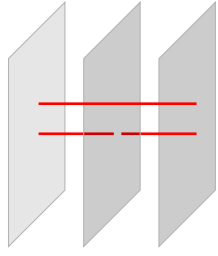

\subsection{Simplicial 2-morphisms}

In this subsection, we construct the indecomposible simplicial 2-morphisms up to annular equivalence.

\begin{theorem}
\label{Thm:delta2finite}
We construct a minimal representative set of indecomposible simplicial 2-morphisms
The set $$E_2=\{a_+,a_-,b_\tau,b_{g0},b_{g1},c_0,c_1\}.$$ is an $E_1$-complete minimal representative set of indecomposible simplicial 2-morphisms 
\end{theorem}

\begin{proof}
Now we consider simplicial 2-morphisms whose boundary are labelled by 1-morphisms in $E_1$ as in Fig.~\ref{fig: 1-morphisms labels of the boundary of a 2-simplex}.
Since only transversal intersection of the 2-simplex $\Delta^2$ and surfaces are curves, the boundary have even number of points. So $a_\tau,e_\tau$ and $e_\tau$ are illegal. 
As $g$ is inequivalent to $\mathbbm{1}$, the two points of $g$ cannot be connected by a curve in $\Delta^2$.
According to the fusion rule in Prop.~\ref{Prop: Ising fusion rule} or Fig.~\ref{fig:trace0}, any 2-morphism with boundary $e_g$ has zero quantum dimension. 
By the discussions in the proof of Theorem \ref{Thm: finite 2-morphism}, $\Delta^2$ has no closed curve when the boundary has points; and $\Delta^2$ has at most one closed curve when the boundary has no points.  We list all the possible curves in $\Delta^2$ with given boundary labels in $E_1$ in Fig.\ref{fig:labels of 2-simplices} 2-morphisms.
There are seven permissible simplical 2-morphisms $E_2=\{a_+,a_-,b_\tau,b_{g0},b_{g1},c_0,c_1\}.$ The 2-morphisms $e_i$, $1\leq i \leq 4$, have boundary $e_g$ and they have zero quantum dimension, so they are in $K_Z$.   
So the $E_2$ is $E_1$-complete.

Now we show that the seven are pairwise inequivalent, so $E_2$ is a minimal representative set. We only need to compare the ones with the same boundary.
By the definition of $r$ in Fig. \ref{fig: red tube}, $a_+$ and $a_-$ inequivalent 
By Lemma \ref{Lem: g 1 inequivalence}, $b_{g0}$ and $b_{g1}$ are inequivalent; $c_{0}$ and $c_{1}$ are inequivalent. 
\end{proof}

\begin{figure}[H]
    \centering
    \begin{tikzpicture}
    
    \begin{scope}[scale=.25]
        \begin{scope}
            \node at (-3.5,0){$a_1=$};
            \draw(0,2)--(1.732,-1)--(-1.732,-1)--(0,2);
            \node at (1.5,1){$\mathbbm{1}$};
            \node at (-1.5,1){$\mathbbm{1}$};
            \node at (0,-1.8){$\mathbbm{1}$};
        \end{scope}

         \begin{scope}[shift={(10,0)}]
        
            \node at (-3.5,0){$a_\tau=$};
            \draw(0,2)--(1.732,-1)--(-1.732,-1)--(0,2);
            \node at (1.5,1){$\mathbbm{1}$};
            \node at (-1.5,1){$\tau$};
            \node at (0,-1.8){$\mathbbm{1}$};
        \end{scope}

        \begin{scope}[shift={(20,0)}]
        
            \node at (-3.5,0){$a_g=$};
            \draw(0,2)--(1.732,-1)--(-1.732,-1)--(0,2);
            \node at (1.5,1){$\mathbbm{1}$};
            \node at (-1.5,1){$g$};
            \node at (0,-1.8){$\mathbbm{1}$};
        \end{scope}

        \begin{scope}[shift={(30,0)}]
        
            \node at (-3.5,0){$b_\tau=$};
            \draw(0,2)--(1.732,-1)--(-1.732,-1)--(0,2);
            \node at (1.5,1){$\tau$};
            \node at (-1.5,1){$\tau$};
            \node at (0,-1.8){$\mathbbm{1}$};
        \end{scope}

        \begin{scope}[shift={(0,-6)}]
        
            \node at (-3.5,0){$b_g=$};
            \draw(0,2)--(1.732,-1)--(-1.732,-1)--(0,2);
            \node at (1.5,1){$g$};
            \node at (-1.5,1){$g$};
            \node at (0,-1.8){$\mathbbm{1}$};
        \end{scope}

         \begin{scope}[shift={(10,-6)}]
        
            \node at (-3.5,0){$c=$};
            \draw(0,2)--(1.732,-1)--(-1.732,-1)--(0,2);
            \node at (1.5,1){$\tau$};
            \node at (-1.5,1){$\tau$};
            \node at (0,-1.8){$g$};
        \end{scope}

        \begin{scope}[shift={(30,-6)}]
        
            \node at (-3.5,0){$e_g=$};
            \draw(0,2)--(1.732,-1)--(-1.732,-1)--(0,2);
            \node at (1.5,1){$g$};
            \node at (-1.5,1){$g$};
            \node at (0,-1.8){$g$};
        \end{scope}

        \begin{scope}[shift={(20,-6)}]
        
            \node at (-3.5,0){$e_\tau=$};
            \draw(0,2)--(1.732,-1)--(-1.732,-1)--(0,2);
            \node at (1.5,1){$\tau$};
            \node at (-1.5,1){$\tau$};
            \node at (0,-1.8){$\tau$};
        \end{scope}
        \end{scope}
        
    \end{tikzpicture}
    \caption{Possible 1-morphisms labels of the boundary of a 2-simplex}
    \label{fig: 1-morphisms labels of the boundary of a 2-simplex}
\end{figure}

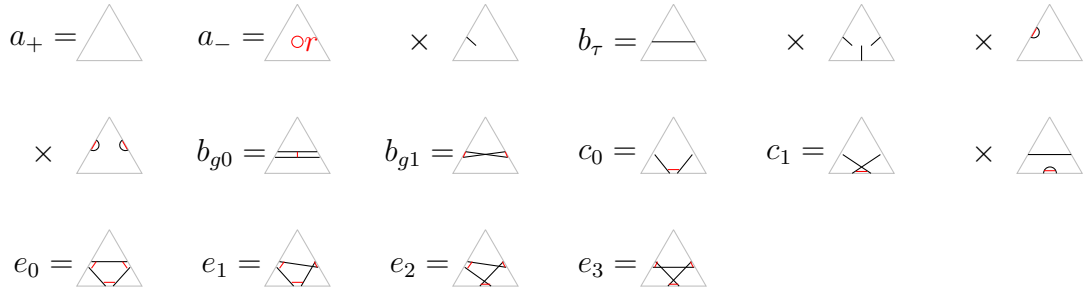
\begin{figure}[H]
    \centering
    \begin{tikzpicture}
        \begin{scope}[scale=.25]
        \begin{scope}[shift={(0,0)}]
            \node at (-3.5,0){$a_+=$};
            \draw[gray!50](0,2)--(1.732,-1)--(-1.732,-1)--(0,2);
        \end{scope}
               
        \begin{scope}[shift={(10,0)}]
        \draw[red] (0,0) circle(.3);
            \node at (-3.5,0){$a_-=$};    
            \node at (.75,-.1)[red]{$r$};
            \draw[gray!50](0,2)--(1.732,-1)--(-1.732,-1)--(0,2);
        \end{scope}

        \begin{scope}[shift={(50,0)}]
            \node at (-3.5,0){$\times$};
            \draw[gray!50](0,2)--(1.732,-1)--(-1.732,-1)--(0,2);

            \draw(-1,.25)arc[start angle=-120,end angle=60,radius=.3];
            \draw[red](-.95,.28)--(-.95+.25,.28+.25*1.732);
        \end{scope}

        \begin{scope}[shift={(30,0)}]

            \node at (-3.5,0){$b_\tau=$};
            \draw[gray!50](0,2)--(1.732,-1)--(-1.732,-1)--(0,2);
            \draw(1.15, 0)--(-1.15, 0);
        \end{scope}

        \begin{scope}[shift={(10,-6)}]

            \node at (-3.5,0){$b_{g0}=$};
            \draw[gray!50](0,2)--(1.732,-1)--(-1.732,-1)--(0,2);

            \draw(1.05, 0.15)--(-1.05, .15);
             \draw(1.18, -.15)--(-1.19, -.15);
             \draw[red](0,0.15)--(0,-.15);
        \end{scope}
        \begin{scope}[shift={(20,-6)}]

            \node at (-3.5,0){$b_{g1}=$};
            \draw[gray!50](0,2)--(1.732,-1)--(-1.732,-1)--(0,2);

            \draw(1.05, 0.15)--(-1.19, -.15);
             \draw(-1.05, .15)--(1.18, -.15);
             \draw[red](1.05, 0.15)--(1.18, -.15);
             \draw[red](-1.05, .15)--(-1.19, -.15);
         \end{scope}

         \begin{scope}[shift={(30,-6)}]

            \node at (-3.5,0){$c_0=$};
            \draw[gray!50](0,2)--(1.732,-1)--(-1.732,-1)--(0,2);

            \draw(1,0)--(0.2,-1);
            \draw(-1,0)--(-0.2,-1);

             \draw[red](-.3,-.8)--(.3,-.8);

        \end{scope}

        \begin{scope}[shift={(20,0)}]
            \node at (-3.5,0){$\times$};
            \draw[gray!50](0,2)--(1.732,-1)--(-1.732,-1)--(0,2);

            \draw(-1,.25)--(-.5,-.2);
        \end{scope}

        \begin{scope}[shift={(40,0)}]
            \node at (-3.5,0){$\times$};
            \draw[gray!50](0,2)--(1.732,-1)--(-1.732,-1)--(0,2);

            \draw(-1,.25)--(-.5,-.2);
            \draw(1,.25)--(.5,-.2);
            \draw(0,.-1)--(0,-.2);
        \end{scope}
        
        \begin{scope}[shift={(0,-6)}]
            \node at (-3.5,0){$\times$};
            \draw[gray!50](0,2)--(1.732,-1)--(-1.732,-1)--(0,2);

            \draw(-1,.25)arc[start angle=-120,end angle=60,radius=.3];
            \draw[red](-.95,.28)--(-.95+.25,.28+.25*1.732);

            \draw(1,.25)arc[start angle=300,end angle=120,radius=.3];
            \draw[red](.95,.28)--(.95-.25,.28+.25*1.732);
        \end{scope}

         \begin{scope}[shift={(40,-6)}]

            \node at (-3.5,0){$c_1=$};
            \draw[gray!50](0,2)--(1.732,-1)--(-1.732,-1)--(0,2);

            \draw(1,0)--(-0.5,-1);
            \draw(-1,0)--(0.5,-1);

             \draw[red](-.3,-.9)--(.3,-.9);

        \end{scope}

         \begin{scope}[shift={(50,-6)}]

            \node at (-3.5,0){$\times$};
            \draw[gray!50](0,2)--(1.732,-1)--(-1.732,-1)--(0,2);

            \draw(1.15,0)--(-1.15,0);

            \draw(-.32,-1)arc[start angle=180,end angle=0, radius=.34];
             \draw[red](-.3,-.86)--(.3,-.86);

        \end{scope}

         \begin{scope}[shift={(0,-12)}]

            \node at (-3.5,0){$e_0=$};
            \draw[gray!50](0,2)--(1.732,-1)--(-1.732,-1)--(0,2);

            \draw(1.1,0)--(0.2,-1);
            \draw(-1.1,0)--(-0.2,-1);

            \draw(0.95,0.3)--(-0.95,0.3);
            \draw[red](.95,-.05)--(.7,0.3);

            \draw[red](-.95,-.05)--(-.7,0.3);

             \draw[red](-.3,-.8)--(.3,-.8);

        \end{scope}

        \begin{scope}[shift={(10,-12)}]

            \node at (-3.5,0){$e_1=$};
            \draw[gray!50](0,2)--(1.732,-1)--(-1.732,-1)--(0,2);

            \draw(0.95,0.3)--(0.2,-1);%
            \draw(-1.1,0)--(-0.2,-1);

            \draw(1.1,0)--(-0.95,0.3);%
            \draw[red](1.05,-.0)--(.9,0.3);

            \draw[red](-.95,-.05)--(-.7,0.3);

             \draw[red](-.3,-.8)--(.3,-.8);

        \end{scope}

        \begin{scope}[shift={(20,-12)}]

            \node at (-3.5,0){$e_2=$};
            \draw[gray!50](0,2)--(1.732,-1)--(-1.732,-1)--(0,2);

            \draw(0.95,0.3)--(-0.3,-1);%
            \draw(-1.1,0)--(0.3,-1);

            \draw(1.1,0)--(-0.95,0.3);%
            \draw[red](1.05,-.0)--(.9,0.3);

            \draw[red](-.95,-.05)--(-.7,0.3);
             \draw[red](-.25,-.9)--(.25,-.9);

        \end{scope}
        
        \begin{scope}[shift={(30,-12)}]

            \node at (-3.5,0){$e_3=$};
            \draw[gray!50](0,2)--(1.732,-1)--(-1.732,-1)--(0,2);

            \draw(0.95,0.3)--(-0.3,-1);%
            \draw(-0.95,0.3)--(0.3,-1);

            \draw(1.1,0)--(-1.1,0);%
            \draw[red](1.05,.0)--(.9,0.3);

            \draw[red](-1,0)--(-.85,0.3);
             \draw[red](-.25,-.9)--(.25,-.9);

        \end{scope}

        \end{scope}

    \end{tikzpicture}
    \caption{Possible labels of 2-simplices}
    \label{fig:labels of 2-simplices}
\end{figure}

\subsection{3-morphisms}

In this subsection, we construct the minimal representative set of indecomposible simplicial 3-morphisms, which form a basis of the vector space with boundary labels in $E_2$.

\begin{theorem}
\label{Thm:delta3finite}
We construct a $E_2$-complete minimal representative set of 3-morphisms $E_3$ in Fig.\ref{fig:3-morphism set}.
\end{theorem}

\begin{proof}
We first label four 2-simplices on the boundary of $\Delta^3$ by elements in $E_2$, so that they are compatible on 1-simplices. For a fixed boundary, we construct 3-morphisms as surfaces with the given boundary.  
By Lemma \ref{Lem: g 1 inequivalence}, every connected surface cannot contains the boundary of a red line. 
The surface with one boundary curve is a disc. The surface with two boundary curves is a tube. 
Modular isotopy, we list all possible diagrams in Fig.\ref{fig:3-morphism set} and $E_3$ is $E_2$ complete.
It is a direct computation to check that 3-morphisms with the same boundary are linearly independent in the vector space. So $E_3$ is minimal.
\end{proof}

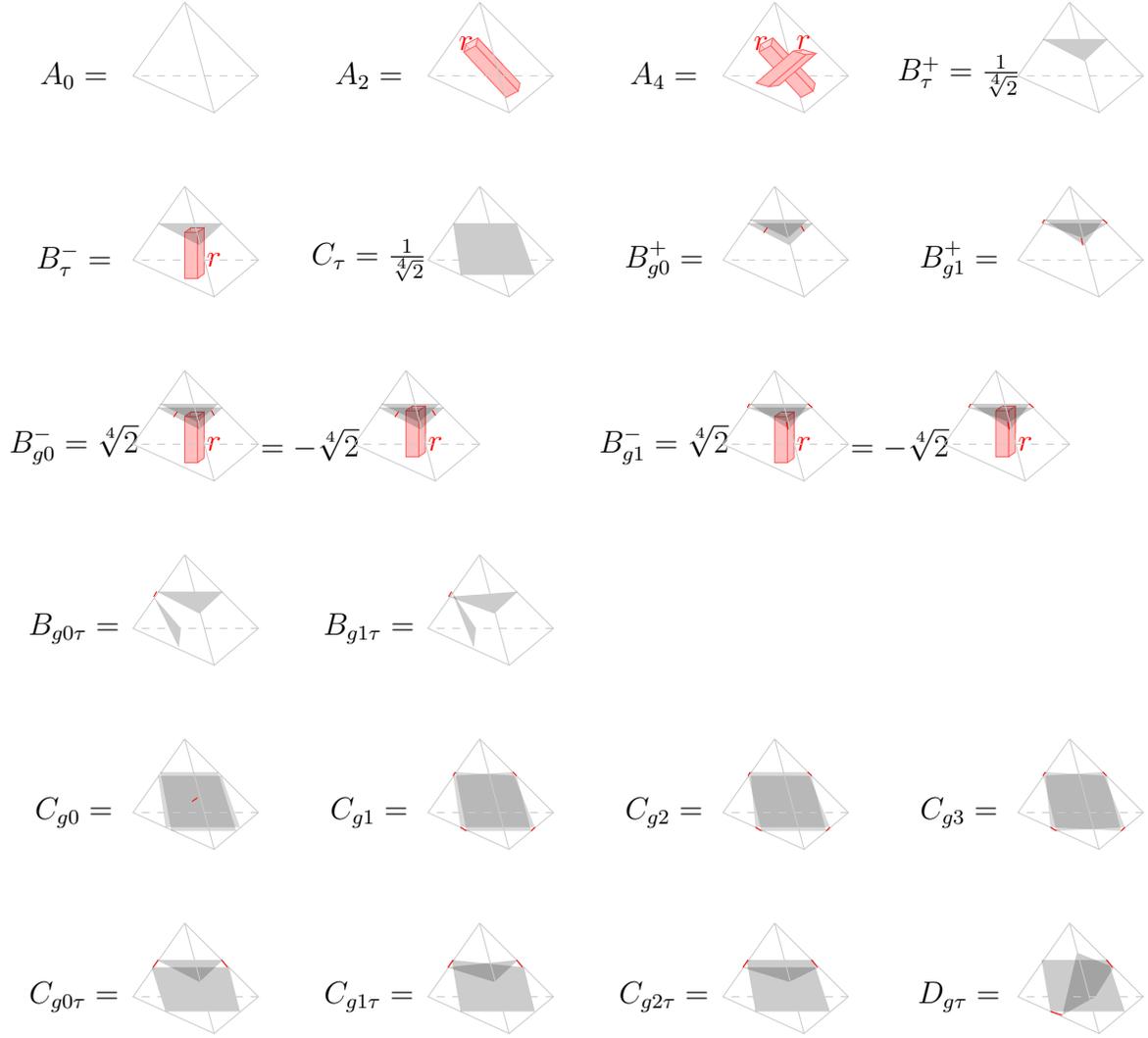
\begin{figure}[H]
    \centering
    \begin{tikzpicture}
        
    \begin{scope}
        \begin{scope}
        \node at (-1.5,-1) {$A_0=$};
            \draw[gray!40](0,0)--(1,-1)--(.4,-1.5)--(-.7,-1)--(0,0);
            \draw[gray!40,dashed](1,-1)--(-.7,-1);
            \draw[gray!40](0,0)--(.4,-1.5);
        \end{scope}

        \begin{scope}[shift={(4,0)}]
        \node at (-1.5,-1) {$A_2=$};

            \begin{scope}[shift={(-.1,-.8)},xscale=.4,yscale=.3,rotate=35]
        \begin{scope}[shift={(0,.5)}]

        \fill[red!20](0,0)--(.25,.25)--(.25,.25-2.5)--(0,-2.5);
        \node at (.3,.3)[red]{$r$};
        \draw[red!60](0,0)--(.25,.25)--(.25,.25-2.5);


        \fill[red!15](0.25,0.25)--(.75,.25)--(.75,-1.25)--(0.25,-1.25);
        \draw[red!60](0.25,0.25)--(.75,.25)--(.75,-1.25);

         \fill[red!25](0,0)--(.5,0)--(.5,-2.5)--(0,-2.5);

         \draw[red!60](0,0)--(.5,0)--(.5,-2.5)--(0,-2.5)--(0,0);
          \fill[red!25](0.5,0)--(.75,.25)--(.75,.25-2.5)--(0.5,-2.5);
          \draw[red!60](0.5,0)--(.75,.25)--(.75,.25-2.5)--(0.5,-2.5)--(0.5,0);
           \end{scope}
        \end{scope}

        \draw[gray!40](0,0)--(1,-1)--(.4,-1.5)--(-.7,-1)--(0,0);
            \draw[gray!40,dashed](1,-1)--(-.7,-1);
            \draw[gray!40](0,0)--(.4,-1.5);
        \end{scope}
        
        \begin{scope}[shift={(8,0)}]
        \node at (-1.5,-1) {$A_4=$};
        
            \draw[gray!40](0,0)--(1,-1)--(.4,-1.5)--(-.7,-1)--(0,0);
            \draw[gray!40,dashed](1,-1)--(-.7,-1);
            \draw[gray!40](0,0)--(.4,-1.5);

            \begin{scope}[shift={(-.1,-.8)},xscale=.4,yscale=.3,rotate=35]
        \begin{scope}[shift={(0,.5)}]

        \fill[red!20](0,0)--(.25,.25)--(.25,.25-2.5)--(0,-2.5);
        \draw[red!60](0,0)--(.25,.25)--(.25,.25-2.5);
        
        \node at (.3,.3)[red]{$r$};


        \fill[red!15](0.25,0.25)--(.75,.25)--(.75,-1.25)--(0.25,-1.25);
        \draw[red!60](0.25,0.25)--(.75,.25)--(.75,-1.25);

        \node at (1.5,-.5)[red]{$r$};

         \fill[red!25](0,0)--(.5,0)--(.5,-2.5)--(0,-2.5);

         \draw[red!60](0,0)--(.5,0)--(.5,-2.5)--(0,-2.5)--(0,0);
          \fill[red!25](0.5,0)--(.75,.25)--(.75,.25-2.5)--(0.5,-2.5);
          \draw[red!60](0.5,0)--(.75,.25)--(.75,.25-2.5)--(0.5,-2.5)--(0.5,0);
           \end{scope}
        \end{scope}

           \begin{scope}[shift={(.15,-.75)},xscale=.4,yscale=.2,rotate=-30]
        \begin{scope}[shift={(0,.5)}]

        \fill[red!20](0,0)--(.25,.25)--(.25,.25-2.5)--(0,-2.5);
        \draw[red!60](0,0)--(.25,.25)--(.25,.25-2.5);


        \fill[red!15](0.25,0.25)--(.75,.25)--(.75,-1.25)--(0.25,-1.25);
        \draw[red!60](0.25,0.25)--(.75,.25)--(.75,-1.25);

         \fill[red!25](0,0)--(.5,0)--(.5,-2.5)--(0,-2.5);

         \draw[red!60](0,0)--(.5,0)--(.5,-2.5)--(0,-2.5)--(0,0);
          \fill[red!25](0.5,0)--(.75,.25)--(.75,.25-2.5)--(0.5,-2.5);
          \draw[red!60](0.5,0)--(.75,.25)--(.75,.25-2.5)--(0.5,-2.5)--(0.5,0);
           \end{scope}
        \end{scope}

        \end{scope}

         \begin{scope}[shift={(12,0)}]
        \node at (-1.5,-1) {$B_{\tau}^+=\frac{1}{\sqrt[4]{2}}$};
            \draw[gray!40](0,0)--(1,-1)--(.4,-1.5)--(-.7,-1)--(0,0);
            \fill[fill opacity=.2](.52,-.5)--(-.38,-.5)--(.2,-.8);
            \draw[gray!40,dashed](1,-1)--(-.7,-1);
            \draw[gray!40](0,0)--(.4,-1.5);
        \end{scope}
        
         \begin{scope}[shift={(0,-2.5)}]
        \node at (-1.5,-1) {$B_{\tau}^-=$};

        \node at (.4,-1)[red]{$r$};

            \draw[gray!40,dashed](1,-1)--(-.7,-1);
           \begin{scope}[shift={(0,-.75)},xscale=.35,yscale=.25]
        \begin{scope}[shift={(0,.5)}]

        \fill[red!20](0,0)--(.25,.25)--(.25,.25-2.5)--(0,-2.5);
        \draw[red!60](0,0)--(.25,.25)--(.25,.25-2.5);


        \fill[red!15](0.25,0.25)--(.75,.25)--(.75,-1.25)--(0.25,-1.25);
        \draw[red!60](0.25,0.25)--(.75,.25)--(.75,-1.25);

         \fill[red!25](0,0)--(.5,0)--(.5,-2.5)--(0,-2.5);

         \draw[red!60](0,0)--(.5,0)--(.5,-2.5)--(0,-2.5)--(0,0);
          \fill[red!25](0.5,0)--(.75,.25)--(.75,.25-2.5)--(0.5,-2.5);
          \draw[red!60](0.5,0)--(.75,.25)--(.75,.25-2.5)--(0.5,-2.5)--(0.5,0);
           \end{scope}
        \end{scope}

            \fill[fill opacity=.2](.52,-.5)--(-.38,-.5)--(.2,-.8);
             \draw[gray!40](0,0)--(1,-1)--(.4,-1.5)--(-.7,-1)--(0,0);
               
            \draw[gray!40](0,0)--(.4,-1.5);
        \end{scope}

         \begin{scope}[shift={(4,-2.5)}]
        \node at (-1.5,-1) {$C_{\tau}=\frac{1}{\sqrt[4]{2}}$};
            \draw[gray!40](0,0)--(1,-1)--(.4,-1.5)--(-.7,-1)--(0,0);
            \fill[fill opacity=.2](.5,-.5)--(-.36,-.5)--(-.26,-1.2)--(.75,-1.2);
            \draw[gray!40,dashed](1,-1)--(-.7,-1);
            \draw[gray!40](0,0)--(.4,-1.5);
        \end{scope}

         \begin{scope}[shift={(8,-2.5)}]
            \node at (-1.5,-1) {$B_{g0}^+=$};
        
            \draw[gray!40](0,0)--(1,-1)--(.4,-1.5)--(-.7,-1)--(0,0);
            \fill[fill opacity=.2](.52,-.5)--(-.38,-.5)--(.2,-.8);
            \fill[fill opacity=.2](.46,-.45)--(-.35,-.45)--(.18,-.7);
            \draw[red](-.15,-.63)--(-.1,-.56);
            \draw[red](.4,-.61)--(.36,-.54);
            \draw[gray!40,dashed](1,-1)--(-.7,-1);
            \draw[gray!40](0,0)--(.4,-1.5);
        \end{scope}

        \begin{scope}[shift={(12,-2.5)}]
            \node at (-1.5,-1) {$B_{g1}^+=$};
        
            \draw[gray!40](0,0)--(1,-1)--(.4,-1.5)--(-.7,-1)--(0,0);
            \fill[fill opacity=.2](.52,-.5)--(-.38,-.5)--(.18,-.7);
            \fill[fill opacity=.2](.46,-.45)--(-.35,-.45)--(.2,-.8);
            \draw[red](.51,-.5)--(.45,-.45);
            \draw[red](.16,-.7)--(.18,-.8);
            \draw[red](-.37,-.5)--(-.34,-.45);
            \draw[gray!40,dashed](1,-1)--(-.7,-1);
            \draw[gray!40](0,0)--(.4,-1.5);
            
        \end{scope}
        
        \begin{scope}[shift={(0,-5)}]
        
        \node at (.4,-1)[red]{$r$};
        \node at (-1.5,-1) {$B_{g0}^{-}=\sqrt[4]{2}$};
         \begin{scope}[shift={(0,-.75)},xscale=.35,yscale=.25]
        \begin{scope}[shift={(0,.5)}]

        \fill[red!20](0,0)--(.25,.25)--(.25,.25-2.5)--(0,-2.5);
        \draw[red!60](0,0)--(.25,.25)--(.25,.25-2.5);


        \fill[red!15](0.25,0.25)--(.75,.25)--(.75,-1.25)--(0.25,-1.25);
        \draw[red!60](0.25,0.25)--(.75,.25)--(.75,-1.25);

         \fill[red!25](0,0)--(.5,0)--(.5,-2.5)--(0,-2.5);

         \draw[red!60](0,0)--(.5,0)--(.5,-2.5)--(0,-2.5)--(0,0);
          \fill[red!25](0.5,0)--(.75,.25)--(.75,.25-2.5)--(0.5,-2.5);
          \draw[red!60](0.5,0)--(.75,.25)--(.75,.25-2.5)--(0.5,-2.5)--(0.5,0);
           \end{scope}
        \end{scope}
           
            \draw[gray!40](0,0)--(1,-1)--(.4,-1.5)--(-.7,-1)--(0,0);
            \fill[fill opacity=.2](.52,-.5)--(-.38,-.5)--(.2,-.8);
            \fill[fill opacity=.2](.46,-.45)--(-.35,-.45)--(.18,-.7);
            \draw[red](-.15,-.63)--(-.1,-.56);
            \draw[red](.4,-.61)--(.36,-.54);
            \draw[gray!40,dashed](1,-1)--(-.7,-1);
            \draw[gray!40](0,0)--(.4,-1.5);
            
        \end{scope}
        
\begin{scope}[shift={(3,-5)}]

        \node at (.4,-1)[red]{$r$};
        \node at (-1.3,-1) {$=-\sqrt[4]{2}$};

        \begin{scope}[shift={(0,-.67)},xscale=.35,yscale=.25]
        \begin{scope}[shift={(0,.5)}]

        \fill[red!20](0,0)--(.25,.25)--(.25,.25-2.5)--(0,-2.5);
        \draw[red!60](0,0)--(.25,.25)--(.25,.25-2.5);


        \fill[red!15](0.25,0.25)--(.75,.25)--(.75,-1.25)--(0.25,-1.25);
        \draw[red!60](0.25,0.25)--(.75,.25)--(.75,-1.25);

         \fill[red!25](0,0)--(.5,0)--(.5,-2.5)--(0,-2.5);

         \draw[red!60](0,0)--(.5,0)--(.5,-2.5)--(0,-2.5)--(0,0);
          \fill[red!25](0.5,0)--(.75,.25)--(.75,.25-2.5)--(0.5,-2.5);
          \draw[red!60](0.5,0)--(.75,.25)--(.75,.25-2.5)--(0.5,-2.5)--(0.5,0);
           \end{scope}
        \end{scope}
           
            \draw[gray!40](0,0)--(1,-1)--(.4,-1.5)--(-.7,-1)--(0,0);
            \fill[fill opacity=.2](.52,-.5)--(-.38,-.5)--(.2,-.8);
            \fill[fill opacity=.2](.46,-.45)--(-.35,-.45)--(.18,-.7);
            \draw[red](-.15,-.63)--(-.1,-.56);
            \draw[red](.4,-.61)--(.36,-.54);
            \draw[gray!40,dashed](1,-1)--(-.7,-1);
            \draw[gray!40](0,0)--(.4,-1.5);

        \end{scope}

        \begin{scope}[shift={(8,-5)}]
            \node at (-1.5,-1) {$B_{g1}^{-}=\sqrt[4]{2}$};
            
        \node at (.4,-1)[red]{$r$};
            \begin{scope}[shift={(0,-.75)},xscale=.35,yscale=.25]
        \begin{scope}[shift={(0,.5)}]

        \fill[red!20](0,0)--(.25,.25)--(.25,.25-2.5)--(0,-2.5);
        \draw[red!60](0,0)--(.25,.25)--(.25,.25-2.5);


        \fill[red!15](0.25,0.25)--(.75,.25)--(.75,-1.25)--(0.25,-1.25);
        \draw[red!60](0.25,0.25)--(.75,.25)--(.75,-1.25);

         \fill[red!25](0,0)--(.5,0)--(.5,-2.5)--(0,-2.5);

         \draw[red!60](0,0)--(.5,0)--(.5,-2.5)--(0,-2.5)--(0,0);
          \fill[red!25](0.5,0)--(.75,.25)--(.75,.25-2.5)--(0.5,-2.5);
          \draw[red!60](0.5,0)--(.75,.25)--(.75,.25-2.5)--(0.5,-2.5)--(0.5,0);
           \end{scope}
        \end{scope}

            \draw[gray!40](0,0)--(1,-1)--(.4,-1.5)--(-.7,-1)--(0,0);
            \fill[fill opacity=.2](.52,-.5)--(-.38,-.5)--(.18,-.7);
            \fill[fill opacity=.2](.46,-.45)--(-.35,-.45)--(.2,-.8);
            \draw[red](.51,-.5)--(.45,-.45);
            \draw[red](.16,-.7)--(.18,-.8);
            \draw[red](-.37,-.5)--(-.34,-.45);
            
            \draw[gray!40,dashed](1,-1)--(-.7,-1);
            \draw[gray!40](0,0)--(.4,-1.5);

        \end{scope}

\begin{scope}[shift={(11,-5)}]
            \node at (-1.3,-1) {$=-\sqrt[4]{2}$};
            
        \node at (.4,-1)[red]{$r$};
            \begin{scope}[shift={(0,-.67)},xscale=.35,yscale=.25]
        \begin{scope}[shift={(0,.5)}]

        \fill[red!20](0,0)--(.25,.25)--(.25,.25-2.5)--(0,-2.5);
        \draw[red!60](0,0)--(.25,.25)--(.25,.25-2.5);


        \fill[red!15](0.25,0.25)--(.75,.25)--(.75,-1.25)--(0.25,-1.25);
        \draw[red!60](0.25,0.25)--(.75,.25)--(.75,-1.25);

         \fill[red!25](0,0)--(.5,0)--(.5,-2.5)--(0,-2.5);

         \draw[red!60](0,0)--(.5,0)--(.5,-2.5)--(0,-2.5)--(0,0);
          \fill[red!25](0.5,0)--(.75,.25)--(.75,.25-2.5)--(0.5,-2.5);
          \draw[red!60](0.5,0)--(.75,.25)--(.75,.25-2.5)--(0.5,-2.5)--(0.5,0);
           \end{scope}
        \end{scope}

            \draw[gray!40](0,0)--(1,-1)--(.4,-1.5)--(-.7,-1)--(0,0);
            \fill[fill opacity=.2](.52,-.5)--(-.38,-.5)--(.18,-.7);
            \fill[fill opacity=.2](.46,-.45)--(-.35,-.45)--(.2,-.8);
            \draw[red](.51,-.5)--(.45,-.45);
            \draw[red](.16,-.7)--(.18,-.8);
            \draw[red](-.37,-.5)--(-.34,-.45);
            
            \draw[gray!40,dashed](1,-1)--(-.7,-1);
            \draw[gray!40](0,0)--(.4,-1.5);

        \end{scope}

         \begin{scope}[shift={(0,-7.5)}]
        \node at (-1.5,-1) {$B_{g0\tau}=$};
            \draw[gray!40](0,0)--(1,-1)--(.4,-1.5)--(-.7,-1)--(0,0);
            
            \fill[fill opacity=.2](.52,-.5)--(-.38,-.5)--(.2,-.8);

            \fill[opacity=.2](-.42,-.57)--(-.05,-1)--(-.08,-1.26);
            \draw[red](-.42,-.57)--(-.38,-.5);
            
            \draw[gray!40,dashed](1,-1)--(-.7,-1);
            \draw[gray!40](0,0)--(.4,-1.5);
            
        \end{scope}

        \begin{scope}[shift={(4,-7.5)}]
        \node at (-1.5,-1) {$B_{g1\tau}=$};
            \draw[gray!40](0,0)--(1,-1)--(.4,-1.5)--(-.7,-1)--(0,0);
     
            \fill[fill opacity=.2](.52,-.5)--(-.42,-.57)--(.2,-.8);
            \fill[opacity=.2](-.38,-.5)--(-.05,-1)--(-.08,-1.26);
            \draw[red](-.42,-.57)--(-.38,-.5);
            \draw[gray!40,dashed](1,-1)--(-.7,-1);
            \draw[gray!40](0,0)--(.4,-1.5);
            
        \end{scope}

         \begin{scope}[shift={(0,-12.5)}]
        \node at (-1.5,-1) {$C_{g0\tau}=$};
            \draw[gray!40](0,0)--(1,-1)--(.4,-1.5)--(-.7,-1)--(0,0);
            
            \fill[fill opacity=.2](.52,-.5)--(-.38,-.5)--(.2,-.8);
            \fill[fill opacity=.2](.6,-.6)--(-.45,-.6)--(-.26,-1.2)--(.75,-1.2);

            \draw[red](.50,-.5)--(.58,-.6);
            \draw[red](-.36,-.5)--(-.43,-.6);
            \draw[gray!40,dashed](1,-1)--(-.7,-1);
            \draw[gray!40](0,0)--(.4,-1.5);
            
        \end{scope}

        \begin{scope}[shift={(4,-12.5)}]
        \node at (-1.5,-1) {$C_{g1\tau}=$};
            \draw[gray!40](0,0)--(1,-1)--(.4,-1.5)--(-.7,-1)--(0,0);
            
            \fill[fill opacity=.2](.52,-.5)--(-.45,-.6)--(.2,-.8);
            \fill[fill opacity=.2](.6,-.6)--(-.38,-.5)--(-.26,-1.2)--(.75,-1.2);

            \draw[red](.50,-.5)--(.58,-.6);
            \draw[red](-.36,-.5)--(-.43,-.6);
            \draw[gray!40,dashed](1,-1)--(-.7,-1);
            \draw[gray!40](0,0)--(.4,-1.5);
            
        \end{scope}

         \begin{scope}[shift={(8,-12.5)}]
        \node at (-1.5,-1) {$C_{g2\tau}=$};
            \draw[gray!40](0,0)--(1,-1)--(.4,-1.5)--(-.7,-1)--(0,0);
            
            \fill[fill opacity=.2](.6,-.6)--(-.45,-.6)--(.2,-.8);
            \fill[fill opacity=.2](.52,-.5)--(-.38,-.5)--(-.26,-1.2)--(.75,-1.2);

            \draw[red](.50,-.5)--(.58,-.6);
            \draw[red](-.36,-.5)--(-.43,-.6);
            \draw[gray!40,dashed](1,-1)--(-.7,-1);
            \draw[gray!40](0,0)--(.4,-1.5);
            
        \end{scope}
        
        \begin{scope}[shift={(12,-12.5)}]
        \node at (-1.5,-1) {$D_{g\tau}=$};
            \draw[gray!40](0,0)--(1,-1)--(.4,-1.5)--(-.7,-1)--(0,0);
            
            \fill[fill opacity=.2](.6,-.6)--(.35,-1)--(-.1,-1.25)--(.1,-.4);
            \fill[fill opacity=.2](.52,-.5)--(-.38,-.5)--(-.26,-1.2)--(.75,-1.2);

            \draw[red](.50,-.5)--(.58,-.6);
            \draw[red](-.1,-1.25)--(-.26,-1.2);
            \draw[gray!40,dashed](1,-1)--(-.7,-1);
            \draw[gray!40](0,0)--(.4,-1.5);
            
        \end{scope}

         \begin{scope}[shift={(0,-10)}]
        \node at (-1.5,-1) {$C_{g0}=$};
            \draw[gray!40](0,0)--(1,-1)--(.4,-1.5)--(-.7,-1)--(0,0);
            \fill[fill opacity=.15](.5,-.5)--(-.36,-.5)--(-.26,-1.2)--(.75,-1.2);
            \fill[fill opacity=.15](.45,-.45)--(-.33,-.45)--(-.18,-1.25)--(.7,-1.25);

            \draw[red](.1,-.85)--(.17,-.8);
            \draw[gray!40,dashed](1,-1)--(-.7,-1);
            \draw[gray!40](0,0)--(.4,-1.5);
            
        \end{scope}

        \begin{scope}[shift={(8,-10)}]
        \node at (-1.5,-1) {$C_{g2}=$};
            \draw[gray!40](0,0)--(1,-1)--(.4,-1.5)--(-.7,-1)--(0,0);
            \fill[fill opacity=.15](.5,-.5)--(-.36,-.5)--(-.18,-1.25)--(.7,-1.25);
            \fill[fill opacity=.15](.45,-.45)--(-.33,-.45)--(-.26,-1.2)--(.75,-1.2);

            \draw[red](-.36,-.5)--(-.33,-.45);
            \draw[red](.5,-.5)--(.45,-.45);
            \draw[red](-.26,-1.2)--(-.18,-1.25);
            \draw[red](.75,-1.2)--(.7,-1.25);
            \draw[gray!40,dashed](1,-1)--(-.7,-1);
            \draw[gray!40](0,0)--(.4,-1.5);
            
        \end{scope}

        \begin{scope}[shift={(4,-10)}]
        \node at (-1.5,-1) {$C_{g1}=$};
            \draw[gray!40](0,0)--(1,-1)--(.4,-1.5)--(-.7,-1)--(0,0);
            \fill[fill opacity=.15](.5,-.5)--(-.33,-.45)--(-.18,-1.25)--(.7,-1.25);
            \fill[fill opacity=.15](.45,-.45)--(-.36,-.5)--(-.26,-1.2)--(.75,-1.2);

            \draw[red](-.36,-.5)--(-.33,-.45);
            \draw[red](.5,-.5)--(.45,-.45);
            \draw[red](-.26,-1.2)--(-.18,-1.25);
            \draw[red](.75,-1.2)--(.7,-1.25);
            \draw[gray!40,dashed](1,-1)--(-.7,-1);
            \draw[gray!40](0,0)--(.4,-1.5);
            
        \end{scope}

         \begin{scope}[shift={(12,-10)}]
        \node at (-1.5,-1) {$C_{g3}=$};
            \draw[gray!40](0,0)--(1,-1)--(.4,-1.5)--(-.7,-1)--(0,0);
            \fill[fill opacity=.15](.45,-.45)--(-.36,-.5)--(-.18,-1.25)--(.75,-1.2);
            \fill[fill opacity=.15](.5,-.5)--(-.33,-.45)--(-.26,-1.2)--(.7,-1.25);

            \draw[red](-.36,-.5)--(-.33,-.45);
            \draw[red](.5,-.5)--(.45,-.45);
            \draw[red](-.26,-1.2)--(-.18,-1.25);
            \draw[red](.75,-1.2)--(.7,-1.25);
            \draw[gray!40,dashed](1,-1)--(-.7,-1);
            \draw[gray!40](0,0)--(.4,-1.5);
            
        \end{scope}

    \end{scope}
    \end{tikzpicture}
    \caption{A minimal representative set of 3-simplical morphisms (up to rotation and reflection in $\mathbb{R}^3$). All the tubes are red.}
    \label{fig:3-morphism set}
\end{figure}

\subsection{20j-Symbols}

\begin{theorem}
\label{Thm:20j symbols}
The $Z$ value of 4-simplices with compatible boundary morphisms in $E_3$ are listed in Table~\ref{tab:20j-symbols}, which we call the 20j symbols according to the 20 labels for 1-simplices and 2-simplices.
\end{theorem}

\begin{proof}
Take a 4-simplex whose five vertices of $\Delta^4$ are numbered from 1 to 5. 
We label the five 3-simplices of $\partial \Delta^4$ by morphisms in $E_3$, so that they have the same label on 2-simplices. We list the compatible labels in Tables~\ref{tab:20j-symbolsb} and \ref{tab:20j-symbolsb}.
Then the boundary of the 4-simplex is a linear sum of $S^3$ containing surfaces. The shape of the surfaces are listed in the column with head $\sigma^2$.  
We denote $S_\tau$ as a $\tau$-colored sphere, $T_{\tau-}$ as a red-colored torus, $T_{\tau-}^2$ as a $\tau$-colored genus-2 torus. Denote $S_g$ as a red-line connected double sphere. $T_g$ be digging an inner red tube intersecting an inner face $a_-$ on $S_g$, $T^2_g$ digging two red tubes.
Their $Z$ values can be computed directly by Def.~\ref{Def: Ising Z} and is shown in the column with head $\tilde{F}$.

Note that the 3-morphisms in $E_3$ are unnormalized vectors. The corresponding unnormalized 20j-symbol is $\widetilde{F}$. If we normalize the 3-morphisms to be unital vectors, then the normalized 20j-symbol is $F$.
\end{proof}

\begin{table}[H]
\centering
\setlength{\tabcolsep}{.75pt}
\renewcommand{\arraystretch}{0.93} 
\begin{tabular}{||c||c|c|c|c|c|c|c|c|c|c||c|c|c|c|c|c|c|c|c|c||c|c|c|c|c||}
\hline

id&12&13&14&15&23&24&25&34&35&45&123&124&125&134&135&145&234&235&245&345&1234&1235&1245&1345&2345

\\\hline

$0^0$&$1$&$1$&$1$&$1$&$1$&$1$&$1$&$1$&$1$&$1$&$a^+$&$a^+$&$a^+$&$a^+$&$a^+$&$a^+$&$a^+$&$a^+$&$a^+$&$a^+$&$A_0$&$A_0$&$A_0$&$A_0$&$A_0$\\
$0^1$&$1$&$1$&$1$&$1$&$1$&$1$&$1$&$1$&$1$&$1$&$a^+$&$a^+$&$a^+$&$a^+$&$a^+$&$a^-$&$a^+$&$a^+$&$a^-$&$a^-$&$A_0$&$A_0$&$A_2$&$A_2$&$A_2$\\
$0^2$&$1$&$1$&$1$&$1$&$1$&$1$&$1$&$1$&$1$&$1$&$a^+$&$a^+$&$a^+$&$a^+$&$a^-$&$a^-$&$a^+$&$a^-$&$a^-$&$a^+$&$A_0$&$A_2$&$A_2$&$A_2$&$A_2$\\
$0^3$&$1$&$1$&$1$&$1$&$1$&$1$&$1$&$1$&$1$&$1$&$a^+$&$a^+$&$a^-$&$a^-$&$a^+$&$a^-$&$a^-$&$a^-$&$a^+$&$a^+$&$A_2$&$A_2$&$A_2$&$A_2$&$A_2$\\
$0^4$&$1$&$1$&$1$&$1$&$1$&$1$&$1$&$1$&$1$&$1$&$a^+$&$a^+$&$a^-$&$a^-$&$a^+$&$a^+$&$a^-$&$a^-$&$a^-$&$a^-$&$A_2$&$A_2$&$A_2$&$A_2$&$A_4$\\
$0^5$&$1$&$1$&$1$&$1$&$1$&$1$&$1$&$1$&$1$&$1$&$a^+$&$a^+$&$a^+$&$a^-$&$a^-$&$a^-$&$a^-$&$a^-$&$a^-$&$a^-$&$A_2$&$A_2$&$A_2$&$A_4$&$A_4$\\
$0^6$&$1$&$1$&$1$&$1$&$1$&$1$&$1$&$1$&$1$&$1$&$a^-$&$a^-$&$a^-$&$a^-$&$a^-$&$a^-$&$a^-$&$a^-$&$a^-$&$a^-$&$A_4$&$A_4$&$A_4$&$A_4$&$A_4$\\
$1^0_\tau$&$1$&$1$&$1$&$\tau$&$1$&$1$&$\tau$&$1$&$\tau$&$\tau$&$a^+$&$a^+$&$b_\tau$&$a^+$&$b_\tau$&$b_\tau$&$a^+$&$b_\tau$&$b_\tau$&$b_\tau$&$A_0$&$B_\tau^+$&$B_\tau^+$&$B_\tau^+$&$B_\tau^+$\\
$1^1_\tau$&$1$&$1$&$1$&$\tau$&$1$&$1$&$\tau$&$1$&$\tau$&$\tau$&$a^+$&$a^+$&$b_\tau$&$a^-$&$b_\tau$&$b_\tau$&$a^-$&$b_\tau$&$b_\tau$&$b_\tau$&$A_2$&$B_\tau^+$&$B_\tau^+$&$B_\tau^-$&$B_\tau^-$\\
$1^2_\tau$&$1$&$1$&$1$&$\tau$&$1$&$1$&$\tau$&$1$&$\tau$&$\tau$&$a^-$&$a^-$&$b_\tau$&$a^-$&$b_\tau$&$b_\tau$&$a^-$&$b_\tau$&$b_\tau$&$b_\tau$&$A_4$&$B_\tau^-$&$B_\tau^-$&$B_\tau^-$&$B_\tau^-$\\
$1^0_{g0}$&$1$&$1$&$1$&$g$&$1$&$1$&$g$&$1$&$g$&$g$&$a^+$&$a^+$&$b_{g0}$&$a^+$&$b_{g0}$&$b_{g0}$&$a^+$&$b_{g0}$&$b_{g0}$&$b_{g0}$&$A_0$&$B_{g0}^+$&$B_{g0}^+$&$B_{g0}^+$&$B_{g0}^+$\\
$1^0_{g1}$&$1$&$1$&$1$&$g$&$1$&$1$&$g$&$1$&$g$&$g$&$a^+$&$a^+$&$b_{g0}$&$a^+$&$b_{g0}$&$b_{g1}$&$a^+$&$b_{g0}$&$b_{g1}$&$b_{g1}$&$A_0$&$B_{g0}^+$&$B_{g1}^+$&$B_{g1}^+$&$B_{g1}^+$\\
$1^0_{g2}$&$1$&$1$&$1$&$g$&$1$&$1$&$g$&$1$&$g$&$g$&$a^+$&$a^+$&$b_{g0}$&$a^+$&$b_{g1}$&$b_{g1}$&$a^+$&$b_{g1}$&$b_{g1}$&$b_{g0}$&$A_0$&$B_{g1}^+$&$B_{g1}^+$&$B_{g1}^+$&$B_{g1}^+$\\
$1^1_{g0}$&$1$&$1$&$1$&$g$&$1$&$1$&$g$&$1$&$g$&$g$&$a^+$&$a^+$&$b_{g0}$&$a^-$&$b_{g0}$&$b_{g0}$&$a^-$&$b_{g0}$&$b_{g0}$&$b_{g0}$&$A_2$&$B_{g0}^+$&$B_{g0}^+$&$B_{g0}^-$&$B_{g0}^-$\\
$1^1_{g1}$&$1$&$1$&$1$&$g$&$1$&$1$&$g$&$1$&$g$&$g$&$a^+$&$a^+$&$b_{g0}$&$a^-$&$b_{g0}$&$b_{g1}$&$a^-$&$b_{g0}$&$b_{g1}$&$b_{g1}$&$A_2$&$B_{g0}^+$&$B_{g1}^+$&$B_{g1}^-$&$B_{g1}^-$\\
$1^1_{g2}$&$1$&$1$&$1$&$g$&$1$&$1$&$g$&$1$&$g$&$g$&$a^+$&$a^+$&$b_{g0}$&$a^-$&$b_{g1}$&$b_{g1}$&$a^-$&$b_{g1}$&$b_{g1}$&$b_{g0}$&$A_2$&$B_{g1}^+$&$B_{g1}^+$&$B_{g1}^-$&$B_{g1}^-$\\
$1^2_{g0}$&$1$&$1$&$1$&$g$&$1$&$1$&$g$&$1$&$g$&$g$&$a^-$&$a^-$&$b_{g0}$&$a^-$&$b_{g0}$&$b_{g0}$&$a^-$&$b_{g0}$&$b_{g0}$&$b_{g0}$&$A_4$&$B_{g0}^-$&$B_{g0}^-$&$B_{g0}^-$&$B_{g0}^-$\\
$1^2_{g1}$&$1$&$1$&$1$&$g$&$1$&$1$&$g$&$1$&$g$&$g$&$a^-$&$a^-$&$b_{g0}$&$a^-$&$b_{g0}$&$b_{g1}$&$a^-$&$b_{g0}$&$b_{g1}$&$b_{g1}$&$A_4$&$B_{g0}^-$&$B_{g1}^-$&$B_{g1}^-$&$B_{g1}^-$\\
$1^2_{g2}$&$1$&$1$&$1$&$g$&$1$&$1$&$g$&$1$&$g$&$g$&$a^-$&$a^-$&$b_{g0}$&$a^-$&$b_{g1}$&$b_{g1}$&$a^-$&$b_{g1}$&$b_{g1}$&$b_{g0}$&$A_4$&$B_{g1}^-$&$B_{g1}^-$&$B_{g1}^-$&$B_{g1}^-$\\

$1_{g0\tau}^+$&$1$&$1$&$\tau$&$\tau$&$1$&$\tau$&$\tau$&$\tau$&$\tau$&$g$&$a^+$&$b_\tau$&$b_\tau$&$b_\tau$&$b_\tau$&$c_0$&$b_\tau$&$b_\tau$&$c_0$&$c_0$&$B_\tau^+$&$B_\tau^+$&$B_{g0\tau}$&$B_{g0\tau}$&$B_{g0\tau}$\\
$1_{g1\tau}^+$&$1$&$1$&$\tau$&$\tau$&$1$&$\tau$&$\tau$&$\tau$&$\tau$&$g$&$a^+$&$b_\tau$&$b_\tau$&$b_\tau$&$b_\tau$&$c_1$&$b_\tau$&$b_\tau$&$c_1$&$c_1$&$B_\tau^+$&$B_\tau^+$&$B_{g1\tau}$&$B_{g1\tau}$&$B_{g1\tau}$\\
$1_{g0\tau}^-$&$1$&$1$&$\tau$&$\tau$&$1$&$\tau$&$\tau$&$\tau$&$\tau$&$g$&$a^-$&$b_\tau$&$b_\tau$&$b_\tau$&$b_\tau$&$c_0$&$b_\tau$&$b_\tau$&$c_0$&$c_0$&$B_\tau^-$&$B_\tau^-$&$B_{g0\tau}$&$B_{g0\tau}$&$B_{g0\tau}$\\
$1_{g1\tau}^-$&$1$&$1$&$\tau$&$\tau$&$1$&$\tau$&$\tau$&$\tau$&$\tau$&$g$&$a^-$&$b_\tau$&$b_\tau$&$b_\tau$&$b_\tau$&$c_1$&$b_\tau$&$b_\tau$&$c_1$&$c_1$&$B_\tau^-$&$B_\tau^-$&$B_{g1\tau}$&$B_{g1\tau}$&$B_{g1\tau}$\\

\hline

\end{tabular}
\caption{\label{tab:20j-symbolsb}Labels of simplices of $\Delta^4$.}
\end{table}

\begin{table}[H]
\centering
\setlength{\tabcolsep}{.75pt}
\renewcommand{\arraystretch}{0.93} 
\begin{tabular}{||c||c|c|c|c|c|c|c|c|c|c||c|c|c|c|c|c|c|c|c|c||c|c|c|c|c||}
\hline

id&12&13&14&15&23&24&25&34&35&45&123&124&125&134&135&145&234&235&245&345&1234&1235&1245&1345&2345

\\\hline
$2_\tau^+$&$1$&$1$&$\tau$&$\tau$&$1$&$\tau$&$\tau$&$\tau$&$\tau$&$1$&$a^+$&$b_\tau$&$b_\tau$&$b_\tau$&$b_\tau$&$b_\tau$&$b_\tau$&$b_\tau$&$b_\tau$&$b_\tau$&$B_\tau^+$&$B_\tau^+$&$C_\tau$&$C_\tau$&$C_\tau$\\
$2_\tau^-$&$1$&$1$&$\tau$&$\tau$&$1$&$\tau$&$\tau$&$\tau$&$\tau$&$1$&$a^-$&$b_\tau$&$b_\tau$&$b_\tau$&$b_\tau$&$b_\tau$&$b_\tau$&$b_\tau$&$b_\tau$&$b_\tau$&$B_\tau^-$&$B_\tau^-$&$C_\tau$&$C_\tau$&$C_\tau$\\
$2_{g0\tau}^+$&$1$&$1$&$\tau$&$g$&$1$&$\tau$&$g$&$\tau$&$g$&$\tau$&$a^+$&$b_\tau$&$b_{g0}$&$b_\tau$&$b_{g0}$&$c_0$&$b_\tau$&$b_{g0}$&$c_0$&$c_0$&$B_\tau^+$&$B_{g0}^+$&$C_{g0\tau}$&$C_{g0\tau}$&$C_{g0\tau}$\\

$2_{g1\tau}^+$&$1$&$1$&$\tau$&$g$&$1$&$\tau$&$g$&$\tau$&$g$&$\tau$&$a^+$&$b_\tau$&$b_{g0}$&$b_\tau$&$b_{g1}$&$c_1$&$b_\tau$&$b_{g1}$&$c_1$&$c_0$&$B_\tau^+$&$B_{g1}^+$&$C_{g2\tau}$&$C_{g1\tau}$&$C_{g1\tau}$\\
$2_{g2\tau}^+$&$1$&$1$&$\tau$&$g$&$1$&$\tau$&$g$&$\tau$&$g$&$\tau$&$a^+$&$b_\tau$&$b_{g0}$&$b_\tau$&$b_{g0}$&$c_1$&$b_\tau$&$b_{g0}$&$c_1$&$c_1$&$B_\tau^+$&$B_{g0}^+$&$C_{g2\tau}$&$C_{g2\tau}$&$C_{g2\tau}$\\
$2_{g0\tau}^-$&$1$&$1$&$\tau$&$g$&$1$&$\tau$&$g$&$\tau$&$g$&$\tau$&$a^-$&$b_\tau$&$b_{g0}$&$b_\tau$&$b_{g0}$&$c_0$&$b_\tau$&$b_{g0}$&$c_0$&$c_0$&$B_\tau^-$&$B_{g0}^-$&$C_{g0\tau}$&$C_{g0\tau}$&$C_{g0\tau}$\\
$2_{g1\tau}^-$&$1$&$1$&$\tau$&$g$&$1$&$\tau$&$g$&$\tau$&$g$&$\tau$&$a^-$&$b_\tau$&$b_{g0}$&$b_\tau$&$b_{g1}$&$c_1$&$b_\tau$&$b_{g1}$&$c_1$&$c_0$&$B_\tau^-$&$B_{g1}^-$&$C_{g2\tau}$&$C_{g1\tau}$&$C_{g1\tau}$\\
$2_{g2\tau}^-$&$1$&$1$&$\tau$&$g$&$1$&$\tau$&$g$&$\tau$&$g$&$\tau$&$a^-$&$b_\tau$&$b_{g0}$&$b_\tau$&$b_{g0}$&$c_1$&$b_\tau$&$b_{g0}$&$c_1$&$c_1$&$B_\tau^-$&$B_{g0}^-$&$C_{g2\tau}$&$C_{g2\tau}$&$C_{g2\tau}$\\
$2_{g2\tau}^-$&$1$&$1$&$\tau$&$g$&$1$&$\tau$&$g$&$\tau$&$g$&$\tau$&$a^-$&$b_\tau$&$b_{g0}$&$b_\tau$&$b_{g0}$&$c_1$&$b_\tau$&$b_{g0}$&$c_1$&$c_1$&$B_\tau^-$&$B_{g0}^-$&$C_{g2\tau}$&$C_{g2\tau}$&$C_{g2\tau}$\\

$2_{g0}^+$&$1$&$1$&$g$&$g$&$1$&$g$&$g$&$g$&$g$&$1$&$a^+$&$b_{g0}$&$b_{g0}$&$b_{g0}$&$b_{g0}$&$b_{g0}$&$b_{g0}$&$b_{g0}$&$b_{g0}$&$b_{g0}$&$B_{g0}^+$&$B_{g0}^+$&$C_{g0}$&$C_{g0}$&$C_{g0}$\\
$2_{g1}^+$&$1$&$1$&$g$&$g$&$1$&$g$&$g$&$g$&$g$&$1$&$a^+$&$b_{g0}$&$b_{g0}$&$b_{g0}$&$b_{g1}$&$b_{g1}$&$b_{g0}$&$b_{g1}$&$b_{g1}$&$b_{g0}$&$B_{g0}^+$&$B_{g1}^+$&$C_{g2}$&$C_{g1}$&$C_{g1}$\\
$2_{g2}^+$&$1$&$1$&$g$&$g$&$1$&$g$&$g$&$g$&$g$&$1$&$a^+$&$b_{g0}$&$b_{g1}$&$b_{g1}$&$b_{g0}$&$b_{g1}$&$b_{g1}$&$b_{g1}$&$b_{g0}$&$b_{g0}$&$B_{g1}^+$&$B_{g1}^+$&$C_{g1}$&$C_{g1}$&$C_{g2}$\\
$2_{g3}^+$&$1$&$1$&$g$&$g$&$1$&$g$&$g$&$g$&$g$&$1$&$a^+$&$b_{g0}$&$b_{g0}$&$b_{g0}$&$b_{g0}$&$b_{g1}$&$b_{g0}$&$b_{g0}$&$b_{g1}$&$b_{g1}$&$B_{g0}^+$&$B_{g0}^+$&$C_{g2}$&$C_{g2}$&$C_{g2}$\\
$2_{g4}^+$&$1$&$1$&$g$&$g$&$1$&$g$&$g$&$g$&$g$&$1$&$a^+$&$b_{g0}$&$b_{g1}$&$b_{g1}$&$b_{g0}$&$b_{g0}$&$b_{g1}$&$b_{g1}$&$b_{g1}$&$b_{g1}$&$B_{g1}^+$&$B_{g1}^+$&$C_{g1}$&$C_{g1}$&$C_{g3}$\\

$2_{g5}^+$&$1$&$1$&$g$&$g$&$1$&$g$&$g$&$g$&$g$&$1$&$a^+$&$b_{g0}$&$b_{g0}$&$b_{g1}$&$b_{g1}$&$b_{g1}$&$b_{g1}$&$b_{g1}$&$b_{g1}$&$b_{g1}$&$B_{g1}^+$&$B_{g1}^+$&$C_{g2}$&$C_{g3}$&$C_{g3}$\\

$2_{g0}^-$&$1$&$1$&$g$&$g$&$1$&$g$&$g$&$g$&$g$&$1$&$a^-$&$b_{g0}$&$b_{g0}$&$b_{g0}$&$b_{g0}$&$b_{g0}$&$b_{g0}$&$b_{g0}$&$b_{g0}$&$b_{g0}$&$B_{g0}^-$&$B_{g0}^-$&$C_{g0}$&$C_{g0}$&$C_{g0}$\\
$2_{g1}^-$&$1$&$1$&$g$&$g$&$1$&$g$&$g$&$g$&$g$&$1$&$a^-$&$b_{g0}$&$b_{g0}$&$b_{g0}$&$b_{g1}$&$b_{g1}$&$b_{g0}$&$b_{g1}$&$b_{g1}$&$b_{g0}$&$B_{g0}^-$&$B_{g1}^-$&$C_{g2}$&$C_{g1}$&$C_{g1}$\\
$2_{g2}^-$&$1$&$1$&$g$&$g$&$1$&$g$&$g$&$g$&$g$&$1$&$a^-$&$b_{g0}$&$b_{g1}$&$b_{g1}$&$b_{g0}$&$b_{g1}$&$b_{g1}$&$b_{g1}$&$b_{g0}$&$b_{g0}$&$B_{g1}^-$&$B_{g1}^-$&$C_{g1}$&$C_{g1}$&$C_{g2}$\\
$2_{g3}^-$&$1$&$1$&$g$&$g$&$1$&$g$&$g$&$g$&$g$&$1$&$a^-$&$b_{g0}$&$b_{g0}$&$b_{g0}$&$b_{g0}$&$b_{g1}$&$b_{g0}$&$b_{g0}$&$b_{g1}$&$b_{g1}$&$B_{g0}^-$&$B_{g0}^-$&$C_{g2}$&$C_{g2}$&$C_{g2}$\\
$2_{g4}^-$&$1$&$1$&$g$&$g$&$1$&$g$&$g$&$g$&$g$&$1$&$a^-$&$b_{g0}$&$b_{g1}$&$b_{g1}$&$b_{g0}$&$b_{g0}$&$b_{g1}$&$b_{g1}$&$b_{g1}$&$b_{g1}$&$B_{g1}^-$&$B_{g1}^-$&$C_{g1}$&$C_{g1}$&$C_{g3}$\\
$2_{g5}^-$&$1$&$1$&$g$&$g$&$1$&$g$&$g$&$g$&$g$&$1$&$a^-$&$b_{g0}$&$b_{g0}$&$b_{g1}$&$b_{g1}$&$b_{g1}$&$b_{g1}$&$b_{g1}$&$b_{g1}$&$b_{g1}$&$B_{g1}^-$&$B_{g1}^-$&$C_{g2}$&$C_{g3}$&$C_{g3}$\\

$2_{g0\tau}$&$1$&$\tau$&$\tau$&$\tau$&$\tau$&$\tau$&$\tau$&$1$&$g$&$g$&$b_\tau$&$b_\tau$&$b_\tau$&$b_\tau$&$c_0$&$c_0$&$b_\tau$&$c_0$&$c_0$&$b_{g0}$&$C_\tau$&$B_{g0\tau}$&$B_{g0\tau}$&$C_{g0\tau}$&$C_{g0\tau}$\\
$2_{g1\tau}$&$1$&$\tau$&$\tau$&$\tau$&$\tau$&$\tau$&$\tau$&$1$&$g$&$g$&$b_\tau$&$b_\tau$&$b_\tau$&$b_\tau$&$c_0$&$c_1$&$b_\tau$&$c_0$&$c_1$&$b_{g1}$&$C_\tau$&$B_{g0\tau}$&$B_{g1\tau}$&$C_{g1\tau}$&$C_{g1\tau}$\\
$2_{g2\tau}$&$1$&$\tau$&$\tau$&$\tau$&$\tau$&$\tau$&$\tau$&$1$&$g$&$g$&$b_\tau$&$b_\tau$&$b_\tau$&$b_\tau$&$c_1$&$c_1$&$b_\tau$&$c_1$&$c_1$&$b_{g0}$&$C_\tau$&$B_{g1\tau}$&$B_{g1\tau}$&$C_{g2\tau}$&$C_{g2\tau}$\\

$3_{g0\tau}$&$1$&$\tau$&$g$&$g$&$\tau$&$g$&$g$&$\tau$&$\tau$&$1$&$b_\tau$&$b_{g0}$&$b_{g0}$&$c_0$&$c_0$&$b_{g0}$&$c_0$&$c_0$&$b_{g0}$&$b_\tau$&$C_{g0\tau}$&$C_{g0\tau}$&$C_{g0}$&$C_{g0\tau}$&$C_{g0\tau}$\\
$3_{g1\tau}$&$1$&$\tau$&$g$&$g$&$\tau$&$g$&$g$&$\tau$&$\tau$&$1$&$b_\tau$&$b_{g0}$&$b_{g0}$&$c_0$&$c_1$&$b_{g1}$&$c_0$&$c_1$&$b_{g1}$&$b_\tau$&$C_{g0\tau}$&$C_{g2\tau}$&$C_{g2}$&$C_{g1\tau}$&$C_{g1\tau}$\\
$3_{g2\tau}$&$1$&$\tau$&$g$&$g$&$\tau$&$g$&$g$&$\tau$&$\tau$&$1$&$b_\tau$&$b_{g0}$&$b_{g1}$&$c_0$&$c_0$&$b_{g0}$&$c_0$&$c_1$&$b_{g1}$&$b_\tau$&$C_{g0\tau}$&$C_{g1\tau}$&$C_{g1}$&$C_{g0\tau}$&$C_{g1\tau}$\\
$3_{g3\tau}$&$1$&$\tau$&$g$&$g$&$\tau$&$g$&$g$&$\tau$&$\tau$&$1$&$b_\tau$&$b_{g0}$&$b_{g1}$&$c_1$&$c_0$&$b_{g1}$&$c_1$&$c_1$&$b_{g0}$&$b_\tau$&$C_{g2\tau}$&$C_{g1\tau}$&$C_{g1}$&$C_{g1\tau}$&$C_{g2\tau}$\\
$3_{g4\tau}$&$1$&$\tau$&$g$&$g$&$\tau$&$g$&$g$&$\tau$&$\tau$&$1$&$b_\tau$&$b_{g1}$&$b_{g1}$&$c_0$&$c_1$&$b_{g1}$&$c_1$&$c_0$&$b_{g1}$&$b_\tau$&$C_{g1\tau}$&$C_{g1\tau}$&$C_{g3}$&$C_{g1\tau}$&$C_{g1\tau}$\\

$4_{g0\tau}$&$1$&$\tau$&$\tau$&$g$&$\tau$&$\tau$&$g$&$g$&$\tau$&$\tau$&$b_\tau$&$b_\tau$&$b_{g0}$&$c_0$&$c_0$&$c_1$&$c_0$&$c_0$&$c_1$&$c_1$&$B_{g0\tau}$&$C_{g0\tau}$&$C_{g2\tau}$&$D_{g\tau}$&$D_{g\tau}$\\
$4_{g1\tau}$&$1$&$\tau$&$\tau$&$g$&$\tau$&$\tau$&$g$&$g$&$\tau$&$\tau$&$b_\tau$&$b_\tau$&$b_{g1}$&$c_0$&$c_0$&$c_1$&$c_0$&$c_1$&$c_0$&$c_1$&$B_{g0\tau}$&$C_{g1\tau}$&$C_{g1\tau}$&$D_{g\tau}$&$D_{g\tau}$\\
$4_{g2\tau}$&$1$&$\tau$&$\tau$&$g$&$\tau$&$\tau$&$g$&$g$&$\tau$&$\tau$&$b_\tau$&$b_\tau$&$b_{g1}$&$c_1$&$c_0$&$c_1$&$c_1$&$c_1$&$c_0$&$c_0$&$B_{g1\tau}$&$C_{g1\tau}$&$C_{g1\tau}$&$D_{g\tau}$&$D_{g\tau}$\\\hline

\end{tabular}
\caption{\label{tab:20j-symbolsa} Labels of simplices of $\Delta^4$.}
\end{table}

\begin{table}[thp]
    \centering
    
\setlength{\tabcolsep}{1pt}
    \begin{tabular}{||c||c|c|c|c|c||c|c|c||c|}
\hline

id&1234&1235&1245&1345&2345&$\sigma^2$&$\widetilde{F}$&$F$ & intersecting variants
\\\hline

$0^0$&$A_0$&$A_0$&$A_0$&$A_0$&$A_0$&$\emptyset$&1&1&$\emptyset$\\

$0^1$&$A_0$&$A_0$&$A_2$&$A_2$&$A_2$&$T_{\tau-}$&1&1&$\emptyset$\\

$0^2$&$A_2$&$A_2$&$A_2$&$A_2$&$A_2$&$T_{\tau-}$&1&1&$\emptyset$\\

$0^3$&$A_2$&$A_2$&$A_2$&$A_2$&$A_4$&$2T_{\tau-}$&1&1&$\emptyset$\\

$0^4$&$A_2$&$A_2$&$A_2$&$A_4$&$A_4$&$2T_{\tau-}$&1&1&$\emptyset$\\

$0^5$&$A_4$&$A_4$&$A_4$&$A_4$&$A_4$&$3T_{\tau-}$&1&1&$\emptyset$\\

$1^0_\tau$&$A_0$&$B_\tau^+$&$B_\tau^+$&$B_\tau^+$&$B_\tau^+$&$S_\tau$&$\sqrt{2}$&$\sqrt{2}/2$&$\emptyset$\\

$1^1_\tau$&$A_2$&$B_\tau^-$&$B_\tau^-$&$B_\tau^+$&$B_\tau^+$&$T_{\tau-}$&1&$\sqrt{2}/2$&$\emptyset$\\

$1^2_\tau$&$A_4$&$B_\tau^-$&$B_\tau^-$&$B_\tau^-$&$B_\tau^-$&$T^{2}_{\tau-}$&$\sqrt{2}/2$&$\sqrt{2}/2$&$\emptyset$\\

$2^+_\tau$&$C_\tau$&$B_\tau^+$&$C_\tau$&$B_\tau^+$&$C_\tau$&$S_\tau$&$\sqrt{2}$&$\sqrt{2}/2^{5/4}$&$\emptyset$\\

$2^-_\tau$&$C_{\tau}$&$B_\tau^-$&$C_\tau$&$B_\tau^-$&$C_\tau$&$T_{\tau-}$&$1$&$1/2^{3/4}$&$\emptyset$\\

$1^0_{g0}$&$A_0$&$B_{g0}^+$&$B_{g0}^+$&$B_{g0}^+$&$B_{g0}^+$&$S_g$&$1$&1&$1^0_{g1}, 1^0_{g2}$\\

$1^1_{g0}$&$A_2$&$B_{g0}^-$&$B_{g0}^-$&$B_{g0}^+$&$B_{g0}^+$&$T_{g-}$&$\sqrt{2}-\sqrt{2}/2$&1&$1^1_{g1},1^1_{g2}$\\

$1^2_{g0}$&$A_4$&$B_{g0}^-$&$B_{g0}^-$&$B_{g0}^-$&$B_{g0}^-$&$T^2_{g-}$&$1/2$&$1$&$1^2_{g1},1^2_{g2}$\\

$1_{g0\tau}^+$&$B_\tau^+$&$B_\tau^+$&$B_{g0\tau}$&$B_{g0\tau}$&$B_{g0\tau}$&$S_g$&$2-1$&$1/\sqrt{2}$&$1^+_{g1\tau}$\\

$1_{g0\tau}^-$&$B_\tau^-$&$B_\tau^-$&$B_{g0\tau}$&$B_{g0\tau}$&$B_{g0\tau}$&$\frac{1}{\sqrt{2}}S_g$&$1/\sqrt{2}$&$1/\sqrt{2}$&$1^-_{g1\tau}$\\

$2^+_{g0\tau}$&$B_\tau^+$&$B_{g0}^+$&$C_{g0\tau}$&$C_{g0\tau}$&$C_{g0\tau}$&$S_g$&$2-1$&$1/2^{1/4}$&$2^+_{g1\tau},2^+_{g2\tau}$\\

$2^-_{g0\tau}$&$B_\tau^-$&$B_{g0}^-$&$C_{g0\tau}$&$C_{g0\tau}$&$C_{g0\tau}$&$\frac{1}{\sqrt{2}}S_g$&$1/\sqrt{2}$&$1/2^{1/4}$&$2^-_{g1\tau},2^-_{g2\tau}$\\

$2_{g0\tau}$&$C_\tau$&$B_{g0\tau}$&$B_{g0\tau}$&$C_{g0\tau}$&$C_{g0\tau}$&$S_g$&$2-1$&$1/2^{1/4}$&$2_{g1\tau},2_{g2\tau}$\\

$2^+_{g0}$&$B_{g0}^+$&$B_{g0}^+$&$C_{g0}$&$C_{g0}$&$C_{g0}$&$S_g$&$2-1$&1&$2^+_{g1},2^+_{g2},2^+_{g3},2^+_{g4},2^+_{g5}$\\

$2^-_{g0}$&$B_{g0}^-$&$B_{g0}^-$&$C_{g0}$&$C_{g0}$&$C_{g0}$&$T_{g-}$&$\sqrt{2}/2$&$1$&$2^-_{g1},2^-_{g2},2^-_{g3},2^-_{g4},2^-_{g5}$\\

$3_{g0\tau}$&$C_{g0\tau}$&$C_{g0\tau}$&$C_{g0}$&$C_{g0\tau}$&$C_{g0\tau}$&$S_g$&$2-1$&1&$3_{g1\tau},3_{g2\tau},3_{g3\tau},3_{g4\tau}$\\

$4_{g0\tau}$&$B_{g0\tau}$&$C_{g0\tau}$&$C_{g2\tau}$&$D_{g\tau}$&$D_{g\tau}$&$S_g$&$2-1$&1&$4_{g1\tau},4_{g2\tau}$ \\

\hline

\end{tabular}
    \caption{Table of bulks, surfaces and values of 20j-symbols. Here $\tilde{F}$ is unnormalized and $F$ is normalized. The last column presents the other simplex id with the same $F$ symbols.}
    \label{tab:20j-symbols}
\end{table}

\clearpage

\subsection{Pachner Moves}
A well-known result of Pachner states that two triangulations of a PL-manifold are related by a finite sequence of Pachner moves \cite{Pan91}. The boundary of the $(n+2)$-simplex $\Delta^{n+2}$ has $(n+1)$-simplices $\Delta_{n+1,i}$, $1\leq i\leq n+3$. The $(k,n+3-k)$ Pachner move changes $k$ $(n+1)$-simplices to the rest $(n+3-k)$ $(n+1)$-simplices.

When $n=2$, it is well-know that the value $F$ of $\Delta_{3,i}$ labelled by morphisms of a spherical 2-category on the four 0-faces is called the quantum 6j-symbol or the $F$-symbol. They satisfy the pentagon equation corresponding to the $(2,3)$ Pachner move. The number $6$ comes from the number of $1$-simplices in $\Delta^3$, and the spin $j$ was originally considered as the spin $j$ irreducible representation of $SU(2)$. For a general $n$, the value $F$ of $\Delta_{n+1,i}$ labelled by $n$-morphisms of a spherical $n$-category on the $n+2$ 0-faces is called the $F$-symbol. 

By Theorem \ref{Thm: Semisimple invariant}, the $F$-symbols satisfy the $(k,n+3-k)$ Pachner move for all $k$. In particular, the 20j-symbols satisfies the higher pentagon equations corresponding to $(3,3)$, $(2,4)$ and $(1,5)$ Pachner moves. To double check the theory and the numerical computation, we have verified all Pachner moves of 20j-symbols on a computer. There are 2044 (1-5) equations, 30464 (4-2) equations and 50709 (3-3) equations.

We list the first two identities of Pachner moves of our 20j-symbols as examples.

\begin{table}[thp]
\centering
\setlength{\tabcolsep}{.6pt}
\begin{tabular}{|c||c|c|c|c|c||c||c|c|c|c|c|c|c|c||}

\hline
id&1234&1235&1245&1345&2345&16,...,56&$\prod\limits_{i=1}^5\frac{Tr(\Delta^1_i)}{\mu_1}$&$126,...,456$&$\prod\limits_{i=1}^{10}\frac{Tr(\Delta^2_i)}{\mu_2}$&\#&12346,...,23456&$\prod\limits_{i=1}^5F(\Delta^4_i)$\\\hline  
$0^0$&$A_0$&$A_0$&$A_0$&$A_0$&$A_0$&$5\times 1$&$\frac{1}{2^5}$&10$a_\pm$&1&$2^4$&$0^1$ or $0^2$&1\\

$0^0$&$A_0$&$A_0$&$A_0$&$A_0$&$A_0$&$5 \times\tau$&$\frac{\sqrt{2}^5}{2^5}$&10$b_\tau$&$\sqrt{2}^{10}$&1&$5\times1^0_\tau$&$\frac{1}{\sqrt{2}^5}$\\

$0^0$&$A_0$&$A_0$&$A_0$&$A_0$&$A_0$&$5 \times g$&$\frac{1}{2^5}$&10$b_{g0}$ or $b_{g1}$&1&$2^4$&$5 \times 1^0_g$&$1$\\

\hline
\multicolumn{9}{c}

$2 = \frac{1}{2}+1+\frac{1}{2}$

\\
\hline

$1_\tau^0$&$A_0$&$B_\tau^+$&$B_\tau^+$&$B_\tau^+$&$B_\tau^+$&$\tau,4\times1$&$\frac{\sqrt{2}}{2^5}$&4$b_\tau$,6$a_\pm$&$\sqrt{2}^{4}$&$2^3$&$4 \times 1^i_\tau,0$&$\frac{1}{\sqrt{2}^4}$\\  

$1_\tau^0$&$A_0$&$B_\tau^+$&$B_\tau^+$&$B_\tau^+$&$B_\tau^+$&$1,4\times\tau$&$\frac{\sqrt{2}^4}{2^5}$&$10b_\tau$&$\sqrt{2}^{10}$&1&$4\times2_\tau^+,1_\tau^0$&$\frac{\sqrt{2}}{2^4}$\\ 

$1_\tau^0$&$A_0$&$B_\tau^+$&$B_\tau^+$&$B_\tau^+$&$B_\tau^+$&$g$,$4 \times\tau$&$\frac{\sqrt{2}^4}{2^5}$&4$c$,6$b_\tau$&$\sqrt{2}^{6}$&$2$&$4\times 1_{gi\tau}^+$, $1_{\tau}^0$&$(\frac{\sqrt{2}}{2})^5$\\

$1_\tau^0$&$A_0$&$B_\tau^+$&$B_\tau^+$&$B_\tau^+$&$B_\tau^+$&$\tau$,$4 \times g$&$
\frac{\sqrt{2}}{2^5}$&4$c$,6$b_g^{\pm}$&1&$2^4$&$4\times2_{gi\tau}^+,1_{g0}^0$&$\frac{1}{2}$\\  
\hline
\multicolumn{10}{c}

$\sqrt{2} = \frac{\sqrt{2}}{4}+\frac{\sqrt{2}}{4}+\frac{\sqrt{2}}{4}+\frac{\sqrt{2}}{4}$\\

\hline
\end{tabular}
\caption{\label{tab:Example of Pachner Moves}Examples of Pachner moves of 20j-symbols}
\end{table}

\subsection{Conjecture}
The methods in this section seem to work in all dimensions based on a carefully study of the cobordism theory  \cite{Tho54,Wil66}. We conjecture that 

\begin{conjecture}
We define the $S^n$ functional $Z$ on stratified n-manifold $(S^n, \cup_i S_i)$ where $\{S_i\}_{i\in I}$ are PL surfaces in $S^n$ which may intersect transversely and be unoriented. It has the stratification:
\begin{enumerate}
    \item $M^n$ is $S^n$; 
    \item $M^{n-1}$ is the union of PL $(n-1)$ hyper surfaces;
    \item $M^{n-k}$ is the union of intersection $(n-k)$-manifolds of every $k$ hyper surfaces, $2\leq k\leq n$.
\end{enumerate}
The intersection of every $(n+1)$ hyper surfaces is empty.
Let $e_i$ be the Euler number of $S_i$. We define 
\begin{equation}
Z(S^n, \cup_i {S_i})=\prod_{i=1}^n 2^{1-\frac{e_i}{4}}.
\end{equation}
We conjecture that $S^n$ function $Z$ is complete positive and complete finite.
\end{conjecture}

\begin{acknowledgement}
The author would like to thank many people for helpful discussions on this topic, especially to 
Arthur Jaffe, Liang Kong, Maxim Konsevitch, Zhenghan Wang, Xiao-Gang Wen, Edward Witten, and colleagues Song Cheng, Jianfeng Lin, Fan Lu, Shuang Ming, Nicolai Reshetikhin, Yuze Ruan, Ningfeng Wang, Yilong Wang, Jinsong Wu, Shing-Tung Yau, Zishuo Zhao and Hao Zheng at Tsinghua University and BIMSA for fruitful discussions and comments; thank Yuze Ruan and Zishuo Zhao for pointing out gaps in early notes; thank Ningfeng Wang for the help on computing the 20j-symbols. The author would like to thank Guoliang Yu for the hospitality during the visit at Texas A$\&$M University.
The author is supported by BMSTC and ACZSP (Grant No. Z221100002722017), Beijing Natural Science Foundation Key Program (Grant No. Z220002), and by NKPs (Grant no. 2020YFA0713000).
\end{acknowledgement}

 \bibliographystyle{abbrv}
 \bibliography{TQFT}

\end{document}